[1994/12/01]
\documentclass[manuscript]{article}
\title{ {\bf Quantum Random \\  Self-Modifiable Computation} }
\author{Michael Stephen Fiske}
\date{May 3, 2019}

\usepackage{amsmath,amsthm,amsfonts,amssymb,mathtools,braket}
\usepackage{graphics}

\usepackage[textwidth=6.8in]{geometry}

\chardef\bslash=`\\ 

\newtheorem{thm}{Theorem}[section]
\newtheorem{cor}[thm]{Corollary}
\newtheorem{lem}[thm]{Lemma}

\newtheorem{ax}{Axiom}

\theoremstyle{definition}
\newtheorem{defn}{Definition}[section]

\theoremstyle{remark}
\newtheorem{rem}{Remark}[section]

\newtheorem{example}{Example}
\newtheorem{Machine}{Ex-Machine}
\newtheorem{Instructions}{Machine Instructions}

\newcommand{\eval}[2][\right]{\relax
  \ifx#1\right\relax \left.\fi#2#1\rvert}

\begin{document}

\maketitle

\begin{abstract}
            Among the fundamental questions in computer science, at least two have a deep impact on mathematics.    
            What can computation compute?    
            How many steps does a computation require to solve an instance of the 3-SAT problem?
            Our work addresses the first question, by introducing a new model called the {\it ex-machine}.
            The ex-machine executes Turing machine instructions and two special types of instructions.
            {\it Quantum random instructions}  are physically realizable with a quantum random number generator.  
            {\it Meta instructions} can add new states and add new instructions to the 
            ex-machine.    
            A countable set of ex-machines is constructed,
            each with a finite number of states and instructions; each ex-machine 
            can compute a Turing incomputable language, whenever the quantum randomness measurements behave like 
            unbiased Bernoulli trials.  In 1936, Alan Turing posed the halting problem for Turing machines and 
            proved that this problem is unsolvable for Turing machines.   
            Consider an  enumeration $\mathcal{E}_a(i) = (\mathfrak{M}_i, T_i)$ of all Turing machines  
            $\mathfrak{M}_i$ and initial tapes $T_i$.  
            Does there exist an ex-machine $\mathfrak{X}$ that has at least one evolutionary path 
            $\mathfrak{X}$ $\rightarrow$ $\mathfrak{X}_1$ $\rightarrow$ 
            $\mathfrak{X}_2$ $\rightarrow$  $\dots$ $\rightarrow$  $\mathfrak{X}_m$, so at the 
            $m$th 
            stage ex-machine $\mathfrak{X}_m$ can correctly determine for $0 \le i \le m$  whether $\mathfrak{M}_i$'s 
            execution on tape $T_i$ eventually halts?  We demonstrate an ex-machine $\mathfrak{Q}(x)$ 
            that has one such evolutionary path.    The existence of this evolutionary path suggests that  
            David Hilbert was not misguided to propose in 1900 that mathematicians search for finite processes to 
            help construct mathematical proofs \cite{hilbert}.  
            Our refinement is that we cannot use a fixed computer program that 
            behaves according to a fixed set of mechanical rules.  We must pursue methods that exploit  
            randomness and self-modification so that the complexity of the program can increase as it computes.
\end{abstract}

\tableofcontents


\section{Introduction}

Consider two fundamental questions in computer science:

\begin{enumerate}

\item {  
            What can a computing machine compute?    
       }    

\medskip

\item {  
           How many computational steps does a computational machine require to solve an instance of the 3-SAT problem?   
           The 3-SAT problem \cite{cook_3SAT} is the basis for the  famous $P {\overset{?}= NP}$ problem \cite{cook_p_vs_np}.   
      }

\end{enumerate}

\noindent  These two questions are usually studied with    
the assumption that the Turing machine (TM) \cite{turing36} is the standard model of computation 
\cite{davis,garey,lewis,minsky,rogers,sipser}.  

      We introduce a new  model, called the {\it ex-machine}, and reexamine the first question. 
      The ex-machine model bifurcates the first question into two questions.   {\it What is computation?} 
      {\it What can computation compute?}  An ex-machine computation adds two special types of 
      instructions to the Turing machine instructions. 
      The name ex-machine --- derived from the latin {\it extra machinam} ---  
      was chosen because ex-machine computation generates new dynamical behaviors that one may no longer recognize as a machine.

       The {\it meta instruction} is one type of special instruction. 
       When an ex-machine executes a meta instruction, 
       the meta instruction can add new states and add new instructions or replace instructions.   
       The meta instruction enables the 
       complexity \cite{shannon,schmitt} of an ex-machine to increase during its execution,
       unlike a typical machine (e.g., the inclined plane, lever, pulley, wedge, 
       wheel and axle, Archimedean screw, Galilean telescope or bicycle).

      The {\it quantum random instruction} is the other special instruction.  It can be physically realized with a 
      quantum random number generator \cite{gabriel,ma,pironio,rohe,svozil_3_criteria,stefanov,stipcevic,wahl}.  
      Due to the {\it quantum random instructions}, the execution 
      behavior of two ex-machines may be distinct, even though the two ex-machines start their execution with the 
      same input on the tape, the same instructions, the same initial states, and so on.  Two distinct identical 
      ex-machines may exhibit different execution behaviors even when started with identical initial conditions.  
      When this property of the quantum random instructions is combined with the appropriate use of meta instructions, 
      two identical machines with the same initial conditions can evolve to two different ex-machines as the 
      execution of each respective machine proceeds.

      Some of the ex-machine programs provided here compute beyond the Turing barrier.  
      A countable set of ex-machines are explicitly defined.  Every  one of these 
      ex-machines can evolve to compute a Turing incomputable language with probability measure 1, 
      whenever the quantum random measurements (trials) behave like unbiased Bernoulli trials.  
      (A Turing machine cannot compute a Turing incomputable language.)

      In 1936, Alan Turing posed the halting problem and proved that the halting problem for Turing machines 
      is unsolvable by a Turing machine  \cite{davis,minsky,rogers,turing36}.  
      Consider the ex-machine halting problem:  Given an enumeration  
      $\mathcal{E}_a(i) = (\mathfrak{M}_i, T_i)$ of all Turing machines  $\mathfrak{M}_i$ and initial tapes $T_i$, 
      each finitely bounded and containing only blank symbols outside the bounds, 
      does there exist an ex-machine $\mathfrak{X}$ that has at least one evolutionary path
      $\mathfrak{X}$ $\rightarrow$
      $\mathfrak{X}_1$ $\rightarrow$ 
      $\mathfrak{X}_2$ $\rightarrow$  $\dots$ $\rightarrow$  $\mathfrak{X}_m$, so at stage $m$, the  
      ex-machine $\mathfrak{X}_m$ can correctly determine for $0 \le i \le m$  whether $\mathfrak{M}_i$'s 
      execution on tape $T_i$ eventually halts?  We demonstrate an 
      ex-machine $\mathfrak{Q}(x)$ that has one such evolutionary path.  
      At stage $m$, the self-modifying ex-machine's evolutionary path  
      $\mathfrak{Q}(h_{\mathcal{E}_a}(0)$   \hskip 0.2pc $x)$  $\rightarrow$  
      $\mathfrak{Q}(h_{\mathcal{E}_a}(0)$ $h_{\mathcal{E}_a}(1)$  \hskip 0.2pc $x)$  $\rightarrow$  $\dots$ 
      $\mathfrak{Q}(h_{\mathcal{E}_a}(0)$ $h_{\mathcal{E}_a}(1)$ $\dots$ $h_{\mathcal{E}_a}(m)$  \hskip 0.2pc $x)$ 
      has used a finite amount of computational resources and measured a finite amount of quantum randomness.
 

      Consider the Goldbach Conjecture \cite{goldbach} and the 
      Riemann Hypothesis \cite{riemann}, which are both famous, unsolved math problems.  
      Each of these problems can be expressed as an instance of Turing's halting problem with a 
      particular Turing machine.  
      (See machine instructions \ref{ins:goldbach} and \cite{turing_oracle}.) 
      A large scale, physical realization of an ex-machine and further research might present an opportunity 
      to study these mathematical problems and other difficult problems with new computational 
      and conceptual tools.


      \subsection{Related Work -- Computation}
            
      The rest of the introduction discusses some related results on computation using 
      quantum randomness, and the theory of quantum randomness.  
      Some related work on computation is in  \cite{fiske_tic} and \cite{fiske_qaem}.  
      In \cite{fiske_tic}, a parallel computational machine, 
      called the active element machine, uses its meta commands and quantum randomness to construct a 
      computational procedure that behaves like a quantum black box.  Using quantum randomness as 
      a source of unpredictability and the meta commands to self-modify the active element machine, 
      this procedure emulates a universal Turing machine so that an outside observer is unable to 
      comprehend what Turing machine instructions are being executed by the emulation of 
      the universal Turing machine.

      In \cite{fiske_qaem}, based on a Turing machine's states and alphabet symbols, 
      a transformation $\phi$ was defined from the Turing machine's instructions 
      to a finite number of affine functions in the two dimensional plane $\mathbb{Q} \times \mathbb{Q}$, 
      where  $\mathbb{Q}$ is the rational numbers.  Now for the details:  let states 
      $Q = \{q_1, \dots, q_{|Q|} \}$, alphabet $A = \{a_1, \dots, a_{|A|}\}$, a halt state $h$ that is not in $Q$,  
      and program $\eta : Q \times A \rightarrow Q \cup \{h\} \times A \times \{-1, +1\}$ be a Turing 
      machine.  This next part defines a one-to-one mapping $\phi$ from Turing program $\eta$ to a 
      finite set of affine functions, whose 
      domain is a bounded subset of $\mathbb{Q} \times \mathbb{Q}$.  Set $B = |A| + |Q| + 1$.  Define 
      symbol value function $\nu: \{h\} \cup Q \cup A \rightarrow \mathbb{N}$ \hskip 0.1pc as \hskip 0.1pc
      $\nu(h) = 0$, $\nu(a_i) = i$ and $\nu(q_i) = i + |A|$.

      $T_k$ is the alphabet symbol in the $k$th tape square.    
      $\phi$ maps right computational step $\eta(q, T_k) = (r, \alpha, +1)$ to the affine function 
      $f(x, y) = \Big{(} Bx - B^2 \nu(T_k),$ \hskip 0.3pc $\frac{1}{B}y + B \nu(r) + \nu(\alpha) - \nu(q) \Big{)}$.  During 
      this computational step, state $q$ moves to state $r$.  Alphabet symbol $\alpha$ replaces $T_k$ on tape square $k$, 
      and the tape head moves to tape square $k+1$, one square to the right.   

      Similarly, $\phi$ maps left computational step $\eta(q, T_k) = (r, \alpha, -1)$ to the affine function 
      $g(x, y) = \Big{(} \frac{1}{B}x + B \nu(T_{k-1}) + \nu(\alpha) - \nu(T_k),$ \hskip 0.4 pc 
      $By + B \nu(r) - B^2 \nu(q) - B \nu(T_{k-1}) \Big{)}$. 
      $\phi$ maps machine configuration  $(q, k, T) \in Q \times \mathbb{Z} \times A^{\mathbb{Z}}$ to the point 
      $\phi(q, k, T) = \Big{(} {\overset{\infty}{\underset{j=-1}\sum}} \nu(T_{k+j+1})  B^{-j},$  
      \hskip 0.4pc 
      $B \nu(q) + {\overset{\infty}{\underset{j=0}\sum}} \nu(T_{k-j-1}) B^{-j} \Big{)}$ in the 
      $\mathbb{Q} \times \mathbb{Q}$ plane.  Point $\phi(q, k, T)$ is in $\mathbb{Q} \times \mathbb{Q}$ 
      because only a finite number of tape squares contain non-blank symbols, so the tail of each 
      infinite sum is a geometric series.

      Each affine function's domain is a subset of some unit square 
      $\{ (x, y) \in \mathbb{Q} \times \mathbb{Q}:  m \le x \le m+1 $ \verb|and| $ n \le y \le n+1 \}$,  
      where $m$ and $n$ are integers.  Via the $\phi$ transformation, a finitely bounded initial tape and initial 
      state of the Turing machine are mapped to an initial point with rational cooordinates in one of the unit squares.  
      Hence, {\it $\phi$ transforms Turing's halting problem to the following dynamical systems problem.}
      If machine configuration $(q,k,T)$ halts after $n$ computational steps, then the orbit of $\phi(q,k,T)$ 
      exits one of the unit squares on the $n$th iteration.  If machine configuration $(r,j,S)$ is immortal 
      (i.e., never halts), then the orbit of $\phi(r,j,S)$ remains in these finite number of unit squares forever. 

      Dynamical system $\frac{dx}{dt} = F(x, y)$ and $\frac{dy}{dt} = G(x, y)$ is {\it autonomous} if the 
      independent variable $t$ does not appear in $F$ and $G$.  A discrete, autonomous dynamical system 
      is comprised of a function $f: X \rightarrow X$, where $X$ is a topological space and the orbits 
      $\mathfrak{O}(f, p) = \{ f^k(p): p \in X$ \verb|and| $k \in \mathbb{N} \}$ are studied.

       Consider the following augmentation of the discrete, autonomous dynamical system $(f, X)$.  
       After the 1st iteration,  $f$ is perturbed to $f_1$ where $f \ne f_1$ and after the
       second iteration $f_1$ is perturbed to $f_2$ so that $f_2 \ne f_1$ and $f_2, \ne f$ and so on 
       where $f_i \ne f_j$ for all $i \ne j$.  
       Then the dynamical system $(f_1, f_2, \dots f_k \dots, X)$ is a 
       {\it discrete, non-autonomous} dynamical system \cite{fiske_thesis}.

      For a particular Turing machine, set $X$ equal to the union of all the unit squares induced by $\phi$
      and define $f$ based on the finite number of affine functions, resulting from the $\phi$ transformation.   
      In terms of dynamical systems theory, the $\phi$ transformation shows that  
      {\it each Turing machine is a discrete, autonomous dynamical system}.  
      In \cite{fiske_qaem}, we stated that an active element 
      machine using quantum randomness was a non-autonomous dynamical system capable of generating non-Turing 
      computational behaviors;  however, no new specific machines exhibiting novel behaviors were provided,  
      except for a reference to procedure 2 in \cite{fiske_tic}.  In this sense, our research is a 
      continuation of \cite{fiske_tic,fiske_qaem}, but arguably provides a more transparent computational model 
      for studying what can be computed with randomness and self-modification.  

      \subsection{Related Work -- Quantum Randomness}

      Some other related work pertains to the theory of quantum randomness.  
      The classic EPR paper \cite{epr} presented a paradox that led Einstein, Podolsky and Rosen to conclude  
      that quantum mechanics is an incomplete theory and should be supplemented with additional variables. 
      They believed that the statistical predictions of quantum mechanics were correct, but only as a 
      consequence of the statistical distributions of these hidden variables.  Moreover, they believed 
      that the specification of these hidden variables could predetermine the result of measuring 
      any observable of the system.

      Due to an ambiguity in the EPR argument, Bohr \cite{bohr} explained that no paradox or contradiction 
      can be derived from the assumption that Schrodinger's wave function \cite{schrodinger} 
      contains a complete description of physical reality.  Namely, in the quantum theory, 
      it is impossible to control the interaction between the object being observed and the 
      measurement apparatus.   Per Heisenberg's uncertainty principle 
      \cite{heisenberg_1927}, momentum is transferred between them during position measurements, and the object is 
      displaced during momentum measurements.   Based on the link between the wave function and the probability 
      amplitude, first proposed by Born \cite{born_1926},  Bohr's response set the stage for the problem of 
      hidden variables and the development of quantum mechanics as a statistical scientific theory.  


      In \cite{bohm_aharanov}, Bohm and Aharanov advocated a Stern-Gerlach magnet \cite{gerlach1,gerlach2,gerlach3} 
      example to address the hidden variables problem.  Using a gedankenexperiment \cite{bohm_gedanken} of Bohm, 
      Bell showed that no local hidden variable theory can reproduce all of the statistical 
      predictions of quantum mechanics and maintain local realism \cite{bell}.  
      Clauser, Horne, Shimony and Holt derived a new form of Bell's inequality \cite{clauser}, 
      called the CHSH inequality, along with a proposed physically-realizable experiment 
      to test their inequality.

      In \cite{shalm}, their experiment tests the CHSH inequality.  Using entangled photon pairs, 
      their experiment found a loophole-free \cite{larsson} violation of local realism.   
      They estimated the degree to which a local realistic system could predict their measurement choices, 
      and obtained a smallest adjusted $p$-value equal to $2.3 \times 10^{-7}$.  Hence, they rejected the hypothesis 
      that local realism governed their experiment.  Recently, a quantum randomness expander has been 
      constructed, based on the CHSH inequality \cite{pironio}.

      By taking into account the algebraic structure of quantum mechanical observables, 
      Kochen and Specker \cite{kochen} provided a proof for the nonexistence of hidden variables. 
      In \cite{svozil_3_criteria}, Svozil proposed three criteria for building 
      quantum random number generators based on beam splitters:  
      (A) Have three or more mutually exclusive outcomes correspond to the invocation of Hilbert spaces with dimension at least 3; 
      (B) Use pure states in conjugated bases for preparation and detection;  
      (C) Use entangled singlet (unique) states to eliminate bias.

      By extending the theory of Kochen and Specker, Calude and Svozil developed an initial Turing incomputable theory of 
      quantum randomness \cite{calude_2008} --- applicable to beam splitters --- that has been recently advanced 
      further by  Abbott, Calude, and Svozil \cite{calude_qr2012,calude_qr2014,calude_qr2015}.
      A more comprehensive summary of their work  will be provided in section 3. 


\section{The Ex-Machine}\label{sect:qr_self_modify_xmachine}

We introduce a {\it quantum random, self-modifiable  machine} that adds two 
special types of instructions to the Turing machine \cite{turing36}.  Before the 
quantum random and meta instructions are defined, we present some preliminary
notation, the standard instructions, and a Collatz machine example. 
 


$\mathbb{Z}$ denotes the integers. $\mathbb{N}$ and $\mathbb{N}^+$ are the 
non-negative and positive integers, respectively.  
The finite set $Q = \{0, 1, 2, \dots, n-1 \} \subset \mathbb{N}$
 represents the ex-machine states.   This representation of the  
 ex-machine states helps specify how new states are added to $Q$ when 
 a meta instruction is executed.   Let  $\mathfrak{A} = \{ a_1, \dots, a_n \}$, 
 where each $a_i$ represents a distinct symbol.  
 The set $A = \{$\verb|0, 1, #|$\} \cup \mathfrak{A}$ consists of alphabet 
 (tape) symbols, where \verb|#| is the blank symbol and $\{0, 1,$ \verb|#|$\} \cap \mathfrak{A}$ 
 is the empty set.  In some  ex-machines, $A = \{$\verb|0, 1, #, Y, N, a|$\}$, 
 where $a_1 =$ \verb|Y|, $a_2 =$ \verb|N|, $a_3 =$ \verb|a|. In some ex-machines,  $A = \{$\verb|0, 1, #|$\}$, 
where  $\mathfrak{A}$ is the empty set.   The alphabet symbols are read from and written 
on the tape.   The ex-machine tape $T$ is a function $T : \mathbb{Z} \rightarrow A$ with an initial  
condition:  before the ex-machine starts executing, there exists an $N > 0$ so that $T(k) = $ \verb|#| when $|k| > N$.
In other words, before the ex-machine starts executing, all tape squares contain blank symbols, 
except for a finite number of tape squares.  When this initial condition holds for tape $T$, we say that 
tape $T$ is {\it finitely bounded}.

\subsection{Standard Instructions}

\begin{defn}\label{defn:standard_instruction}  \hskip 1pc  { Execution of Standard Instructions}

\smallskip

\noindent  The standard ex-machine instructions  $\mathcal{S}$ satisfy   
$\mathcal{S} \subset  Q \times A \times Q \times A \times \{-1, 0, 1 \}$ and 
a uniqueness condition:
If  $(q_1, \alpha_1, r_1, a_1, y_1)$ $\in$ $\mathcal{S}$ 
    \verb| and |  
    $(q_2, \alpha_2, r_2, a_2, y_2)$ $\in$ $\mathcal{S}$ 
    \verb| and | 
    $(q_1, \alpha_1, r_1, a_1, y_1) \ne (q_2, \alpha_2, r_2, a_2, y_2)$, 
    then $q_1 \ne q_2$\verb| or |$\alpha_1 \ne \alpha_2$.   
A standard instruction $I = (q, a, r, \alpha, y)$ is similar to a Turing machine tuple \cite{davis,post,turing36}. 
When the ex-machine is in state $q$ and the tape head is scanning alphabet symbol $a = T(k)$ at tape square $k$,  
instruction $I$ is executed as follows:

  \begin{itemize}

   \item[$\bullet$]  {  
                         The ex-machine state moves from state $q$ to state $r$.        
                     }

   \medskip

   \item[$\bullet$]  {  
                         The ex-machine replaces alphabet symbol $a$ with alphabet symbol $\alpha$ so that $T(k) = \alpha$.  
                         The rest of the tape remains unchanged.                     
                     }

   \medskip

   \item[$\bullet$] {   
                         If $y = -1$, the ex-machine moves its tape head one square to the left on the 
                         tape and is subsequently scanning the alphabet symbol $T(k-1)$ in tape square $k-1$.          
                    }

   \medskip

   \item[$\bullet$] {   
                         If $y = +1$, the ex-machine moves its tape head one square to the right on the tape 
                         and is subsequently scanning  the alphabet symbol $T(k+1)$ in tape square $k+1$.     
                    }
  
   \medskip

   \item[$\bullet$] {   
                        If $y = 0$, the ex-machine does not moves its tape head and is subsequently scanning 
                        the alphabet symbol $T(k) = \alpha$ in tape square $k$.    
                    }

   \end{itemize}

\end{defn}

\begin{rem}

A Turing machine \cite{turing36} has a finite set of states $Q$, a finite alphabet $A$, a finitely bounded tape,   
and a finite set of standard ex-machine instructions that are executed according to 
definition \ref{defn:standard_instruction}.   In other words, an ex-machine that uses only standard instructions 
is computationally equivalent to a Turing machine.    Hence, an ex-machine with only standard instructions 
will be called a standard machine or a Turing machine. 

\end{rem}



 The Collatz conjecture has an interesting relationship to Turing's halting problem, 
 which will be discussed further in section \ref{sect:two_research_problems}. 
 Furthermore, there is a generalization of the Collatz function that is unsolvable for a 
 standard machine \cite{conway_72}.   


\begin{defn}\label{defn:collatz_function}   \hskip 1pc {Collatz Conjecture}

\smallskip 

\noindent Define the Collatz function $f: \mathbb{N}^+ \rightarrow \mathbb{N}^+$, 
where $f(n) = \frac{n}{2}$ when $n$ is even and $f(n) = 3n+1$ when $n$ is odd.  
Zero iterations of $f$ is $f^0(n) = n$.  
$k$ iterations of $f$ is represented as $f^k(n)$.
The orbit of $n$ with respect to $f$ is 
$\mathcal{O}(f, n) = \{ f^k(n): k \in \mathbb{N} \}$.   
Observe that $f(5) = 16$, $f^2(5) = 8$, $f^3(5) = 4$, $f^4(5) = 2$, $f^5(5) = 1$, 
so $\mathcal{O}(f, 5) = \{5, 16, 8, 4, 2, 1\}$.     
The Collatz conjecture states that for any positive integer $n$,  $\mathcal{O}(f, n)$ contains $1$. 

\end{defn}

We specify a Turing machine that for each $n$  computes the orbit  $\mathcal{O}(f, n)$.
The standard machine halts if the orbit  $\mathcal{O}(f, n)$ contains 1.  
Set $A = \{$\verb|0|, \verb|1|, \verb|#|, \verb|E|$\}$.  
Set $Q = \{$\verb|a|, \verb|b|, \verb|c|, \verb|d|, \verb|e|, \verb|f|, \verb|g|, 
\verb|h|, \verb|i|, \verb|j|, \verb|k|, \verb|l|, \verb|m|, \verb|n|, \verb|p|, \verb|q| $\}$ 
where \verb|a| $= 0 $, \verb|b| $= 1$, \verb|c| $= 2$, $\dots$, \verb|n| $= 13$, \verb|p| $= 14$, and \verb|q| $= 15$.

Machine instructions \ref{ins:collatz_machine} shows a list of standard instructions that 
compute $\mathcal{O}(f, n)$.  The initial tape is  
{\small  \verb|# #| {\verb|1|}$^n$\verb|#|}, where it is 
understood that the remaining tape squares, beyond the leftmost \verb|#| and rightmost \verb|#|, 
contain only blank symbols. The space means the tape head is scanning the \hskip 0.1pc {\small \verb|#| } \hskip 0.05pc 
adjacent to the leftmost \verb|1|.    The initial state is \verb|q|.

\begin{Instructions}\label{ins:collatz_machine}  \hskip 1pc  {Collatz Machine}

{  \scriptsize  
\begin{verbatim} 
    ;;  Comments follow two semicolons. 
    (q,  #,  a,  #,  1) 
    (q,  0,  p,  0,  1) 
    (q,  1,  p,  1,  1) 
                
    (a,  #,  p,  #,  1) 
    (a,  0,  p,  0,  1) 
    (a,  1,  b,  1,  1) 
    
    (b,  #,  h,  #,  -1)   ;; Valid halt # 1#. The Collatz orbit reached 1. 
    (b,  0,  p,  0,  1) 
    (b,  1,  c,  1,  1) 

    (c,  #,  e,  #,  -1) 
    (c,  0,  p,  0,  1) 
    (c,  1,  d,  1,  1) 

    (d,  #,  k,  #,  -1) 
    (d,  0,  p,  0,  1) 
    (d,  1,  c,  1,  1) 

    ;;  n / 2  computation
    (e,  #,  g,  #,  1) 
    (e,  0,  p,  0,  1) 
    (e,  1,  f,  0,  -1) 

    (f,  #,  g,  #,  1) 
    (f,  0,  p,  0,  1) 
    (f,  1,  f,  1,  -1) 

    (g,  #,  j,  #,  -1) 
    (g,  0,  g,  1,  1) 
    (g,  1,  i,  #,  1) 

    (i,  #,  p,  #,  1) 
    (i,  0,  e,  0,  -1) 
    (i,  1,  i,  1,  1) 

    (j,  #,  a,  #,  1) 
    (j,  0,  p,  0,  1) 
    (j,  1,  j,  1,  -1) 

    ;;  3n + 1 computation
    (k,  #,  n,  #,  1) 
    (k,  0,  k,  0,  -1) 
    (k,  1,  l,  0,  1) 

    (l,  #,  m,  0,  1) 
    (l,  0,  l,  0,  1) 
    (l,  1,  p,  1,  1) 

    (m,  #,  k,  0,  -1) 
    (m,  0,  p,  0,  1) 
    (m,  1,  p,  1,  1) 

    ;;  Start  n / 2  computation
    (n,  #,  f,  0,  -1)   
    (n,  0,  n,  1,  1) 
    (n,  1,  p,  1,  1) 

    ;;  HALT with ERROR.  Alphabet symbol E represents an error.
    (p,  #,  h,  E,  0) 
    (p,  0,  h,  E,  0) 
    (p,  1,  h,  E,  0) 
\end{verbatim}
}

\end{Instructions}

\noindent  With input 
{\small  \verb|# #|{\verb|1|}$^n$\verb|#|}, 
the execution of the Collatz machine halts (i.e., moves to the halting state $h$) 
if the orbit $\mathcal{O}(f, n)$ reaches 1.  Below shows the Collatz machine executing 
the first ten instructions with initial tape 
{\small   \verb|# #11111#|  } 
and initial state \verb|q|.  
Each row shows the current tape and machine state after the instruction in that 
row has been executed. The complete execution of the Collatz machine is shown in the 
appendix \ref{appendix:collatz}.   It computes $\mathcal{O}(f, 5)$.  

\smallskip 


{  \scriptsize  
  \begin{verbatim} 
STATE     TAPE              TAPE HEAD     INSTRUCTION EXECUTED          COMMENT
a         ## 11111#####         1         (q, #, a, #, 1)   
b         ##1 1111#####         2         (a, 1, b, 1, 1)   
c         ##11 111#####         3         (b, 1, c, 1, 1)   
d         ##111 11#####         4         (c, 1, d, 1, 1)   
c         ##1111 1#####         5         (d, 1, c, 1, 1)     
d         ##11111 #####         6         (c, 1, d, 1, 1)   
k         ##1111 1#####         5         (d, #, k, #, -1)              Compute 3*5 + 1
l         ##11110 #####         6         (k, 1, l, 0, 1)              
m         ##111100 ####         7         (l, #, m, 0, 1)   
k         ##11110 00###         6         (m, #, k, 0, -1) 
  \end{verbatim}
}



\subsection{Quantum Random Instructions}



 Repeated independent trials are called {\it quantum random Bernoulli trials} \cite{feller_vol1} 
 if there are only two possible outcomes for each trial (i.e., quantum random measurement) 
 and the probability of each outcome remains constant for all trials.  
 {\it Unbiased} means the probability of both outcomes is the same.  
 Below are the formal definitions.

 \begin{ax}\label{axiom_qr1}  \hskip 1pc {  Unbiased Trials.} 

 \smallskip 

 \noindent Consider the bit sequence $(x_1 x_2 \dots )$ in the infinite product space $\{0, 1\}^{\mathbb{N}}$.  
 A single outcome $x_i$ of a bit sequence $(x_1 x_2 \dots )$ generated by quantum randomness  is unbiased. 
 The probability of measuring a 0 or a 1 are equal:     $P(x_i = 1) = P(x_i = 0) = \frac{1}{2}$. 

 \end{ax}

\begin{ax}\label{axiom_qr2}  \hskip 1pc {  Stochastic Independence.} 

\smallskip  

\noindent History has no effect on the next quantum random measurement.  Each outcome $x_i$ is independent of the history.  
No correlation exists between previous or future outcomes.   This is expressed in terms of the conditional probabilities:  
$P(x_i = 1$ $|$ $x_1 = b_1,$ $\dots,$ $x_{i-1} =  b_{i-1}) =  \frac{1}{2}$ and 
$P(x_i = 0$ $|$ $x_1 = b_1,$ $\dots,$ $x_{i-1} =   b_{i-1}) =  \frac{1}{2}$ for
each $b_i \in \{0, 1\}$.

\end{ax}

\noindent  In order to not detract from the formal description of the ex-machine, section 
\ref{sect:quantum_randomness} provides a physical basis for the axioms and a 
discussion of quantum randomness.

The quantum random instructions $\mathcal{R}$ are subsets of  
$Q \times A \times Q \times \{-1, 0, 1\}$ $=$ $\big{\{} (q, a, r, y):$  \hskip 0.05pc
$q,$  \hskip 0.05pc $r$\verb| are in| \hskip 0.1pc $Q$  \hskip 0.1pc
\verb|and |$a$\verb| in |$A$\verb| and| \hskip 0.1pc $y$\verb| in |$\{-1, 0, 1 \}$ $\big{\}}$ that satisfy 
a uniqueness condition defined below.

\begin{defn}\label{defn:qr_instruction}  \hskip 1pc  { Execution of Quantum Random Instructions}

\smallskip

\noindent  The quantum random instructions $\mathcal{R}$ satisfy 
$\mathcal{R} \subset Q \times A \times Q \times \{-1, 0, 1\}$  
and the following uniqueness condition:
If  $(q_1, \alpha_1, r_1, y_1)$ $\in$ $\mathcal{R}$ \verb| and |  
    $(q_2, \alpha_2, r_2, y_2)$ $\in$ $\mathcal{R}$ \verb| and |
    $(q_1, \alpha_1, r_1, y_1) \ne (q_2, \alpha_2, r_2, y_2)$, 
    then $q_1 \ne q_2$\verb| or |$\alpha_1 \ne \alpha_2$.  
When the tape head is scanning alphabet symbol $a$  
and the ex-machine is in state $q$, the quantum random instruction $(q, a, r, y)$ executes as follows:

\begin{itemize}

   \item[$\bullet$]  {  
                         The ex-machine measures a quantum random source that returns a random bit $b$ $\in$ $\{0, 1\}$.
                         (It is assumed that the quantum measurements satisfy unbiased Bernoulli trial  
                          axioms \ref{axiom_qr1} and \ref{axiom_qr2}.)
                     }

   \medskip

   \item[$\bullet$]  {    
                          On the tape, alphabet symbol $a$ is replaced with random bit $b$.  

                          \noindent (This is why $A$ always contains both symbols \verb|0| and \verb|1|.)
                     }  

   \medskip

   \item[$\bullet$]  {    The ex-machine state changes to state $r$.  }

   \medskip

   \item[$\bullet$]  {   
                          The ex-machine moves its tape head left if $y = -1$, right if $y = +1$, 
                          or the tape head does not move if  $y = 0$.  
                     }

\end{itemize}

\end{defn}


\smallskip 

Machine instructions \ref{ins:qr_walk_machine} lists a random walk machine that has only 
standard instructions and quantum random instructions. 
Alphabet {\small $A = \{$\verb|0|, \verb|1|, \verb|#|, \verb|E|$\}$.  
The states are  $Q = \{$\verb|0|, \verb|1|, \verb|2|, \verb|3|, 
\verb|4|, \verb|5|, \verb|6|, \verb|h|$\}$, }
where the halting state 
{\small \verb|h| $= 7$}. 
 A valid initial tape contains only blank symbols; that is, 
{\small \verb|# ##|}.   The valid initial state is {\small  \verb|0|}.

There are three quantum random instructions: {\footnotesize \verb|(0, #, 0, 0)|, 
\verb| (1, #, 1, 0)|}  and  {\footnotesize \verb| (4, #, 4, 0)|}.  The random instruction 
 {\footnotesize  \verb|(0, #, 0, 0)| } is executed first.   
If the quantum random source measures a 1, the machine jumps to state {\small \verb|4| }
and the tape head moves to the right of tape square 0.  If the quantum random source 
measures a 0, the machine jumps to state {\small \verb|1| } and the tape head moves to the 
left of tape square 0.  Instructions containing alphabet 
symbol {\small \verb|E| } provide error checking for an invalid initial tape or initial state;  
in this case, the machine halts with an error.

\begin{Instructions}\label{ins:qr_walk_machine}  \hskip 1pc  {Random Walk}

{  \scriptsize  
\begin{verbatim} 
    ;;  Comments follow two semicolons. 
    (0, #, 0, 0)
    (0, 0, 1, 0, -1)
    (0, 1, 4, 1, 1)

    ;;  Continue random walk to the left of tape square 0  
    (1, #, 1, 0)       
    (1, 0, 1, 0, -1)
    (1, 1, 2, #, 1)

    (2, 0, 3, #, 1)
    (2, #, h, E, 0)
    (2, 1, h, E, 0)

    ;; Go back to state 0.  Numbers of random 0's = number of random 1's. 
    (3, #, 0, #, -1)   

    ;; Go back to state 1.  Numbers of random 0's > number of random 1's. 
    (3, 0, 1, 0, -1)   
    (3, 1, h, E, 0)

    ;;  Continue random walk to the right of tape square 0  
    (4, #, 4, 0)         
    (4, 1, 4, 1, 1)
    (4, 0, 5, #, -1)

    (5, 1, 6, #, -1)
    (5, #, h, E, 0)
    (5, 0, h, E, 0)
    
    ;; Go back to state 0.  Numbers of random 0's = number of random 1's. 
    (6, #, 0, #, 1)    

    ;; Go back to state 4.  Numbers of random 1's > number of random 0's. 
    (6, 1, 4, 1, 1)     
    
    (6, 0, h, E, 0)
\end{verbatim}
}
\end{Instructions}

Below are 31 computational steps of the ex-machine's first execution.  
This random walk machine  never halts when the initial tape is blank 
and the initial state is \verb|0|.  The first quantum random instruction executed 
is  {\small \verb|(0, #, 0, 0)|}.  The quantum random source measured a \verb|0|, so the 
execution of this instruction is shown as  {\small  \verb|(0, #, 0, 0_qr, 0)| }.  
The second quantum random instruction executed is 
{\small  \verb|(1, #, 1, 0)| }.   
The quantum random source measured a \verb|1|, so the execution of instruction  
{\small  \verb|(1, #, 1, 0)|  }
is shown as 
{\small  \verb|(1, #, 1, 1_qr, 0)|  }.

\medskip 

\bigskip 

\noindent {\small  {\bf 1st Execution of Random Walk Machine.  Computational Steps 1-31.}  }

{  \scriptsize  
  \begin{verbatim} 
  STATE   TAPE             TAPE HEAD         INSTRUCTION EXECUTED               
   0      ##### 0##            0             (0, #,  0, 0_qr, 0)   
   1      #### #0##           -1             (0, 0,  1, 0, -1)   
   1      #### 10##           -1             (1, #,  1, 1_qr, 0)   
   2      ##### 0##            0             (1, 1,  2, #, 1)   
   3      ###### ##            1             (2, 0,  3, #, 1)   
   0      ##### ###            0             (3, #,  0, #, -1)   
   0      ##### 0##            0             (0, #,  0, 0_qr, 0)   
   1      #### #0##           -1             (0, 0,  1, 0, -1)   
   1      #### 00##           -1             (1, #,  1, 0_qr, 0)   
   1      ### #00##           -2             (1, 0,  1, 0, -1)   
   1      ### 000##           -2             (1, #,  1, 0_qr, 0)   
   1      ## #000##           -3             (1, 0,  1, 0, -1)   
   1      ## 1000##           -3             (1, #,  1, 1_qr, 0)   
   2      ### 000##           -2             (1, 1,  2, #, 1)   
   3      #### 00##           -1             (2, 0,  3, #, 1)   
   1      ### #00##           -2             (3, 0,  1, 0, -1)   
   1      ### 100##           -2             (1, #,  1, 1_qr, 0)   
   2      #### 00##           -1             (1, 1,  2, #, 1)   
   3      ##### 0##            0             (2, 0,  3, #, 1)   
   1      #### #0##           -1             (3, 0,  1, 0, -1)   
   1      #### 10##           -1             (1, #,  1, 1_qr, 0)   
   2      ##### 0##            0             (1, 1,  2, #, 1)   
   3      ###### ##            1             (2, 0,  3, #, 1)   
   0      ##### ###            0             (3, #,  0, #, -1)   
   0      ##### 0##            0             (0, #,  0, 0_qr, 0)   
   1      #### #0##           -1             (0, 0,  1, 0, -1)   
   1      #### 00##           -1             (1, #,  1, 0_qr, 0)   
   1      ### #00##           -2             (1, 0,  1, 0, -1)   
   1      ### 000##           -2             (1, #,  1, 0_qr, 0)   
   1      ## #000##           -3             (1, 0,  1, 0, -1)   
   1      ## 1000##           -3             (1, #,  1, 1_qr, 0)  
  \end{verbatim}
}

Below are the first 31 steps of the ex-machine's second execution.  
The first quantum random instruction
executed is {\small \verb|(0, #, 0, 0)| }.  The quantum random bit measured was 1, so the result of
this instruction is shown as  {\small \verb|(0, #, 0, 1_qr, 0)| }.  
The second quantum random  instruction executed is 
{\small \verb|(1, #, 1, 0)| }, which measured a 0, so the result of 
this instruction is shown as  {\small  \verb|(1, #, 1, 0_qr, 0)| }.

\medskip 

\bigskip 

\noindent  {\small  {\bf 2nd Execution of Random Walk Machine.  Computational Steps 1-31.} }

{ \scriptsize  
  \begin{verbatim} 
STATE     TAPE           TAPE HEAD         INSTRUCTION EXECUTED  
  0       ## 1#####          0             (0, #,  0, 1_qr, 0)   
  4       ##1 #####          1             (0, 1,  4, 1, 1)   
  4       ##1 0####          1             (4, #,  4, 0_qr, 0)   
  5       ## 1#####          0             (4, 0,  5, #, -1)   
  6       # #######         -1             (5, 1,  6, #, -1)   
  0       ## ######          0             (6, #,  0, #, 1)   
  0       ## 1#####          0             (0, #,  0, 1_qr, 0)   
  4       ##1 #####          1             (0, 1,  4, 1, 1)   
  4       ##1 1####          1             (4, #,  4, 1_qr, 0)   
  4       ##11 ####          2             (4, 1,  4, 1, 1)   
  4       ##11 1###          2             (4, #,  4, 1_qr, 0)   
  4       ##111 ###          3             (4, 1,  4, 1, 1)   
  4       ##111 1##          3             (4, #,  4, 1_qr, 0)   
  4       ##1111 ##          4             (4, 1,  4, 1, 1)   
  4       ##1111 0#          4             (4, #,  4, 0_qr, 0)   
  5       ##111 1##          3             (4, 0,  5, #, -1)   
  6       ##11 1###          2             (5, 1,  6, #, -1)   
  4       ##111 ###          3             (6, 1,  4, 1, 1)   
  4       ##111 0##          3             (4, #,  4, 0_qr, 0)   
  5       ##11 1###          2             (4, 0,  5, #, -1)   
  6       ##1 1####          1             (5, 1,  6, #, -1)   
  4       ##11 ####          2             (6, 1,  4, 1, 1)   
  4       ##11 0###          2             (4, #,  4, 0_qr, 0)   
  5       ##1 1####          1             (4, 0,  5, #, -1)   
  6       ## 1#####          0             (5, 1,  6, #, -1)   
  4       ##1 #####          1             (6, 1,  4, 1, 1)   
  4       ##1 0####          1             (4, #,  4, 0_qr, 0)   
  5       ## 1#####          0             (4, 0,  5, #, -1)   
  6       # #######         -1             (5, 1,  6, #, -1)   
  0       ## ######          0             (6, #,  0, #, 1)   
  0       ## 0#####          0             (0, #,  0, 0_qr, 0)   
  1       # #0#####         -1             (0, 0,  1, 0, -1)   
  \end{verbatim}
}

The first and second executions of the random walk ex-machine verify our 
statement in the introduction:  in contrast with the Turing machine, the execution behavior 
of the same ex-machine may be distinct at two different instances, even though each instance 
of the  ex-machine starts its execution with the same input on the tape, the same initial 
states and same initial instructions.  Hence, the ex-machine is a 
discrete, non-autonomous dynamical system.

\subsection{Meta Instructions}

Meta instructions are the second type of special instructions.  The execution of a meta instruction 
enables the ex-machine to self-modify its instructions.  
This means that an ex-machine's meta instructions can add new states, add new instructions 
or replace instructions.   Formally, the meta instructions $\mathcal{M}$ satisfy  
$\mathcal{M} \subset \{ (q, a, r, \alpha, y, J):$   
$q \in Q$  \hskip 0.3pc \verb|and| \hskip 0.3pc $r \in Q \cup \{|Q|\}$  \hskip 0.3pc  
\verb|and| \hskip 0.3pc  $a, \alpha \in A$  \hskip 0.3pc  \verb|and instruction| \hskip 0.3pc 
$J \in \mathcal{S} \cup  \mathcal{R}$$\}$.


Define $\mathcal{I} = \mathcal{S} \cup \mathcal{R} \cup \mathcal{M}$, as the set of standard, quantum random, and 
meta instructions.   To help describe how a meta instruction modifies $\mathcal{I}$, 
the {\it  unique state, scanning symbol condition} is defined:
for any two distinct instructions chosen from $\mathcal{I}$     
at least one of the first two coordinates must differ.  
More precisely, all 6 of the following uniqueness conditions must hold.

\newpage

\begin{enumerate}

\item  {    
            \hskip 0.5pc
            If $(q_1, \alpha_1, r_1, \beta_1, y_1)$ \hskip 0.3pc \verb|and| \hskip 0.3pc
             $(q_2, \alpha_2, r_2, \beta_2, y_2)$ are both in 
            $\mathcal{S}$, then $q_1 \ne q_2$ \hskip 0.3pc \verb|or| \hskip 0.3pc  $\alpha_1 \ne \alpha_2$.    
       }

\medskip 

\item  {    
            \hskip 0.5pc
            If  $(q_1, \alpha_1, r_1, \beta_1, y_1)$ $\in$ $\mathcal{S}$    
            \hskip 0.3pc   \verb|and|  \hskip 0.3pc 
            $(q_2, \alpha_2, r_2, y_2)$  $\in$  $\mathcal{R}$ or vice versa, 
            then $q_1 \ne q_2$ \hskip 0.3pc \verb|or| \hskip 0.3pc  $\alpha_1 \ne \alpha_2$.     
       }

\medskip 

\item    {    
              \hskip 0.5pc
              If  $(q_1, \alpha_1, r_1, y_1)$   
              \hskip 0.3pc   \verb|and|  \hskip 0.3pc 
              $(q_2, \alpha_2, r_2, y_2)$  are both in  $\mathcal{R}$, 
              then $q_1 \ne q_2$ \hskip 0.3pc \verb|or| \hskip 0.3pc  $\alpha_1 \ne \alpha_2$.     
         }

\medskip

\item    {   
             \hskip 0.5pc
             If  $(q_1, \alpha_1, r_1, y_1)$ $\in$ $\mathcal{R}$    
             \hskip 0.3pc   \verb|and|  \hskip 0.3pc  
             $(q_2, \alpha_2, r_2, a_2, y_2,  J_2)$  $\in$  $\mathcal{M}$ or vice versa, 
             then $q_1 \ne q_2$ \hskip 0.3pc \verb|or| \hskip 0.3pc  $\alpha_1 \ne \alpha_2$.     
         }

\medskip

\item   {    
             \hskip 0.5pc 
             If  $(q_1, \alpha_1, r_1, \beta_1, y_1)$ $\in$ $\mathcal{S}$    
             \hskip 0.3pc   \verb|and|  \hskip 0.3pc  
             $(q_2, \alpha_2, r_2, a_2, y_2,  J_2)$  $\in$  $\mathcal{M}$ or vice versa, 
             then $q_1 \ne q_2$ \hskip 0.3pc \verb|or| \hskip 0.3pc  $\alpha_1 \ne \alpha_2$.     
        }

\medskip 

\item  {  
             \hskip 0.5pc
             If $(q_1, \alpha_1, r_1, a_1, y_1,  J_1)$   
             \hskip 0.3pc   \verb|and|  \hskip 0.3pc 
             $(q_2, \alpha_2, r_2, a_2, y_2,  J_2)$  are both in  $\mathcal{M}$, 
             then $q_1 \ne q_2$ \hskip 0.3pc \verb|or| \hskip 0.3pc  $\alpha_1 \ne \alpha_2$.     
       }

\end{enumerate}

Before a valid machine execution starts, it is assumed that the standard, 
quantum random and meta instructions 
$\mathcal{S} \cup \mathcal{R} \cup \mathcal{M}$ always satisfy 
the unique state, scanning symbol condition.  This condition assures that 
there is no ambiguity on what instruction should be executed 
when the machine is in state $q$ and is scanning tape symbol $a$.  
Furthermore, the execution of a meta instruction preserves this uniqueness condition.

\begin{defn}\label{defn:execution_meta_instruction}   \hskip 1pc  Execution of Meta Instructions 

\smallskip 

\noindent A meta instruction $(q, a, r, \alpha, y, J)$ in $\mathcal{M}$ is executed as follows.   

\begin{itemize}

\item[$\bullet$]  {  
                      The first five coordinates $(q, a, r, \alpha, y)$ are executed as a standard instruction 
                      according to definition  \ref{defn:standard_instruction} with one caveat. 
                      State $q$ may be expressed as $|Q| - c_1$  and
                      state $r$ may be expressed as $|Q|$ or $|Q| - c_2$,  where  $0 < c_1, c_2 \le |Q|$.  
                      When  $(q, a, r, \alpha, y)$ is executed, if $q$ is expressed as  $|Q| - c_1$, the value of 
                      $q$ is instantiated to the current value of $|Q|$ minus $c_1$.  Similarly, 
                      if $r$ is expressed as $|Q|$ or $|Q| - c_2$, the value of state 
                      $r$ is instantiated to the current value of $|Q|$ or $|Q|$ minus $c_2$, respectively.
                  }

\medskip 

\item[$\bullet$]  {   
                      Subsequently, instruction $J$ modifies $\mathcal{I}$, where 
                      instruction $J$ has one of the two forms:  
                      $J = (q, a, r, \alpha, y)$ or $J = (q, a, r, y)$.   
                  }

\medskip

\item[$\bullet$] {  
                      For both forms, if $\mathcal{I}$ $\cup$ $\{ J \}$ still satisfies the unique state,
                      scanning symbol condition, then $\mathcal{I}$ is updated to $\mathcal{I} \cup \{ J \}$. 
                 }

\medskip

\item[$\bullet$] {  
                     Otherwise, there is an instruction $I$ in $\mathcal{I}$ whose 
                     first two coordinates $q$, $a$, are equal to instruction $J$'s 
                     first two coordinates.  In this case, instruction $J$ replaces 
                     instruction $I$ in $\mathcal{I}$.  That is, $\mathcal{I}$ is 
                     updated to $\mathcal{I} \cup \{ J \} - \{ I \}$.   
                 }

\end{itemize}

\end{defn}

\begin{rem}  \hskip 1pc  Ex-machine States and Instructions are Sequences of Sets 

\smallskip 

\noindent 
Now that the meta instruction has been defined, the purpose of this remark is to clarify, in terms of set theory,  
the definitions of machine states, standard, random, and meta instructions.  
In order to be compatible with the foundations of set theory, 
the machine states are formally a sequence of sets.  
Similarly, the standard instructions, random instructions and all 
ex-machine instructions are sequences of sets.  
Hence, the machine states should be expressed as $Q(m)$, 
where $m$ indicates that the ex-machine has executed its $m$th computational step.  
When our notation is formally precise, the standard instructions, random instructions 
and all ex-machine instructions should be expressed as 
$\mathcal{S}(m)$, $\mathcal{R}(m)$, and $\mathcal{I}(m)$, respectively.  
To simplify our notation and not detract from the main ideas, we usually do not include 
``$(m)$" when referring to $Q$, $\mathcal{S}$, $\mathcal{R}$, $\mathcal{M}$ or $\mathcal{I}$. 

\end{rem}

\medskip

\noindent  In regard to definition \ref{defn:execution_meta_instruction}, example \ref{ex:meta_instruction} shows how  
instruction $I$ is added to $\mathcal{I}$ and how new states are instantiated and added to $Q$.  

\bigskip 

\begin{example}\label{ex:meta_instruction}  \hskip 1pc   {Adding New States}

\smallskip 

\noindent Consider the meta instruction $(q, a_1, |Q|-1, \alpha_1, y_1, J)$, 
where $J = (|Q|-1, a_2, |Q|,$ $\alpha_2, y_2)$. 
After the standard instruction  $(q, a_1, |Q|-1, \alpha_1, y_1)$ is executed, 
this meta instruction adds one new state $|Q|$ to the machine states $Q$ and also 
adds the instruction $J$, instantiated with the current value of $|Q|$.    
Figure 1 
shows the execution of this meta instruction for the specific values  
$Q = \{0, 1, 2, 3, 4, 5, 6, 7\}$, $A = \{$\verb|#, 0, 1|$\}$, 
$q = 5$, $a_1 = 0$, $\alpha_1 = 1$, $y_1 = 0$, $a_2 = 1$, $\alpha_2 =$ \verb|#|, 
and $y_2 = -1$.  States and alphabet symbols are shown in red and blue, respectively.

\begin{figure}[h]\label{fig:meta_instruction}
   \centering
   \includegraphics[width=11 cm]{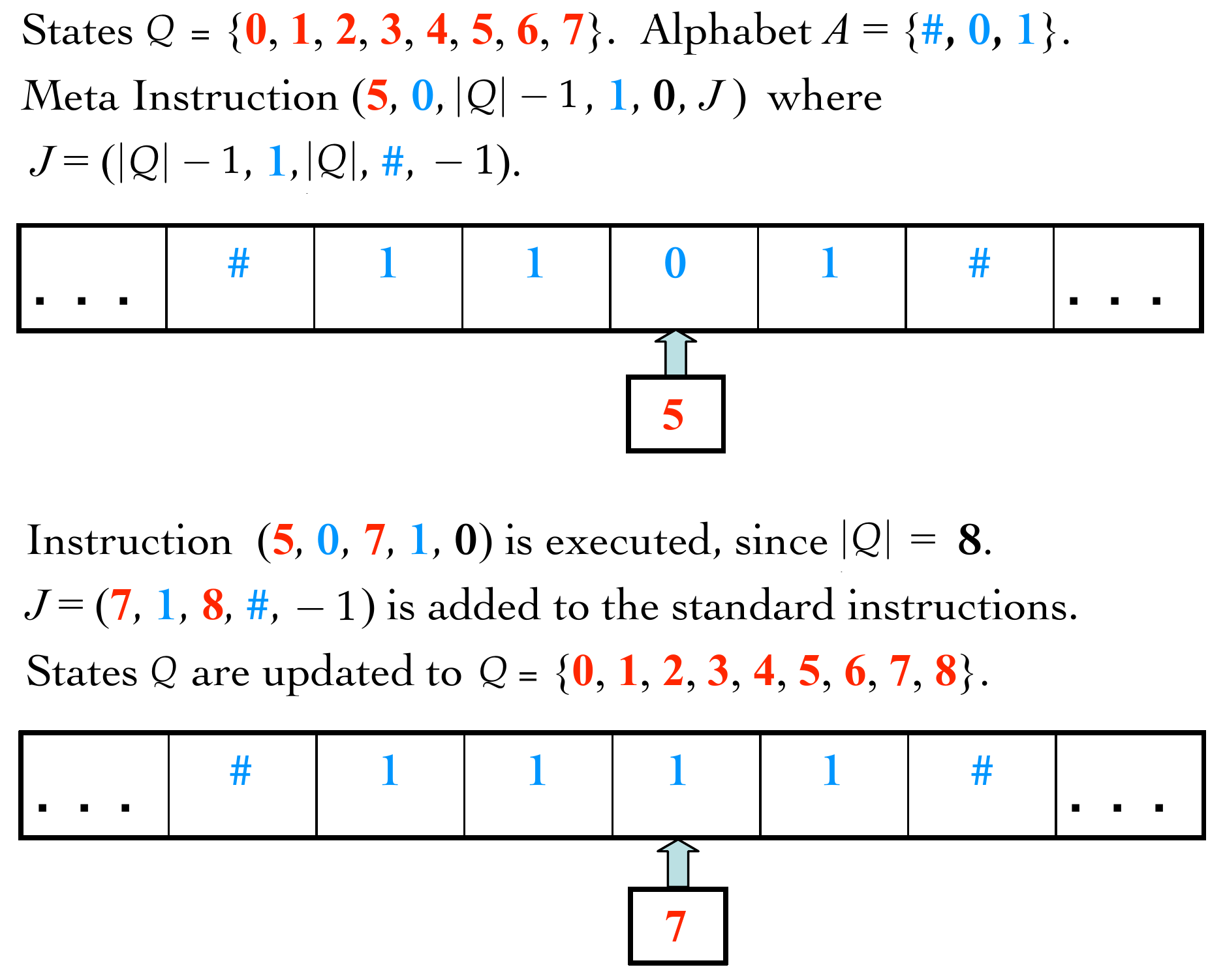}

   \bigskip 

   \includegraphics[width=11 cm]{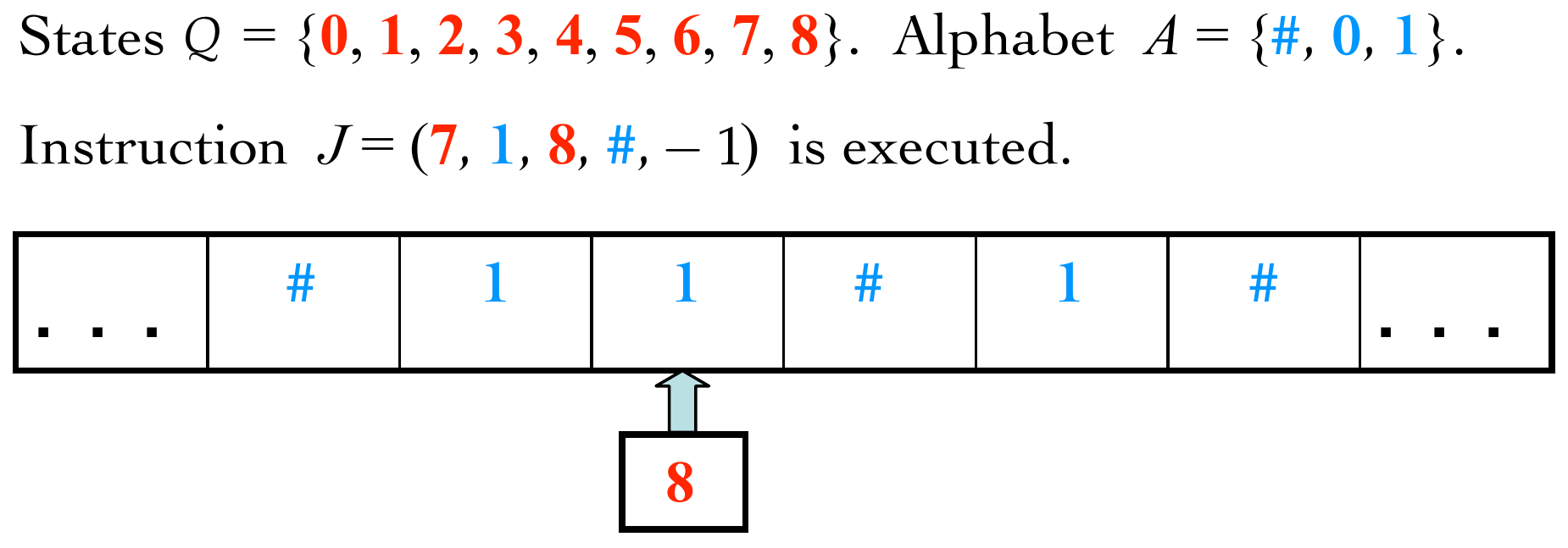}
   \caption{Meta Instruction Execution}
\end{figure}

\end{example}

Let $\mathfrak{X}$ be an ex-machine.  The instantiation of $|Q|-1$ and $|Q|$ in a meta instruction $I$ 
invokes {\it self-reflection} about $\mathfrak{X}$'s current number of states, at the moment when 
$\mathfrak{X}$ executes  $I$. This simple type of self-reflection poses no obstacles in physical realizations.   
In particular, a LISP implementation \cite{mueller} along with quantum random bits  
measured from \cite{wahl} simulates all executions of the ex-machines provided herein.


\bigskip 

\begin{defn}\label{defn:simple_meta_instruction}  \hskip 1pc {Simple Meta Instructions}

\smallskip 

\noindent A simple meta instruction has one of the forms 
$(q, a, |Q|-c_2, \alpha, y)$, $(q, a, |Q|, \alpha, y)$, $(|Q|-c_1, a, r, \alpha, y)$,
$(|Q|-c_1, a, |Q|-c_2, \alpha, y)$, $(|Q|-c_1, a, |Q|, \alpha, y)$, where 
$0 < c_1, c_2 \le |Q|$.  The expressions $|Q|-c_1$, $|Q|-c_2$ and $|Q|$ are instantiated to 
a state based on the current value of $|Q|$ when the instruction is executed.  

\end{defn}

\noindent  In this paper, ex-machines only self-reflect with the symbols 
$|Q|-1$ and  $|Q|$. 

\begin{example}\label{example:simple_meta_instructions}  \hskip 1pc  {Execution of Simple Meta Instructions}

\smallskip

\noindent Let $A = \{$\verb|0, 1, #|$\}$ and $Q = \{0 \}$.  ex-machine $\mathfrak{X}$ has 3 simple meta instructions.

\smallskip 

{  \scriptsize
  \begin{verbatim}
(|Q|-1, #, |Q|-1, 1, 0)
(|Q|-1, 1, |Q|, 0, 1)
(|Q|-1, 0, |Q|, 0, 0)
  \end{verbatim}
}
\end{example}

With an initial blank tape and starting state of 0, the first four computational steps are 
shown below.  In the first step,  $\mathfrak{X}$'s tape head is scanning a {\small \verb|#|}
and the ex-machine state is 0. Since $|Q| = 1$, simple meta instruction 
{\small \verb|(|$|$\verb|Q|$|$\verb|-1, #, |$|$\verb|Q|$|$\verb|-1, 1, 0)| } 
instantiates to {\small \verb|(0, #, 0, 1, 0)|}, and executes.  

\smallskip 

{  \scriptsize
  \begin{verbatim} 
 STATE    TAPE        TAPE HEAD      INSTRUCTION           NEW INSTRUCTION               
   0      ## 1###         0          (0, #, 0, 1, 0)       (0, #, 0, 1, 0)  
   1      ##0 ###         1          (0, 1, 1, 0, 1)       (0, 1, 1, 0, 1)
   1      ##0 1##         1          (1, #, 1, 1, 0)       (1, #, 1, 1, 0)
   2      ##00 ##         2          (1, 1, 2, 0, 1)       (1, 1, 2, 0, 1)   
   \end{verbatim}
}

\noindent  In the second step,  the tape head is scanning a $1$ and the state is 0.  
Since $|Q| = 1$, instruction 
{\small \verb|(|$|$\verb|Q|$|$\verb|-1, 1, |$|$\verb|Q|$|$\verb|, 0, 1)| } 
instantiates to  {\small \verb|(0, 1, 1, 0, 1)|},  
executes and updates $Q = \{0, 1\}$. In the third step, the tape head 
is scanning a {\small \verb|#|} and the state is 1.  Since $|Q| = 2$, instruction 
{\small \verb|(|$|$\verb|Q|$|$\verb|-1, #, |$|$\verb|Q|$|$\verb|-1, 1, 0)| }  
instantiates to  {\small \verb|(1, #, 1, 1, 0)|} 
and executes.  In the fourth step, the tape head is scanning a $1$ and the state is 1.
Since $|Q| = 2$, instruction  
{\small \verb|(|$|$\verb|Q|$|$\verb|-1, 1, |$|$\verb|Q|$|$\verb|, 0, 1)| }  instantiates to 
{\small \verb|(1, 1, 2, 0, 1)|}, executes and updates $Q = \{0, 1, 2\}$. 
During these four steps, two simple meta instructions create four new instructions 
and add new states 1 and 2.


\begin{defn}\label{defn:finite_initial_conditions}  \hskip 1pc  {Finite Initial Conditions}

\smallskip 

\noindent An ex-machine is said to have finite initial conditions if the following four conditions are satisfied before 
the ex-machine starts executing.

\begin{itemize}

\item[1.] { \hskip 1pc The number of states $|Q|$ is finite.  }

\smallskip

\item[2.] { \hskip 1pc The number of alphabet symbols $|A|$ is finite.  }

\smallskip

\item[3.]{  \hskip 1pc The number of machine instructions $|\mathcal{I}|$ is finite.  }

\smallskip 

\item[4.]{  \hskip 1pc The tape is finitely bounded.  }

\smallskip

\end{itemize}

\end{defn}

It may be useful to think about the initial conditions of an ex-machine as analogous to the boundary 
value conditions of a differential equation.    While trivial to verify, the purpose of 
remark \ref{rem:finite_conditions} is to assure that computations performed with an ex-machine 
are physically plausible.

\begin{rem}\label{rem:finite_conditions}  \hskip 1pc   { Finite Initial Conditions }

\smallskip 
 
\noindent If the machine starts its execution with finite initial conditions, then after the machine has 
 executed $l$ instructions for any positive integer $l$, the current number of states $Q(l)$ is 
 finite and the current set of instructions $\mathcal{I}(l)$ is finite.  Also, the tape $T$ is 
 still finitely bounded and the number of measurements obtained from the quantum random source is finite.  

\end{rem}

\begin{proof}
        The remark follows immediately from definition \ref{defn:finite_initial_conditions} 
        of finite initial conditions and machine instruction definitions \ref{defn:standard_instruction}, 
        \ref{defn:qr_instruction}, and  \ref{defn:execution_meta_instruction}.   
        In particular, the execution of one meta instruction adds at most one new 
        instruction and one new state to $Q$.  
\end{proof}

Definition \ref{defn:evolving} defines new ex-machines that may have evolved from 
computations of prior ex-machines that have halted.   The notion of evolving is 
useful because the quantum random and meta instructions can self-modify an 
ex-machine's instructions as it executes.  In contrast with the ex-machine, after 
a Turing machine halts, its instructions have not changed.

This difference motivates the next definition, which is illustrated by the following. 
Consider an initial ex-machine $\mathfrak{X}_0$  that has 9 initial states and 
15 initial instructions.  $\mathfrak{X}_0$ starts executing on a finitely bounded 
tape $T_0$ and halts.  When the ex-machine halts, it (now called $\mathfrak{X}_1$) has 14 states 
and 24 instructions and the current tape is $S_1$.  We say that ex-machine  
$\mathfrak{X}_0$ with tape $T_0$ {\it evolves to} ex-machine $\mathfrak{X}_1$ 
with tape $S_1$.

\begin{defn}\label{defn:evolving} \hskip 1pc  Evolving an ex-machine

\smallskip 

\noindent Let $T_0$, $T_1$, $T_2$ $\dots$ $T_{i-1}$ each be a finitely bounded tape.  
Consider ex-machine $\mathfrak{X}_0$
with finite initial conditions.  $\mathfrak{X}_0$ starts executing with tape $T_0$ and 
evolves to ex-machine $\mathfrak{X}_1$ with tape $S_1$.  
Subsequently, $\mathfrak{X}_1$ starts executing with tape $T_1$ and 
evolves to $\mathfrak{X}_2$ with tape $S_2$.  This means that when ex-machine 
$\mathfrak{X}_1$ starts executing on tape $T_1$, its 
instructions are preserved after the halt with tape $S_1$.  
The ex-machine evolution continues until $\mathfrak{X}_{i-1}$ starts executing with tape $T_{i-1}$ 
and evolves to  ex-machine $\mathfrak{X}_{i}$ with tape $S_{i}$.   One says that ex-machine 
$\mathfrak{X}_0$ with finitely bounded tapes $T_0$, $T_1$, $T_2$ $\dots$ $T_{i-1}$ 
evolves to  ex-machine $\mathfrak{X}_i$ after $i$ halts. 

\end{defn}

When ex-machine $\mathfrak{X}_0$ evolves to $\mathfrak{X}_1$ and subsequently $\mathfrak{X}_1$ evolves to 
$\mathfrak{X}_2$ and so on up to ex-machine  $\mathfrak{X}_n$, then ex-machine $\mathfrak{X}_i$ 
is called an {\it ancestor} of  ex-machine $\mathfrak{X}_j$ whenever $0 \le i < j \le n$. 
Similarly, ex-machine $\mathfrak{X}_j$ is called a {\it descendant} of  
ex-machine $\mathfrak{X}_i$ whenever $0 \le i < j \le n$.  The sequence of ex-machines 
$\mathfrak{X}_0$ $\rightarrow$ $\mathfrak{X}_1$ $\rightarrow$  $\dots$ $\rightarrow$  $\mathfrak{X}_n$  $\dots$ 
is called an {\it evolutionary path}.


\section{Quantum Randomness}\label{sect:quantum_randomness}

On a first reading, one may choose to skip this section, 
by assuming that there is adequate physical justification for axioms  
\ref{axiom_qr1}  and  \ref{axiom_qr2}.   Overall, the ex-machine uses 
quantum randomness as a computational tool.  Hence, part of our goal was to 
use axioms  \ref{axiom_qr1}  and  \ref{axiom_qr2} for our quantum random 
instructions, because the axioms are supported by the 
empirical evidence of various quantum random number generators 
\cite{abellan,bierhorst,kulikov,ma,rohe,stefanov,stipcevic,wahl,yang}.  
In practice, however, a physical implementation of a quantum random number generator can only 
generate a finite amount of data and only a finite number of statistical tests can be 
performed on the data.  Due to these limitations, one goal of quantum random theory 
\cite{pironio,calude_qr2012,calude_qr2014,calude_qr2015,svozil_3_criteria}, 
besides general understanding, is to certify the mathematical properties, assumed 
about actual quantum random number generators, and assure that the theory is 
a reasonable extension of quantum mechanics 
\cite{bohr,born_1925,born_1926,heisenberg_1926,heisenberg_1927,heisenberg_book,schrodinger,schrodinger_1935,epr,bell,clauser,kochen}.


We believe it is valuable to reach both a pragmatic (experimental) and theoretical viewpoint.  
In pure mathematics, the formal system and the logical steps in the mathematical proofs need 
to be checked.  In our situation, it is possible for a mathematical theory of quantum randomness 
(or for that matter any theory in physics) 
to be consistent (i.e., in the sense of mathematics) and have valid mathematical proofs, 
yet the theory still does not adequately model the observable properties of 
the underlying physical reality.  If just one subtle mistake or oversight is made 
while deriving a mathematical formalism from some physical assumptions, then the 
mathematical conclusions arrived at -- based on the theory --  
may represent physical nonsense.
Due to the infinite nature of randomness, this branch of science is faced with the 
challenging situation that the mathematical properties of randomness 
can only be provably tested with an infinite amount of experimental data and an 
infinite number of tests.   Since we only have the means to collect a finite amount of 
data and perform a finite amount of statistical tests, we must acknowledge   
that experimental tests on a quantum random number generator, designed according to a 
mathematical / physical theory, can only falsify the theory \cite{popper} for that class of 
quantum random number generators.  As more experiments are performed,  successful statistical tests 
calculated on longer sequences of random data may strengthen the empirical evidence, but 
they cannot scientifically prove the theory.  This paragraph provides at least some motivation 
for some of our pragmatic points, described in the next three paragraphs.

In sections \ref{msf:x_machine_languages} and \ref{msf:halting_problem}, the mathematical proofs 
rely upon the property that for any $m$, all $2^m$ binary strings are equally likely to 
occur when a quantum random number generator takes $m$ binary measurements.\footnote{One has to be careful not to 
misinterpret quantum random axioms \ref{axiom_qr1}  and  \ref{axiom_qr2}.  For example, 
the Champernowne sequence  $01$  \hskip 0.1pc $00$ $01$ $10$ $11$ \hskip 0.1pc 
$000$ $001$ $010$ $011$ $100$ $101$ $110$ $111$ \hskip 0.1pc $0000$ $\dots$     
is sometimes cited as a sequence that is Borel normal, yet still Turing computable.
However, based on the mathematics of random walks \cite{feller_vol1}, the Champernowne sequence 
catastrophically fails the expected number of changes of sign as $n \rightarrow \infty$.  
Since all $2^m$ strings are equally likely, the expected value of changes of sign follows 
from the reflection principle and simple counting arguments, as shown in III.5 of \cite{feller_vol1}. }  
In terms of the ex-machine computation performed, how one of these binary strings 
is generated from some particular type of quantum process is not the critical issue.

Furthermore, most of the $2^m$ binary strings have high Kolmogorov complexity 
\cite{solomonoff,kolmogorov,chaitin}.  This fact leads to the following mathematical intuition that 
enables new computational behaviors:  the execution of quantum random instructions working 
together with meta instructions enables the ex-machine to increase its program complexity 
\cite{shannon,schmitt} as it evolves.  In some cases, the increase in program complexity can 
increase the ex-machine's computational power as the ex-machine evolves.  Also, notice the 
distinction here between the program complexity of the ex-machine and Kolmogorov complexity.  
The definition of Kolmogorov complexity only applies to standard machines.  Moreover, the 
program complexity (e.g., the Shannon complexity $|Q||A|$ \cite{shannon}) stays fixed for 
standard machines.  In contrast, an ex-machine's program complexity can increase without 
bound, when the ex-machine executes quantum random and meta instructions that productively 
work together.  (For example, see ex-machine \ref{machine:meta_a_machine}, called $\mathfrak{Q}(x)$.)

With this intuition about complexity in mind, we provide a concrete example.  
Suppose the quantum random generator demonstrated in \cite{kulikov}, certified 
by the strong Kochen-Specker theorem \cite{svozil_3_criteria,calude_qr2012,calude_qr2014,calude_qr2015}, 
outputs the 100-bit string $a_0 a_1 \dots a_{99} =$ 
\verb|1011000010101111001100110011100010001110010101011011110000000010011001|

\noindent    \verb|000011010101101111001101010000| 
to ex-machine $\mathfrak{X_1}$.

Suppose a distinct quantum random generator using radioactive decay  \cite{rohe} outputs the 
same 100-bit string $a_0 a_1 \dots a_{99}$ to a distinct ex-machine $\mathfrak{X_2}$.  
Suppose $\mathfrak{X_1}$ and  $\mathfrak{X_2}$ have identical programs  with the same initial 
tapes and same initial state. Even though radioactive decay 
\cite{curie1,curie2,rutherford1,rutherford2,rutherford3} was discovered over 100 years ago 
and  its physical basis is still phenomenological, 
the execution behavior of  $\mathfrak{X_1}$ and  $\mathfrak{X_2}$ 
are indistinguishable for the first 100 executions of their quantum random instructions.
In other words, ex-machines $\mathfrak{X_1}$ and  $\mathfrak{X_2}$ exhibit 
execution behaviors that are independent of the quantum process that generated 
these two identical binary strings.



\subsection{Mathematical Determinism and Unpredictability }

Before some of the deeper theory on quantum randomness is reviewed, we take a step back 
to view randomness from a broader theoretical perspective.  While we generally agree with 
the philosophy of Eagle \cite{eagle} that  {\it randomness is unpredictablity},  
example \ref{ex:math_gedankenexperiment}  helps sharpen the differences 
between {\it indeterminism} and {\it unpredictability}.


\begin{example}\label{ex:math_gedankenexperiment}  \hskip 1pc {A Mathematical Gedankenexperiment }

\end{example}

Our gedankenexperiment demonstrates a deterministic system 
which exhibits an extreme amount of unpredictability.  Some work is needed to define the dynamical system 
and summarize its mathematical properties before we can present the gedankenexperiment.

Consider the quadratic map $f: \mathbb{R} \rightarrow  \mathbb{R}$, where 
$f(x) = \frac{9}{2} x(1 - x)$.  Set $I_0 = [0, \frac{1}{3}]$ and $I_1 = [\frac{1}{3}, \frac{2}{3}]$.   
Set $B = (\frac{1}{3}, \frac{2}{3})$. 
Define the set $\Lambda = \{ x \in [0, 1]:  f^n(x) \in I_0 \cup I_1$ 
\verb| for all | $n \in \mathbb{N} \}$.
 $0$ is a fixed point of $f$ and 
$f^2(\frac{1}{3}) = f^2(\frac{2}{3}) = 0$, so the boundary points of $B$ lie 
in $\Lambda$.   Furthermore, whenever $x \in B$, then $f(x) < 0$ and 
${\underset{n \rightarrow \infty} \lim} f^n(x) = -\infty$.  This means all orbits  
that exit $\Lambda$ head off to $- \infty$.

The inverse image $f^{-1}(B)$ is two open intervals $B_{0} \subset I_0$ and 
$B_{1} \subset I_1$ such that $f(B_{0}) = f(B_{1}) = B$.  Topologically,  $B_{0}$ behaves like  
Cantor's open middle third of $I_0$ and $B_{1}$ behaves like 
Cantor's open middle third of $I_1$.  Repeating the inverse image indefinitely, 
define the set $H = B \cup {\overset{\infty} {\underset{k = 1} \cup}} f^{-k}(B)$.  
Now $H \cup \Lambda = [0, 1]$ and $H \cap \Lambda = \emptyset$.

Using dynamical systems notation, set 
$\Sigma_2 = \{0, 1 \}^{\mathbb{N}}$.  Define the shift map 
$\sigma: \Sigma_2 \rightarrow \Sigma_2$, where 
$\sigma( a_0 a_1 \dots ) = (a_1 a_2 \dots)$.  For each $x$ in $\Lambda$, $x$'s 
trajectory in $I_0$ and $I_1$ corresponds to a unique point in $\Sigma_2$: 
define  $h : \Lambda \rightarrow \Sigma_2$ 
as  $h(x) = (a_0 a_1 \dots)$ such that for each $n \in  \mathbb{N}$, 
set $a_n = 0$ if $f^n(x) \in I_0$  and  $a_n = 1$  if $f^n(x) \in I_1$.

For any two points $(a_0 a_1 \dots)$ and $(b_0 b_1 \dots)$ in $\Sigma_2$, 
define the metric  $d\big{(} (a_0 a_1 \dots),$  $(b_0 b_1 \dots) \big{)}$ $=$ 
${\overset{\infty} {\underset{i = 0}\sum}} |a_i - b_i| 2^{-i}$.  Via 
the standard topology on $\mathbb{R}$ inducing  the subspace 
topology on $\Lambda$,  it is straightforward to verify that $h$ is a homeomorphism 
from $\Lambda$ to $\Sigma_2$.   Moreover, $h \circ f = \sigma \circ h$, so $h$ is 
a topological conjugacy.   The set $H$ and the topological conjugacy $h$ enable us to 
verify that $\Lambda$ is a Cantor set.  This means that $\Lambda$ is uncountable, 
totally disconnected, compact and every point of $\Lambda$ is a limit point of $\Lambda$.

We are ready to pose our {\it mathematical gedankenexperiment}.  We make 
the following assumption about our mathematical observer.  When our  
observer takes a physical measurement of a point $x$ in $\Lambda_2$, 
she measures a 0 if $x$ lies in $I_0$ and measures a 1 if $x$ lies in $I_1$.  
We assume that she cannot make her observation any more accurate based on 
our idealization that is analogous to the following:  measurements at the quantum level 
have limited resolution due to the wavelike properties of matter 
\cite{debroglie_1924,debroglie_1925,born_1926,heisenberg_1926,heisenberg_1927,schrodinger,feynman}.  
Similarly, at the second observation, our observer measures a 0 if 
$f(x)$ lies in $I_0$ and 1 if $f(x)$ lies in $I_1$.  
Our observer continues to make these observations until she has measured whether 
$f^{k-1}(x)$ is in $I_0$ or in $I_1$.  Before making her $k+1$st observation,   
can our observer make an effective prediction whether $f^{k}(x)$ lies in 
$I_0$ or $I_1$ that is correct for more than 50\% of her predictions?


The answer is no when $h(x)$ is a generic point (i.e., in the sense of Lebesgue measure) 
in $\Sigma_2$.    Set $\mathcal{R}$ to the Martin-L{\"o}f random points in $\Sigma_2$.  
Then  $\mathcal{R}$  has Lebesgue measure 1 in $\Sigma_2$ 
\cite{feller_vol1,chaitin_1969}, so its complement $\Sigma_2 - \mathcal{R}$ has Lebesgue measure 0.  
For any $x$ such that $h(x)$ lies in $\mathcal{R}$,   
then our observer cannot predict the orbit of $x$ with a Turing machine.  Hence, via the 
topological conjugacy $h$, we see that for a generic point $x$ in $\Lambda$, $x$'s orbit 
between $I_0$ and $I_1$ is Martin-L{\"o}f random   -- even though $f$ is mathematically 
deterministic and $f$ is a Turing computable function.

Overall, the dynamical system $(f, \Lambda)$ is mathematically deterministic and each real number $x$ in $\Lambda$ has a 
definite value.  However, due to the lack of resolution in the observer's measurements, 
the orbit of generic point $x$ is unpredictable -- in the sense of Martin-L{\"o}f random.


\subsection{Quantum Random Theory}

A deeper theory on quantum randomness stems from the seminal EPR paper \cite{epr}. 
Einstein, Podolsky and Rosen propose a necessary condition for a complete theory of quantum mechanics: 
{\it  Every element of physical reality must have a counterpart in the physical theory}.
Furthermore, they state that elements of physical reality must be found by the results of experiments 
and measurements.   While mentioning that there might be other ways of recognizing a physical reality, 
EPR propose the following as a reasonable criterion for a complete theory of quantum mechanics:

\begin{quote}
            {\it 
                 If, without in any way disturbing a system, one can predict with 
                 certainty (i.e., with probability equal to unity) the value of a physical quantity, 
                 then there exists an element of physical reality corresponding to this physical quantity.  
            }
\end{quote}

	They consider a quantum-mechanical description of a particle, having one degree of freedom.	
  After some analysis, they conclude that a {\it definite value} of the coordinate, for a particle in the 
  state given by $\psi = e^{\frac{2 \pi i} {h} p_o x}$, is not predictable, but may be obtained only 
  by a direct measurement.  However, such a measurement disturbs the particle and changes its state.  
  They remind us that in quantum mechanics, 
  {\it when the momentum of the particle is known, its coordinate has no physical reality}.
  This phenomenon has a more general mathematical condition that if the operators corresponding to two physical 
  quantities, say $A$ and $B$, do not commute, then a precise knowledge of one of them precludes a
   precise knowledge of the other.   Hence, EPR reach the following conclusion:  

\begin{itemize}

\item[(I)]  
               {  
                   The quantum-mechanical description of physical reality given by the wave function is not complete.  
               } 


\verb|OR|

 \smallskip 

\item[(II)]  
            {   
               When the operators corresponding to two physical quantities (e.g., position and momentum) 
               do not commute (i.e. $AB \ne BA$), the two quantities cannot have the same simultaneous reality. 
            }
               
\end{itemize}

EPR justifies this conclusion by reasoning that if both physical quantities had a simultaneous reality and 
consequently definite values, then these definite values would be part of the complete description.  
Moreover, if the wave function provides a complete description of physical reality, then 
the wave function would contain these definite values and the definite values would be predictable.

From their conclusion of I \verb|OR| II, 
EPR assumes the negation of  I -- that the wave function does give a complete description of physical 
reality.  They analyze two systems that interact over a finite interval of time.  And show 
by a thought experiment of measuring each system via wave packet reduction that it is possible to assign two 
different wave functions to the same physical reality.    Upon further analysis 
of two wave functions that are eigenfunctions of two non-commuting operators, they arrive at the conclusion 
that two physical quantities, with non-commuting operators can have simultaneous reality.
From this contradiction or paradox (depending on one's perspective), they conclude that the 
quantum-mechanical description of reality is not complete.

In \cite{bohr}, Bohr responds to the EPR paper.  Via an analysis involving single slit experiments and 
double slit (two or more) experiments, Bohr explains how during position measurements 
that  momentum is transferred between the object being observed and the measurement apparatus.  
Similarly, Bohr explains that during momentum measurements the object is  displaced.  
Bohr also makes a similar observation about time and energy:  
``{\it it is  excluded in principle to control the energy which goes into the clocks without interfering 
 essentially with their use as time indicators}''. 
 Because at the quantum level it is impossible to control the interaction 
between the object being observed and the measurement apparatus, Bohr argues for a 
``{\it final renunciation of the classical ideal of causality}'' 
and a ``{\it radical revision of physical reality}''.

From his experimental analysis, Bohr concludes  that the meaning of EPR's expression 
{\it without in any way disturbing the system} creates an ambiguity in their argument.   
Bohr states: ``There is essentially the question of {\it an influence on the very conditions 
which define the possible types of predictions regarding the future behavior of the system}.  
Since these conditions constitute an inherent element of the description of any phenomenon 
to which the term {\it physical reality} can be properly attached, we see that the argumentation 
of the mentioned authors does not justify their conclusion that quantum-mechanical description 
is essentially incomplete.''    Overall, the EPR versus Bohr-Born-Heisenberg position set the 
stage for the problem of hidden variables in the theory of quantum mechanics.

The Kochen-Specker \cite{specker,kochen} approach -- to a hidden variable theory in quantum mechanics --  
is addressed independently of any reference to locality or non-locality \cite{bell}.  
Instead, they assume a stronger condition than locality: 
hidden variables are only associated with the quantum system being measured.
No hidden variables are associated with the measurement apparatus.  
This is the physical (non-formal) notion of {\it non-contextuality}.


In their theory, a set of observables are defined, where in the case of 
quantum mechanics, the observables (more general) are represented by 
the self-adjoint operators on a separable Hilbert space. The Kochen-Specker 
theorem \cite{specker,kochen} proves that it is impossible for a 
non-contexual hidden variable theory to assign values to finite sets of observables, 
which is also consistent with the theory of quantum mechanics.  More precisely, 
the Kochen-Specker theorem demonstrates a contradiction between 
the following two assumptions:

\begin{itemize}

\item[($\mathcal{A}_1$)] {   
                              The set of observables in question have pre-assigned definite values.  Due to complementarity, 
                              the observables may not be all simultaneously co-measurable, where the formal definition 
                              of co-measurable means that the observables commute. 
                         }

\medskip

\item[($\mathcal{A}_2$)] {   
                             The measurement outcomes of observables are non-contextual.  This means the
                             outcomes are independent of the other co-measurable observables that are measured 
                             along side them, along with the requirement that in any ``complete'' set of 
                             mutually co-measurable yes-no propositions, exactly one proposition should be assigned 
                             the value ``yes''.
                         }

\end{itemize}

\noindent  Making assumption $\mathcal{A}_2$ more precise, the mutually co-measurable yes-no propositions are 
represented by mutually orthogonal projectors spanning the Hilbert space.

The Kochen-Specker theorem does not explicitly identify the observables that violate at least one of the 
assumptions $\mathcal{A}_1$ or $\mathcal{A}_2$.   The original Kochen-Specker theorem was not developed with a  
goal of locating the particular observable(s) that violate assumptions $\mathcal{A}_1$ or $\mathcal{A}_2$.

\medskip

In \cite{calude_qr2012}, Abbott, Calude and Svozil (ACS) advance beyond the Kochen-Specker theory, 
but also preserve the quantum logic formalism of von Neumann \cite{von_neumann1,von_neumann2} and 
Kochen-Specker  \cite{kochen_specker1,kochen_specker2}.  
Their work can be summarized as follows:

\begin{enumerate}

\item{   
         They explicitly formalize the physical notions of value definiteness (indefiniteness)
         and contextuality. 
     }

\smallskip 

\item{
         They sharpen in what sense the Kochen-Specker and Bell-like theorems imply the 
         violation of the non-contextuality assumption $\mathcal{A}_2$.
     }

\smallskip 

\item{
          They provide collections of observables that do not produce Kochen-Specker contradictions.  
     }

\smallskip

\item{
          They propose the reasonable statement that 
          {\it quantum random number generators\footnote{ There are quantum random number
          expanders \cite{pironio,shalm}, based on the Bell inequalities \cite{bell,clauser} and non-locality.  
          We believe random number expanders are better suited for cryptographic applications where an active adversary 
          may be attempting to subvert the generation of a cryptographic key.  
          In this paper, advancing cryptography is not one of our goals.} 
          should be certified by value indefiniteness, 
          based on the particular observables utilized for that purpose}.  
          Hence, their extension of the Kochen-Specker theory needs to 
          locate the violations of non-classicality.  
     }

\smallskip 

\end{enumerate}

The key intuition for quantum random number generators that are designed 
according to the ACS protocol is that value indefiniteness implies unpredictability.    
One of the primary strengths of ACS theory over the Kochen-Specker and Bell-type theorems 
is that it helps identify the particular observables that are value indefinite.  The identification 
of value indefinite observables helps design a quantum random number generator, based on mathematics with 
reasonable physical assumptions rather than ad hoc arguments. In particular, their results assure that, 
if quantum mechanics is correct, the particular observables 
-- used in the measurements that produce the random number sequences -- 
are provably value indefinite.  In order for a quantum random number generator to capture the 
value indefiniteness during its measurements, their main idea is to identify pairs of projection 
observables that satisfy the following:   if one of the projection observables is assigned the value 1 
by an admissible assignment function such that the projector observables $\mathcal{O}$
are non-contextual, then the other observable in the pair must be value indefinite.

Now, some of their formal definitions are briefly reviewed to clarify how the ACS corollary 
helps design a protocol for a dichotomic quantum random bit generator, 
operating in a three-dimensional Hilbert space.    
Vectors $\ket{\psi}$ lie in the Hilbert space  $\mathbb{C}^n$, 
where $\mathbb{C}$ is the complex numbers.    
The {\it observables} $\mathcal{O} \subset \{ P_{\psi}: \ket{\psi} \in \mathbb{C}^n \}$  
are a non-empty subset of the projection operators 
$P_{\psi} = \frac{\ket{\psi} \bra{\psi}}{\braket{\psi | \psi}}$, 
that project onto the linear subspace of $\mathbb{C}^n$, spanned by non-zero vector  $\ket{\psi}$.  
They only consider pure states.  
The set of {\it measurement contexts} over $\mathcal{O}$
is the set $\mathcal{C} \subset \{ \{P_1, P_2, \dots P_n \}$  $|$ 
$P_i \in \mathcal{O}$ \verb| and | $\braket{i | j} = 0$ \verb| for | $i \ne j \}$.
A {\it context} $C \in \mathcal{C}$ is a maximal set of {\it compatible projection observables}, 
where compatible means the observables can be simultaneously measured.

Let  $\nu : \{ (o, C)$ $|$ $o \in \mathcal{O}, C \in \mathcal{C}$
\verb|and|  $o \in \mathcal{C} \}$ 
${\overset{o} \rightarrow}$  $\{0, 1\}$ be a partial function.  
For some $o, o^{\prime} \in \mathcal{O}$ and $C, C^{\prime} \in \mathcal{C}$, then 
$\nu(o, C)= \nu(o^{\prime}, C^{\prime})$ means both $\nu(o, C)$  
and $\nu(o^{\prime}, C^{\prime})$ are defined and they have equal values in $\{0, 1\}$.
The expression $\nu(o, C) \ne \nu(o^{\prime}, C^{\prime})$  means either $\nu(o, C)$ 
or $\nu(o^{\prime}, C^{\prime})$ are not defined or they are both defined but have different values.  
$\nu$ is called an {\it assignment} function and $\nu$ formally expresses the notion of a hidden
variable.  $\nu$ specifies in advance the result obtained from the measurement of an observable.

An observable $o \in C$ is {\it value definite} in the context $C$ with respect to 
$\nu$ if $\nu(o, C)$ is defined.  Otherwise, $o$ is {\it value indefinite} in $C$.  
If $o$ is value definite in all contexts $C \in \mathcal{C}$ for which $o \in C$, 
then $o$ is called value definite with respect to $\nu$
If $o$ is value indefinite in all contexts $C$, then $o$ is called {\it value indefinite} 
with respect to $\nu$.

The set $\mathcal{O}$ is {\it value definite} with respect to 
$\nu$ if every observable $o \in \mathcal{O}$ is value definite with respect to $\nu$. 
 This formal definition of value definiteness corresponds to the 
classical notion of determinism.  Namely, $\nu$ assigns a definite value to an observable, which 
expresses that we are able to predict in advance the value obtained via measurement.

An observable $o \in \mathcal{O}$ is {\it non-contextual} with respect to $\nu$ if for all contexts 
$C, C^{\prime} \in \mathcal{C}$ such that $o \in C$ and $o \in C^{\prime}$, 
then $\nu(o, C) = \nu(o, C^{\prime})$.  
Otherwise, $\nu$ is {\it contextual}.  The set of observables $\mathcal{O}$ is {\it non-contextual}
 with respect to $\nu$ if every observable $o \in \mathcal{O}$ which is not value indefinite 
 (i.e. value definite in some context) is non-contextual with respect to $\nu$.   Otherwise, the 
 set of observables $\mathcal{O}$ is {\it contextual}.   
 Non-contextuality corresponds to the classical notion 
 that the value obtained via measurement is independent of other compatible observables 
 measured alongside it.    If an observable $o$ is non-contextual, then it is value definite. 
 However,  this is false for sets of observables:  $\mathcal{O}$ can be non-contextual but not 
 value definite if $\mathcal{O}$ contains an observable that is value indefinite.

An assignment function $\nu$ is {\it admissible} if the following two conditions hold for $C \in \mathcal{C}$:

\begin{itemize}

\item[1.]{  
                   If there exists an $o \in C$ with $\nu(o, C) = 1$, then $\nu(o^{\prime}, C) = 0$ 
                   for all $o^{\prime} \in C - \{ o\}$.
                }

\smallskip 

\item[2.]{  
                   If there exists an $o \in C$ such that $\nu(o^{\prime}, C) = 0$, 
                   for all $o^{\prime} \in C - \{ o\}$,  then $\nu(o, C) = 1$ 
                }

\end{itemize}

\noindent {\it Admissibility} is analogous to a {\it two-valued measure} used in quantum logic 
\cite{alda1,alda2,zierler,kalmbach,svozil_qlogic}.  
These definitions lead us to the ACS corollary which helps design a protocol
for a quantum random bit generator that relies upon value indefiniteness.

\medskip 

\noindent  {\bf ACS Corollary}.  \hskip 1pc  
{\it  Let $\ket{a}$, $\ket{b}$ in $\mathbb{C}^3$ be unit vectors such 
that $\sqrt{ \frac{5}{14} } \le |\braket{ a | b}| \le \frac{3}{\sqrt{14}}$.
Then there exists a set of projection observables $\mathcal{O}$ containing 
$P_{a}$ and $P_{b}$ and a set of contexts $\mathcal{C}$ over 
 $\mathcal{O}$ such that there is no admissible assignment function under 
 which $\mathcal{O}$ is non-contextual, $P_a$ has the value 1 and 
 $P_b$ is value definite.

 }

\medskip

The ACS experimental protocol starts with a spin-1 source and consists of two sequential measurements.  
Spin-1 particles are prepared in the  $S_z = 0$ state.  (From their eigenstate assumption, 
this operator has a definite value.)  Specifically, the first measurement puts the 
particle in the $S_z = 0$  eigenstate of the spin operator $S_z$.   
Since the preparation state is an eigenstate of the $S_x = 0$ projector observable 
with eigenvalue 0, this outcome has a  definite value and theoretically cannot be reached.  
The second measurement is performed in the eigenbasis of the $S_x$ operator with two 
outcomes $S_x = \pm 1$.  

\begin{figure}[h]\label{fig:beam_splitter}
   \centering
   \includegraphics[width=12 cm]{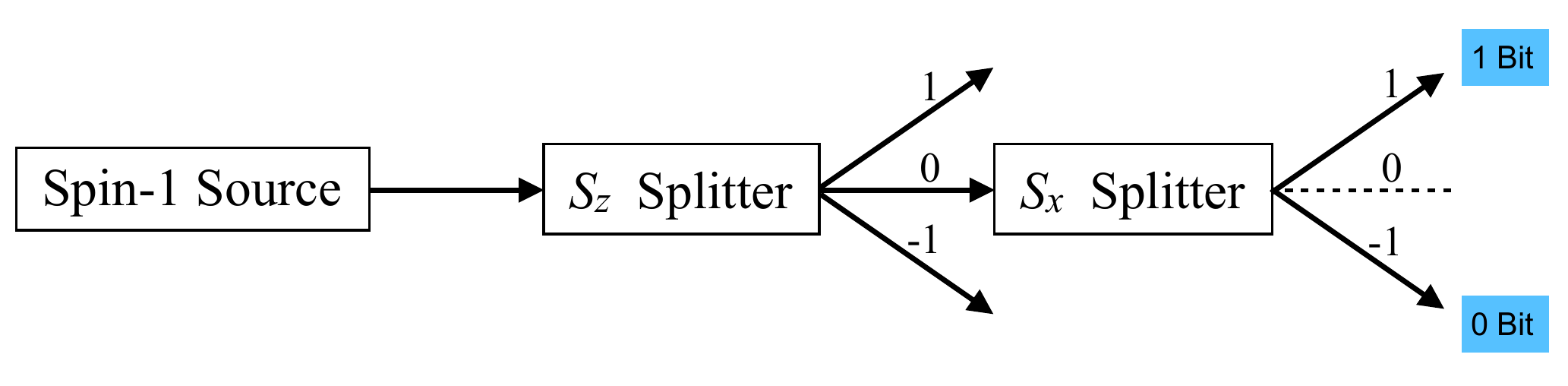}
   \caption[width=12 cm]{   \hskip 0.1pc  Quantum Observables Rendering Random Bits via Value Indefiniteness.  }  
\end{figure}

The $S_x = \pm 1$ outcomes can be assigned 0 and 1, respectively.  Moreover, since 
$\braket{ S_z = 0 | S_x = \pm 1} = \frac{1}{\sqrt{2}}$, the ACS corollary implies that 
neither of the $S_x = \pm 1$ outcomes can have a pre-determined definite value. 
This ACS design provides two key properties:  
(1) Bits 0 and 1 are generated by value indefiniteness. 
(2) Bits 0 and 1 are generated independently with a 50/50 probability.  
Quantum random axioms \ref{axiom_qr1}  and  \ref{axiom_qr2}  only require the 
second property.  It is worth noting, however, that the first property (value indefiniteness) 
helps sharpen the results in section  \ref{msf:x_machine_languages}.


In \cite{kulikov}, a quantum random number generator was built, empirically tested and 
implemented the ACS protocol that is shown in figure 2.  
During testing, the bias of the bits showed a  $50.001\%$ 
mean frequency of obtaining a 0 outcome and a standard deviation of $0.1\%$ 
that is consistent with a bucket size of 999302 bits.  
The entropy for unbiased random numbers obtained from 10 Gigabits of raw data 
was 7.999999 per byte.  The ideal value is, of course,  8.  
The data passed all standard NIST and Diehard statistical test suites \cite{NIST,diehard}. 
Furthermore, the quantum random bits were analyzed with a test \cite{calude_qr2017} 
directly related to algorithmic randomness:  the raw bits generated 
from \cite{kulikov} were used to test the primality of Carmichael numbers smaller 
than $54 \times 10^7$ with the Solovay-Strassen 
probabilistic algorithm.  The metric is the minimum number of random bits needed to 
confirm compositeness.  Ten sequences of raw quantum random bits of length $2^{29}$ 
demonstrated a significant advantage over sequences of the same length, produced 
from three modern pseudo-random generators -- Random123, PCG and xoroshilro128+.


In \cite{calude_qr2014}, Abbott, Calude and Svozil show that the occurrence of quantum randomness 
due to quantum indeterminacy is not a fluke.   They show that the breakdown of non-contextual hidden variable theories  
occurs almost everywhere and prove that quantum indeterminacy 
(i.e. contextuality or value indefiniteness) is a global property.  
They prove that after one arbitrary observable is fixed so that it occurs with certainty, 
almost all (Lebesgue measure one) remaining observables are value indefinite.


\section{Computing Ex-Machine Languages }\label{msf:x_machine_languages}

A class of ex-machines are defined as evolutions of the fundamental ex-machine 
$\mathfrak{Q}(x)$, whose 15 initial instructions are listed 
under ex-machine \ref{machine:meta_a_machine}.  These ex-machines compute 
languages $L$ that are subsets of 
$\{$\verb|a|${\}}^* = \{${\verb|a|}$^n : n \in \mathbb{N} \}$.  
The expression {\verb|a|}$^n$ represents a string of $n$ consecutive \verb|a|'s.  
For example, {\verb|a|}$^5$ $=$ \verb|aaaaa|  and  {\verb|a|}$^0$ is the empty string.


Define the set of languages $$\mathfrak{L} = {\underset{L \subset \{a\}^*} \bigcup } \mbox{\hskip 0.15pc} \{L\}.$$ 
 For each function $f: \mathbb{N} \rightarrow \{0, 1\}$, 
definition \ref{defn:language_L_f}  defines a unique language in  $\mathfrak{L}$.

\begin{defn}\label{defn:language_L_f}  \hskip 1pc  {\it Language $L_f$ }

\smallskip 

\noindent Consider any function $f: \mathbb{N} \rightarrow \{0, 1\}$.  
This means $f$ is a member of the set $\{0, 1\}^{\mathbb{N}}$.  
Function $f$ induces the language $L_f = \{${\verb|a|}$^n : f(n) = 1\}$.  
In other words, for each non-negative integer $n$, string  {\verb|a|}$^n$ 
is in the language $L_f$ if and only if $f(n) = 1$.  

\end{defn}

\noindent Trivially, $L_f$ is a language in $\mathfrak{L}$.  Moreover, these functions $f$ generate all of $\mathfrak{L}$.  

\smallskip

\begin{rem}  \hskip 1pc 
     $\mathfrak{L} = {\underset{f \in \{0, 1\}^{\mathbb{N}} } \bigcup  }  \mbox{\hskip 0.2pc} \{ L_f \}$
\end{rem}

In order to define the halting syntax for the language in $\mathfrak{L}$ that an ex-machine computes, 
choose alphabet set $A =$ $\{$\verb|#|, \verb|0|, \verb|1|, \verb|N|, \verb|Y|, \verb|a|$\}$.

\begin{defn}\label{defn:language_computed} \hskip 1pc 
{\it Language $L$ in $\mathfrak{L}$ that ex-machine $\mathfrak{X}$ computes}

\smallskip  

\noindent  Let $\mathfrak{X}$ be an ex-machine.  
The language $L$ in $\mathfrak{L}$ that $\mathfrak{X}$ computes 
is defined as follows.  A valid initial tape has the form 
\verb|#| \hskip 0.5pc \verb|#|{\verb|a|}$^n$\verb|#|.  
The valid initial tape  \verb|#| \hskip 0.5pc \verb|##| represents the empty string.  
After machine  $\mathfrak{X}$ starts executing with initial tape  
 \verb|#| \hskip 0.5pc \verb|#|{\verb|a|}$^n$\verb|#|,  \hskip 0.2pc
 string {\verb|a|}$^n$ is in $\mathfrak{X}$'s language 
 if ex-machine $\mathfrak{X}  $ halts with tape \verb|#|{\verb|a|}$^n$\verb|#| \hskip 0.5pc \verb|Y#|.  
 \hskip 0.3pc String {\verb|a|}$^n$ is not in $\mathfrak{X}$'s language if $\mathfrak{X}$ halts with tape 
 \verb|#|{\verb|a|}$^n$\verb|#| \hskip 0.5pc \verb|N#|.    

\end{defn}

\noindent  The use of special alphabet symbols \verb|Y| and \verb|N| -- to decide whether {\verb|a|}$^n$ 
is in the language or not in the language -- follows \cite{lewis}.

For a particular string \verb|#| \hskip 0.5pc \verb|#|{\verb|a|}$^m$\verb|# |,  some ex-machine $\mathfrak{X}$ 
could first halt with   \verb|#|{\verb|a|}$^m$\verb|#| \hskip 0.5pc \verb|N#| and in a second computation 
with input \verb|#| \hskip 0.5pc \verb|#|{\verb|a|}$^m$\verb|#| could halt with 
\verb|#|{\verb|a|}$^m$\verb|#| \hskip 0.5pc \verb|Y#|.  This oscillation of halting outputs could 
continue indefinitely and in some cases the oscillation can be aperiodic.  In this case, 
$\mathfrak{X}$'s language would not be well-defined according to definition \ref{defn:language_computed}.  
These types of ex-machines will be avoided in this paper.

There is a subtle difference between $\mathfrak{Q}(x)$ and an ex-machine 
$\mathfrak{X}$ whose halting output never stabilizes.  In contrast to the Turing machine, two different 
instances of the ex-machine $\mathfrak{Q}(x)$ can evolve to two different machines and compute distinct 
languages according to definition \ref{defn:language_computed}.   However, 
after  $\mathfrak{Q}(x)$ has evolved to a new machine $\mathfrak{Q}(a_0 a_1 \dots a_m$ $x)$ as a result of 
a prior execution with input tape \verb|#| \hskip 0.5pc \verb|#|{\verb|a|}$^m$\verb|#|, then 
for each $i$ with $0 \le i \le m$, machine  
$\mathfrak{Q}(a_0 a_1 \dots a_m$ $x)$ always halts with the same output when 
presented with input tape \verb|# #|{\verb|a|}$^i$\verb|#|.  In other words,  
$\mathfrak{Q}(a_0 a_1 \dots a_m$ $x)$'s halting output stabilizes on all input strings  
{\verb|a|}$^i$ where $0 \le i \le m$.   
Furthermore, it is the ability of $\mathfrak{Q}(x)$ to exploit the non-autonomous 
behavior of its two quantum random instructions that enables an evolution of  
$\mathfrak{Q}(x)$ to compute languages that are Turing incomputable.


We designed ex-machines that compute subsets of  $\{$\verb|a|${\}}^*$  
rather than subsets of $\{0, 1\}^*$  because the resulting specification of $\mathfrak{Q}(x)$ 
is much simpler and more elegant.  It is straightforward to list a standard machine 
that bijectively translates each {\verb|a|}$^n$  to a binary string in $\{0, 1\}^*$ as follows.  
The  empty string in $\{$\verb|a|${\}}^*$ maps to the  
empty string in $\{0, 1\}^*$.  Let $\psi$ represent this translation map.  
Hence,   \verb|a| ${\overset{\psi} \rightarrow}$ \verb|0|, 
\hskip 0.3pc  \verb|aa|  ${\overset{\psi} \rightarrow}$  \verb|1|,  
\hskip 0.3pc  \verb|aaa| ${\overset{\psi} \rightarrow}$  \verb|00|,  
\hskip 0.3pc  {\verb|a|}$^4$ ${\overset{\psi} \rightarrow}$  \verb|01|,  
\hskip 0.3pc  {\verb|a|}$^5$ ${\overset{\psi} \rightarrow}$  \verb|10|,  
\hskip 0.3pc  {\verb|a|}$^6$ ${\overset{\psi} \rightarrow}$  \verb|11|, 
\hskip 0.3pc  {\verb|a|}$^7$ ${\overset{\psi} \rightarrow}$  \verb|000|, and so on.     
Similarly, an inverse translation standard machine computes the inverse of $\psi$.
 Hence \verb|0| ${\overset{\psi^{-1}} \rightarrow}$  \verb|a|, 
 \hskip 0.3pc  \verb|1| ${\overset{\psi^{-1}} \rightarrow}$  \verb|aa|, 
 \hskip 0.3pc  \verb|00| ${\overset{\psi^{-1}} \rightarrow}$  \verb|aaa|, 
 and so on.  
The translation and inverse translation computations immediately transfer any results about the 
ex-machine computation of subsets of $\{$\verb|a|${\}}^*$ to corresponding subsets 
of $\{0, 1\}^*$ via $\psi$.  In particular, the following remark is relevant for our discussion.

\smallskip  

\begin{rem}
 
 Every subset of  $\{$\verb|a|${\}}^*$ is computable by some ex-machine if and only if 
 every subset of  $\{0, 1\}^*$  is computable by some ex-machine.  

\end{rem}

\begin{proof}
The remark immediately follows from the fact that the translation map $\psi$ and the inverse translation map $\psi^{-1}$ are 
computable with a standard machine.  
\end{proof}

\smallskip 





When the quantum randomness in  $\mathfrak{Q}$'s two quantum random instructions 
satisfy axiom \ref{axiom_qr1} (unbiased bernoulli trials) and axiom 
 \ref{axiom_qr2}  (stochastic independence), for each $n \in \mathbb{N}$, all $2^n$ finite paths of 
 length $n$ in the infinite, binary tree of figure 3 
  are equally likely.  (Feller  \cite{feller_vol1,feller_vol2} covers random walks.) 
  Moreover, there is a one-to-one correspondence between a function $f: \mathbb{N} \rightarrow \{0, 1\}$ and an 
infinite downward path in the infinite binary tree of figure 3.  
The beginning of an infinite downward path is shown in red, and starts as $(0, 1, 1, 0 \dots)$.

\begin{figure}[h]
   \label{fig:infinite_binary_tree}
   \centering
   \includegraphics[width=11 cm]{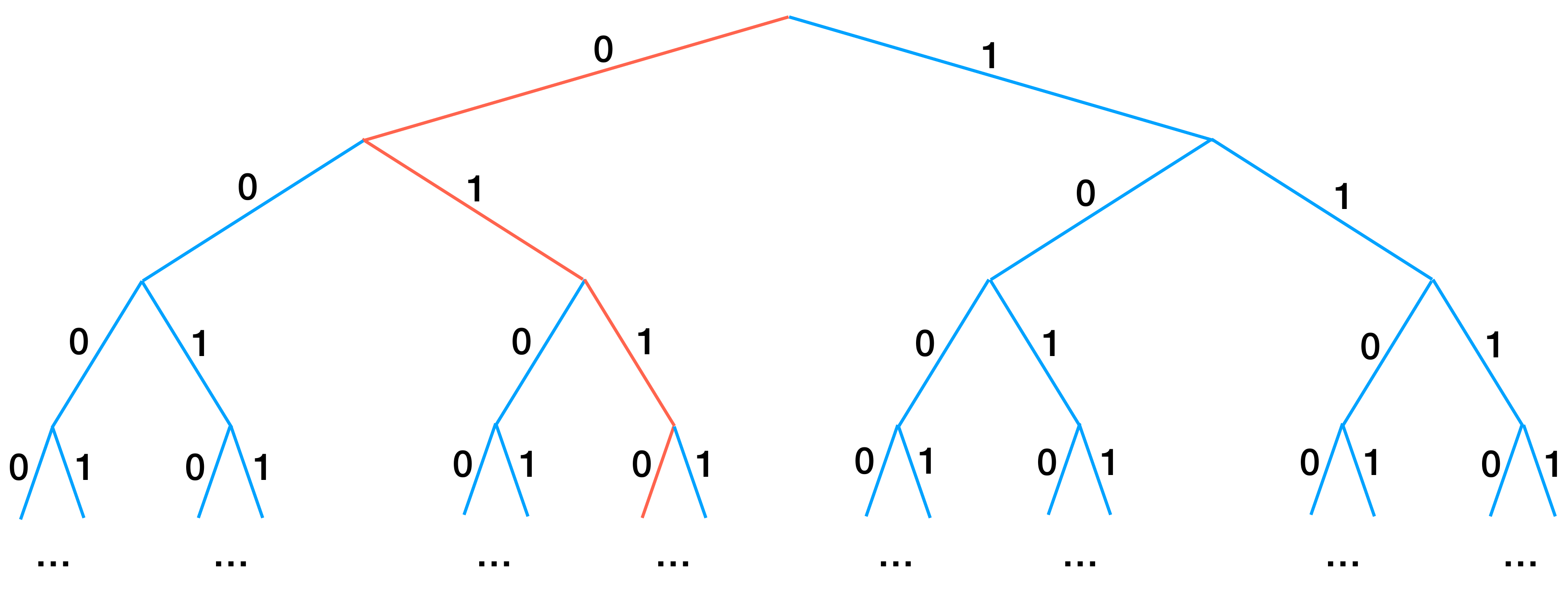}
   \caption{Infinite Binary Tree.  A Graphical Representation of $\{0, 1\}^{\mathbb{N}}$.  }
\end{figure}

Based on this one-to-one correspondence between functions  $f: \mathbb{N} \rightarrow \{0, 1\}$ 
and downward paths in the infinite binary tree, an examination of 
$\mathfrak{Q}(x)$'s execution behavior will show that $\mathfrak{Q}(x)$ can evolve to compute any language 
$L_f$ in $\mathfrak{L}$ when quantum random instructions {\small \verb|(x, #, x, 0)| } and
{\small \verb|(x, a, t, 0)|} satisfy axioms \ref{axiom_qr1} and \ref{axiom_qr2}.

 
\smallskip 

\begin{Machine}\label{machine:meta_a_machine}  \hskip 1pc    $\mathfrak{Q}(x)$

\smallskip 

\noindent {\small $A =$ $\{$\verb|#|, \verb|0|, \verb|1|, \verb|N|, \verb|Y|, \verb|a|$\}$}.  
States  
 {\small $Q =$ $\{ \verb|0|, \verb|h|,  \verb|n|, \verb|y|, \verb|t|, \verb|v|, \verb|w|, \verb|x|, \verb|8| \}$ }
where halting state {\small \verb|h| $= 1$}, 
and states {\small \verb|n| $= 2$, \verb|y| $= 3$, \verb|t| $= 4$, 
\verb|v| $= 5$, \verb|w| $= 6$, \verb|x| $= 7$}.   
The initial state is always $0$.  The letters represent states instead of explicit numbers 
because these states have special purposes.  (Letters are used solely for the reader's benefit.)
State $n$ indicates NO that the string is not in the language.   
State $y$ indicates YES that the string is in the language.    
State $x$ is used to generate a new random bit;  
this random bit determines  the string corresponding to the current value of $|Q|-1$. 
The fifteen instructions of $\mathfrak{Q}(x)$  are shown below.  

\end{Machine}


{ \footnotesize
  \begin{verbatim} 
    (0, #, 8, #, 1)
    (8, #, x, #, 0)
    (y, #, h, Y, 0)
    (n, #, h, N, 0)

    (x, #, x, 0)
    (x, a, t, 0)
    
    (x, 0, v, #, 0,  (|Q|-1, #, n, #, 1) )
    (x, 1, w, #, 0,  (|Q|-1, #, y, #, 1) )

    (t, 0, w, a, 0,  (|Q|-1, #, n, #, 1) )
    (t, 1, w, a, 0,  (|Q|-1, #, y, #, 1) )

    (v, #, n, #, 1,  (|Q|-1, a, |Q|, a, 1) )

    (w, #, y, #, 1, (|Q|-1, a, |Q|, a, 1) )
    (w, a, |Q|, a, 1, (|Q|-1, a, |Q|, a, 1) )

    (|Q|-1, a, x, a, 0)
    (|Q|-1, #, x, #, 0)
  \end{verbatim}
}

With initial state 0 and initial tape \verb|# #aaaa##|, an execution of machine $\mathfrak{Q}(x)$ 
is shown below.  


\bigskip 

{ \scriptsize
  \begin{verbatim} 
STATE   TAPE         HEAD   INSTRUCTION EXECUTED                        NEW INSTRUCTION                   
  8     ## aaaa###    1     (0, #, 8, #, 1)   
  x     ## aaaa###    1     (8, a, x, a, 0)   
  t     ## 1aaa###    1     (x, a, t, 1_qr, 0)   
  w     ## aaaa###    1     (t, 1, w, a, 0, (|Q|-1, #, y, #, 1))        (8, #, y, #, 1)
  9     ##a aaa###    2     (w, a, |Q|, a, 1, (|Q|-1, a, |Q|, a, 1))    (8, a, 9, a, 1)
  x     ##a aaa###    2     (9, a, x, a, 0)                             (9, a, x, a, 0)
  t     ##a 1aa###    2     (x, a, t, 1_qr, 0)   
  w     ##a aaa###    2     (t, 1, w, a, 0, (|Q|-1, #, y, #, 1) )       (9, #, y, #, 1)
 10     ##aa aa###    3     (w, a, |Q|, a, 1, (|Q|-1, a, |Q|, a, 1))    (9, a, 10, a, 1)
  x     ##aa aa###    3     (10, a, x, a, 0)                            (10, a, x, a, 0)   
  t     ##aa 0a###    3     (x, a, t, 0_qr, 0)   
  w     ##aa aa###    3     (t, 0, w, a, 0, (|Q|-1, #, n, #, 1) )       (10, #, n, #, 1)
 11     ##aaa a###    4     (w, a, |Q|, a, 1, (|Q|-1, a, |Q|, a, 1))    (10, a, 11, a, 1)
  x     ##aaa a###    4     (11, a, x, a, 0)                            (11, a, x, a, 0)   
  t     ##aaa 1###    4     (x, a, t, 1_qr, 0)   
  w     ##aaa a###    4     (t, 1, w, a, 0, (|Q|-1, #, y, #, 1) )       (11, #, y, #, 1)
 12     ##aaaa ###    5     (w, a, |Q|, a, 1, (|Q|-1, a, |Q|, a, 1))    (11, a, 12, a, 1)
  x     ##aaaa ###    5     (12, #, x, #, 0)                            (12, #, x, #, 0) 
  x     ##aaaa 0##    5     (x, #, x, 0_qr, 0)   
  v     ##aaaa ###    5     (x, 0, v, #, 0, (|Q|-1, #, n, #, 1) )       (12, #, n, #, 1)
  n     ##aaaa# ##    6     (v, #, n, #, 1, (|Q|-1, a, |Q|, a, 1))      (12, a, 13, a, 1)
  h     ##aaaa# N#    6     (n, #, h, N, 0)   
  \end{verbatim}
}

During this execution, $\mathfrak{Q}(x)$ replaces instruction 
{\small \verb|(8, #, x, #, 0)|} with {\small \verb|(8, #, y, #, 1)|}.  
Meta instruction  
{\small \verb|(w, a, |$|$\verb|Q|$|$\verb|, a, 1, (|$|$\verb|Q|$|$\verb|-1, a, |$|$\verb|Q|$|$\verb|, a, 1) )| } 
executes and replaces {\small \verb|(8, a, x, a, 0)|} 
with new instruction 
{\small \verb|(8, a, 9, | 

\noindent \verb|a, 1)|}. 
Also, simple meta instruction  
{\small \verb|(|$|$\verb|Q|$|$\verb|-1, a, x, a, 0)|  }
temporarily added instructions 
{\small \verb|(9, a, x, a, 0)|},  
{\small \verb|(10, a, x, a, 0)|}, 
and 
{\small \verb|(11, a, x, a, 0)|}.   

Subsequently, these new instructions were replaced by  
{\small \verb| (9, a, 10, a, 1)|}, 
{\small \verb| (10, a, 11, a, 1)|}, and 
{\small \verb|(11, a, 12, | 

\noindent \verb|a, 1)|}, 
respectively.    Similarly, simple meta instruction 
{\small  \verb|(|$|$\verb|Q|$|$\verb|-1, #, x, #, 0)|} 
added instruction {\small \verb|(12, #, x, #, 0)|}
and this instruction was replaced by instruction 
{\small  \verb|(12, #, n, #, 1)|}.
Lastly, instructions 
{\small  \verb| (9, #, y, #, 1)|}, 
{\small  \verb| (10, #, | 

\noindent \verb|n, #, 1)|},
{\small  \verb|(11, #, y, #, 1)|}, and 
{\small  \verb|(12, a, 13,|} 
{\small  \verb|a, 1)|}  were added.

Furthermore, five new states 
{\small  \verb|9|}, {\small  \verb|10|}, {\small \verb|11|}, 
{\small \verb|12|}  and {\small  \verb|13|} were 
added to $Q$.  After this computation halts, the machine states are 
{\small  $Q =$ $\{$\verb|0|, \verb|h|,  \verb|n|, \verb|y|, \verb|t|, \verb|v|, \verb|w|, \verb|x|, \verb|8|, 
\verb|9|, \verb|10|, \verb|11|, \verb|12|, \verb|13|$\}$} and the resulting 
ex-machine evolved to  has 24 instructions.  It is called $\mathfrak{Q}(11010$ $x)$.


\begin{Machine}\label{instructions:Q_x_11010}  \hskip 1pc  $\mathfrak{Q}(11010$ $x)$

{  \footnotesize
   \begin{verbatim} 
    (0, #, 8, #, 1)
    
    (y, #, h, Y, 0)
    (n, #, h, N, 0)

    (x, #, x, 0)
    (x, a, t, 0)
    
    (x, 0, v, #, 0,  (|Q|-1, #, n, #, 1) )
    (x, 1, w, #, 0,  (|Q|-1, #, y, #, 1) )

    (t, 0, w, a, 0,  (|Q|-1, #, n, #, 1) )
    (t, 1, w, a, 0,  (|Q|-1, #, y, #, 1) )

    (v, #, n, #, 1,  (|Q|-1, a, |Q|, a, 1) )

    (w, #, y, #, 1, (|Q|-1, a, |Q|, a, 1) )
    (w, a, |Q|, a, 1, (|Q|-1, a, |Q|, a, 1) )

    (|Q|-1, a, x, a, 0)
    (|Q|-1, #, x, #, 0)

    (8, #, y, #, 1)
    (8, a, 9, a, 1)

    (9, #, y, #, 1)
    (9, a, 10, a, 1)

    (10, #, n, #, 1)
    (10, a, 11, a, 1)

    (11, #, y, #, 1)
    (11, a, 12, a, 1)

    (12, #, n, #, 1)
    (12, a, 13, a, 1)
  \end{verbatim}
}

\end{Machine}

\noindent   New instructions {\small  \verb|(8, #, y, #, 1)| },
{\small  \verb|(9, #, y, #, 1)|}, and  {\small  \verb|(11, #, y, #, 1)|}  
help $\mathfrak{Q}(11010$ $x)$ compute that the 
empty string, {\small \verb|a| } and {\small \verb|aaa| } are in its language, respectively.  
Similarly, the new instructions {\small  \verb|(10, #, n, #, 1)|} 
and 
{\small \verb|(12, #, n, #, 1)| }  
help  
$\mathfrak{Q}(11010$ $x)$ compute that \verb|aa| and 
 \verb|aaaa|   
are not in its language, respectively.

The zeroeth, first, and third $1$ in $\mathfrak{Q}(11010$ $x)$'s name indicate that 
the empty string,   \verb|a|  and   \verb|aaa| 
are in  $\mathfrak{Q}(11010$ $x)$'s language.  
The second and fourth 0 indicate strings  \verb|aa|  
and   \verb|aaaa|  are not in its language.  
The symbol $x$ indicates that all strings {\verb|a|}$^n$
with $n \ge 5$ have not yet been determined 
whether they are in $\mathfrak{Q}(11010$ $x)$'s language or not in its language.

Starting at state 0, ex-machine $\mathfrak{Q}(11010$ $x)$ computes that the empty string is 
in $\mathfrak{Q}(11010$ $x)$'s language.    

{ \scriptsize
  \begin{verbatim} 
STATE      TAPE        TAPE HEAD        INSTRUCTION EXECUTED  
  8        ## ###          1            (0, #, 8, #, 1)   
  y        ### ##          2            (8, #, y, #, 1)   
  h        ### Y#          2            (y, #, h, Y, 0)   
  \end{verbatim}
}

\noindent Starting at state 0, ex-machine $\mathfrak{Q}(11010$ $x)$ computes that string \verb|a| is  in 
$\mathfrak{Q}(11010$ $x)$'s language.    

{ \scriptsize
  \begin{verbatim} 
STATE      TAPE        TAPE HEAD        INSTRUCTION EXECUTED  
  8        ## a###         1            (0, #, 8, #, 1)   
  9        ##a ###         2            (8, a, 9, a, 1)   
  y        ##a# ##         3            (9, #, y, #, 1)   
  h        ##a# Y#         3            (y, #, h, Y, 0)   
  \end{verbatim}
}

\noindent Starting at state 0, $\mathfrak{Q}(11010$ $x)$ computes that string  \verb|aa| is not in 
$\mathfrak{Q}(11010$ $x)$'s language.

{ \scriptsize
  \begin{verbatim} 
STATE     TAPE          TAPE HEAD      INSTRUCTION EXECUTED                                                          
  8       ## aa###          1          (0, #, 8, #, 1)   
  9       ##a a###          2          (8, a, 9, a, 1)   
 10       ##aa ###          3          (9, a, 10, a, 1)   
  n       ##aa# ##          4          (10, #, n, #, 1)   
  h       ##aa# N#          4          (n, #, h, N, 0)   
  \end{verbatim}
}

\noindent Starting at state 0, $\mathfrak{Q}(11010$ $x)$ computes that   \verb|aaa|
is in $\mathfrak{Q}(11010$ $x)$'s language.

{ \scriptsize
  \begin{verbatim} 
STATE     TAPE          TAPE HEAD      INSTRUCTION EXECUTED                                                          
  8       ## aaa###         1          (0, #, 8, #, 1)   
  9       ##a aa###         2          (8, a, 9, a, 1)   
 10       ##aa a###         3          (9, a, 10, a, 1)   
 11       ##aaa ###         4          (10, a, 11, a, 1)   
  y       ##aaa# ##         5          (11, #, y, #, 1)   
  h       ##aaa# Y#         5          (y, #, h, Y, 0)   
  \end{verbatim}
}

\noindent  Starting at state 0, $\mathfrak{Q}(11010$ $x)$ computes that   \verb|aaaa|  
is not in $\mathfrak{Q}(11010$ $x)$'s language.

{ \scriptsize
  \begin{verbatim} 
STATE     TAPE          TAPE HEAD      INSTRUCTION EXECUTED                                                          
  8       ## aaaa####       1          (0, #, 8, #, 1)   
  9       ##a aaa####       2          (8, a, 9, a, 1)   
 10       ##aa aa####       3          (9, a, 10, a, 1)   
 11       ##aaa a####       4          (10, a, 11, a, 1)   
 12       ##aaaa ####       5          (11, a, 12, a, 1)   
  n       ##aaaa# ###       6          (12, #, n, #, 1)   
  h       ##aaaa# N##       6          (n, #, h, N, 0)   
  \end{verbatim}
}

Note that for each of these executions, no new states were added and no instructions 
were added or replaced. Thus, for all subsequent executions, ex-machine 
$\mathfrak{Q}(11010$ $x)$ computes that the empty string,   \verb|a| and 
 \verb|aaa|  are in its language.  Similarly, strings  \verb|aa| 
and  \verb|aaaa|
are not in $\mathfrak{Q}(11010$ $x)$'s language for all subsequent 
executions of $\mathfrak{Q}(11010$ $x)$.

Starting at state 0, we examine an execution of 
ex-machine $\mathfrak{Q}(11010$ $x)$ on input tape  \verb|# #aaaaaaa##|.

\medskip 

{  \scriptsize
   \begin{verbatim} 
STATE   TAPE            HEAD   INSTRUCTION EXECUTED                        NEW INSTRUCTION 
  8     ## aaaaaaa###     1    (0, #, 8, #, 1)   
  9     ##a aaaaaa###     2    (8, a, 9, a, 1)   
 10     ##aa aaaaa###     3    (9, a, 10, a, 1)   
 11     ##aaa aaaa###     4    (10, a, 11, a, 1)   
 12     ##aaaa aaa###     5    (11, a, 12, a, 1)   
 13     ##aaaaa aa###     6    (12, a, 13, a, 1)   
  x     ##aaaaa aa###     6    (13, a, x, a, 0)   
  t     ##aaaaa 0a###     6    (x, a, t, 0_qr, 0)   
  w     ##aaaaa aa###     6    (t, 0, w, a, 0, (|Q|-1, #, n, #, 1) )      (13, #, n, #, 1)
 14     ##aaaaaa a###     7    (w, a, |Q|, a, 1, (|Q|-1, a, |Q|, a, 1))   (13, a, 14, a, 1)
  x     ##aaaaaa a###     7    (14, a, x, a, 0)                           (14, a, x, a, 0)  
  t     ##aaaaaa 1###     7    (x, a, t, 1_qr, 0)   
  w     ##aaaaaa a###     7    (t, 1, w, a, 0, (|Q|-1, #, y, #, 1) )      (14, #, y, #, 1)
 15     ##aaaaaaa ###     8    (w, a, |Q|, a, 1, (|Q|-1, a, |Q|, a, 1))   (14, a, 15, a, 1)
  x     ##aaaaaaa ###     8    (15, #, x, #, 0)                           (15, #, x, #, 0)  
  x     ##aaaaaaa 1##     8    (x, #, x, 1_qr, 0)   
  w     ##aaaaaaa ###     8    (x, 1, w, #, 0, (|Q|-1, #, y, #, 1))       (15, #, y, #, 1)
  y     ##aaaaaaa# ##     9    (w, #, y, #, 1, (|Q|-1, a, |Q|, a, 1))     (15, a, 16, a, 1)
  h     ##aaaaaaa# Y#     9    (y, #, h, Y, 0)   
   \end{verbatim}
}

\noindent Overall, during this execution ex-machine $\mathfrak{Q}(11010$ $x)$ evolved to 
ex-machine $\mathfrak{Q}(11010$ $011$ $x)$.  Three quantum random instructions were executed.   
The first quantum random instruction  {\small \verb|(x, a, t, 0)|} measured a 0 so it is shown above as 
{\small \verb|(x, a, t, 0_qr, 0)|}.  The result of this 0 bit measurement adds the instruction 
{\small \verb|(13, #, n, #, 1)|}, so that in all subsequent executions of ex-machine 
$\mathfrak{Q}(11010$ $011$ $x)$, string {\verb|a|}${^5}$ is not in  $\mathfrak{Q}(11010$ $011$ $x)$'s language.  
Similarly, the second quantum random instruction {\small  \verb|(x, a, t, 0)|} measured a 1 so it is shown above as 
{\small \verb|(x, a, t, 1_qr, 0)|}.  The result of this 1 bit measurement adds the instruction 
{\small \verb|(14, #, y, #, 1)|}, so that in all subsequent executions, 
 string {\verb|a|}${^6}$ is in  $\mathfrak{Q}(11010$ $011$ $x)$'s language.  
 Finally, the third quantum random instruction {\small \verb|(x, #, x, 0)|} 
 measured a 1 so it is shown above as 
{\small \verb|(x, #, x, 1_qr, 0)|}.   The result of this 1 bit measurement adds the instruction 
{\small \verb|(15, #, y, #, 1)|}, so that in all subsequent executions, string {\verb|a|}${^7}$ is in  
$\mathfrak{Q}(11010$ $011$ $x)$'s language.  

\smallskip   

 Lastly, starting at state 0, we examine a distinct execution of $\mathfrak{Q}(11010$ $x)$ on 
 input tape \verb|# #aaaaaaa##|.  A distinct execution of $\mathfrak{Q}(11010$ $x)$ evolves to ex-machine  
  $\mathfrak{Q}(11010$ $000$ $x)$.


{ \scriptsize
  \begin{verbatim} 
STATE  TAPE           HEAD   INSTRUCTION EXECUTED                         NEW INSTRUCTION 
  8    ## aaaaaaa###    1    (0, #, 8, #, 1)   
  9    ##a aaaaaa###    2    (8, a, 9, a, 1)   
 10    ##aa aaaaa###    3    (9, a, 10, a, 1)   
 11    ##aaa aaaa###    4    (10, a, 11, a, 1)   
 12    ##aaaa aaa###    5    (11, a, 12, a, 1)   
 13    ##aaaaa aa###    6    (12, a, 13, a, 1)   
  x    ##aaaaa aa###    6    (13, a, x, a, 0)   
  t    ##aaaaa 0a###    6    (x, a, t, 0_qr, 0)   
  w    ##aaaaa aa###    6    (t, 0, w, a, 0, (|Q|-1, #, n, #, 1))        (13, #, n, #, 1)
 14    ##aaaaaa a###    7    (w, a, |Q|, a, 1, (|Q|-1, a, |Q|, a, 1))    (13, a, 14, a, 1)
  x    ##aaaaaa a###    7    (14, a, x, a, 0)                            (14, a, x, a, 0) 
  t    ##aaaaaa 0###    7    (x, a, t, 0_qr, 0)   
  w    ##aaaaaa a###    7    (t, 0, w, a, 0, (|Q|-1, #, n, #, 1))        (14, #, n, #, 1)
 15    ##aaaaaaa ###    8    (w, a, |Q|, a, 1, (|Q|-1, a, |Q|, a, 1))    (14, a, 15, a, 1)
  x    ##aaaaaaa ###    8    (15, #, x, #, 0)                            (15, #, x, #, 0)
  x    ##aaaaaaa 0##    8    (x, #, x, 0_qr, 0)   
  v    ##aaaaaaa ###    8    (x, 0, v, #, 0, (|Q|-1, #, n, #, 1))        (15, #, n, #, 1)
  n    ##aaaaaaa# ##    9    (v, #, n, #, 1, (|Q|-1, a, |Q|, a, 1))      (15, a, 16, a, 1)
  h    ##aaaaaaa# N#    9    (n, #, h, N, 0)   
  \end{verbatim}
}

\medskip 

Based on our previous examination of ex-machine $\mathfrak{Q}(x)$ evolving to $\mathfrak{Q}(11010$  $x)$
and then subsequently $\mathfrak{Q}(11010$ $x)$ evolving to $\mathfrak{Q}(11010$ $011$ $x)$, ex-machine  
\ref{defn:Q_a0_a1_am_listing}  specifies $\mathfrak{Q}(a_0 a_1 \dots a_m$ $x)$ 
in terms of initial states and  initial instructions.  

\medskip 

\begin{Machine}\label{defn:Q_a0_a1_am_listing}  \hskip 1pc  $\mathfrak{Q}(a_0 a_1 \dots a_m$ $x)$

\smallskip  

\noindent  Let $m \in \mathbb{N}$.  
 Set $Q = \{$\verb|0|, \verb|h|,  \verb|n|, \verb|y|, \verb|t|, \verb|v|, \verb|w|, \verb|x|, \verb|8|, 
             \verb|9|, \verb|10|,  $\dots$ $m+8, m+9$ $\}$.  For $0 \le i \le m$, each $a_i$ is 0 or 1.  
ex-machine $\mathfrak{Q}(a_0 a_1 \dots a_m$ $x)$'s instructions are shown below. 
Symbol $b_8 =$ \verb|y | if $a_0 = 1$.  Otherwise, symbol 
    $b_8 =$ \verb|n | if $a_0 = 0$.  Similarly, symbol 
    $b_9 =$ \verb|y | if $a_1 = 1$.  Otherwise, symbol 
    $b_9 =$ \verb|n | if $a_1 = 0$.  
And so on until reaching the second to the last instruction 
{\small  \verb|(|$m+8$\verb|, #,| $b_{m+8}$\verb|, #, 1)|}, 
 symbol $b_{m+8} =$ \verb|y | if $a_m = 1$.  
Otherwise, symbol $b_{m+8} =$ \verb|n | if $a_m = 0$.

\end{Machine}


{ \footnotesize
  \begin{verbatim} 
     (0, #, 8, #, 1)
    
     (y, #, h, Y, 0)
     (n, #, h, N, 0)

     (x, #, x, 0)
     (x, a, t, 0)
    
     (x, 0, v, #, 0,  (|Q|-1, #, n, #, 1) )
     (x, 1, w, #, 0,  (|Q|-1, #, y, #, 1) )

     (t, 0, w, a, 0,  (|Q|-1, #, n, #, 1) )
     (t, 1, w, a, 0,  (|Q|-1, #, y, #, 1) )

     (v, #, n, #, 1,  (|Q|-1, a, |Q|, a, 1) )

     (w, #, y, #, 1, (|Q|-1, a, |Q|, a, 1) )
     (w, a, |Q|, a, 1, (|Q|-1, a, |Q|, a, 1) )

     (|Q|-1, a, x, a, 0)
     (|Q|-1, #, x, #, 0)
\end{verbatim}

\medskip 

    \verb|  (8, #,| $b_{8}$\verb|, #, 1)|

    \verb|  (8, a, 9, a, 1)|

    \medskip 

    \verb|  (9, #,| $b_{9}$\verb|, #, 1)|

    \verb|  (9, a, 10, a, 1)|

    \medskip 

    \verb|  (10, #,| $b_{10}$\verb|, #, 1)|

    \verb|  (10, a, 11, a, 1)|
}

\medskip 

 \hskip  0.8pc    . . .

 \medskip 

{ \footnotesize

    \verb|  (|$i+8$\verb|, #,| $b_{i+8}$\verb|, #, 1)|

    \verb|  (|$i+8$\verb|, a,| $i+9$\verb|, a, 1)|

}

\medskip

\hskip 0.8pc    . . .

\medskip 

{ \footnotesize

    \verb|  (|$m+7$\verb|, #,| $b_{m+7}$\verb|, #, 1)|

    \verb|  (|$m+7$\verb|, a,| $m+8$\verb|, a, 1)|

    \medskip 

    \verb|  (|$m+8$\verb|, #,| $b_{m+8}$\verb|, #, 1)|

    \verb|  (|$m+8$\verb|, a,| $m+9$\verb|, a, 1)|

}



\begin{lem}\label{lemma:Q_a0_a1_am}

\noindent Whenever $i$ satisfies $0 \le i \le m$, string 
{\small  {\verb|a|}$^{i}$ } is in 
$\mathfrak{Q}(a_0 a_1 \dots a_m$ $x)$'s language 
if $a_i = 1$;  string {\small  {\verb|a|}$^{i}$ } is not in $\mathfrak{Q}(a_0 a_1 \dots a_m$ $x)$'s language if $a_i = 0$.  
Whenever $n > m$, it has not yet been determined whether string 
{\small  {\verb|a|}$^{n}$ }
is in 
$\mathfrak{Q}(a_0 a_1 \dots a_m$ $x)$'s language or not in its language.

\end{lem}

\begin{proof}[Proof]
         When $0 \le i \le m$,  the first consequence follows immediately from the definition of
          {\small  {\verb|a|}$^{i}$ } being in 
          $\mathfrak{Q}(a_0 a_1 \dots a_m$ $x)$'s language and from ex-machine \ref{defn:Q_a0_a1_am_listing}.
           In instruction  
           {\small   \verb|(|$i+8$\verb|, #,| $b_{i+8}$\verb|, #, 1)| }   
           the state value of $b_{i+8}$ is 
            \verb|y|  if $a_i = 1$ and $b_{i+8}$ is 
            \verb|n|  if $a_i = 0$.


          For the indeterminacy of strings {\small  {\verb|a|}$^{n}$ } when $n > m$, ex-machine $\mathfrak{Q}(a_0 \dots a_m$ $x)$
          executes its last instruction
          {\small  \verb|(|$m+8$\verb|, a,| $m+9$\verb|, a, 1)| } when it is scanning the $m$th 
          {\small \verb|a| } in 
          {\small {\verb|a|}$^{n}$}.  Subsequently, for each {\small  \verb|a|} on the tape to the right of 
          {\small  \verb| #|{\verb|a|}$^m$,\verb| |
          } ex-machine $\mathfrak{Q}(a_0  \dots a_m$ $x)$ 
          executes the quantum random instruction {\small  \verb|(x, a, t, 0)|}.

          If the execution of {\small \verb|(x, a, t, 0)| } measures a {\small  \verb|0|}, the two meta instructions 
          {\small   \verb|(t, 0, w, a, 0, (|$|$\verb|Q|$|$\verb|-1, #, n, #, 1) )| }
          and
          {\small   \verb|(w, a, |$|$\verb|Q|$|$\verb|, a, 1, (|$|$\verb|Q|$|$\verb|-1, a, |$|$\verb|Q|$|$\verb|,|  
                    \verb|a, 1) )| 
          }
          are executed.  If the next alphabet symbol to the right is an \verb|a|, then a  
          new standard instruction is executed that is instantiated from the 
          simple meta instruction 
          {\small  \verb|(|$|$\verb|Q|$|$\verb|-1, a, x, a, 0)| }.   
          If the tape head was scanning the last {\small  \verb|a| } in 
          {\small  {\verb|a|}$^{n}$}, 
          then a new standard instruction is executed that is instantiated from 
          the simple meta instruction 
          {\small  \verb|(|$|$\verb|Q|$|$\verb|-1, #, x, #, 0)| }.

          If the execution of {\small \verb|(x, a, t, 0)| } measures a {\small  \verb|1|}, the two meta instructions 
          {\small  \verb|(t, 1, w, a, 0, (|$|$\verb|Q|$|$\verb|-1, #, y, #, 1) )| } and  
          {\small  \verb|(w, a, |$|$\verb|Q|$|$\verb|, a, 1, (|$|$\verb|Q|$|$\verb|-1, a, |$|$\verb|Q|$|$\verb|,|  \verb|a, 1) )| 
          }
          are executed.   If the next alphabet symbol to the right is an 
          {\small \verb|a|}, then a  
          new standard instruction is executed that is instantiated from the 
          simple meta instruction 
          {\small  \verb|(|$|$\verb|Q|$|$\verb|-1, a, x, a, 0)| }.   
          If the tape head was scanning the last {\small \verb|a| } in 
          {\small  {\verb|a|}$^{n}$}, 
          then a new standard instruction is executed that is instantiated from 
          the simple meta instruction 
          {\small  \verb|(|$|$\verb|Q|$|$\verb|-1, #, x, #, 0)|}.

          In this way, for each
          {\small  \verb|a|  } on the tape to the right of 
          {\small  \verb| #|{\verb|a|}$^m$,\verb| | }
          the execution of the quantum random instruction 
          {\small  \verb|(x, a, t, 0)|  } 
          determines whether each string
          {\small  {\verb|a|}$^{m+k}$}, satisfying $1 \le k \le n - m$,  
          is in or not in $\mathfrak{Q}(a_0 a_1 \dots a_n$ $x)$'s language.

          After the execution of 
          {\small  \verb|(|$|$\verb|Q|$|$\verb|-1, #, x, #, 0)| }, 
          the tape head is scanning a blank symbol, so the quantum random instruction
          {\small  \verb|(x, #, x, 0)| } is executed.  If a 
          {\small  \verb|0| } is measured by the quantum random source, 
          the meta instructions  
          {\small  \verb|(x, 0, v, #, 0, (|$|$\verb|Q|$|$\verb|-1, #, n, #, 1) )| } 
          and 
          {\small \verb|(v, #, n, #, 1, (|$|$\verb|Q|$|$\verb|-1, a, |$|$\verb|Q|$|$\verb|,| \verb| a,|  \verb|1) )| } 
          are executed.   
          Then the last instruction executed is  {\small  \verb|(n, #, h, N, 0)| } which indicates 
          that {\small {\verb|a|}$^{n}$} is not in  $\mathfrak{Q}(a_0 a_1 \dots a_n$ $x)$'s  language.

          If the execution of {\small  \verb|(x, #, x, 0)| } measures a 
          {\small \verb|1|}, the meta instructions
          {\small   \verb|(x, 1, w,| \verb|#, 0, (|$|$\verb|Q|$|$\verb|-1, #, y, #, 1) )|
          } and           
          {\small   \verb|(w, #, y, #, 1, (|$|$\verb|Q|$|$\verb|-1, a, |$|$\verb|Q|$|$\verb|,|
                    \verb| a, 1) )| 
          }
          are executed.   Then the last instruction executed is {\small  \verb|(y, #, h, Y, 0)| } which 
          indicates that  {\small  {\verb|a|}$^{n}$ } is in $\mathfrak{Q}(a_0 a_1 \dots a_n$ $x)$'s language.

          During the execution of the instructions,  for each  {\small \verb|a| } on the tape to the right of  
          {\small  \verb| #|{\verb|a|}$^m$,  }
          $\mathfrak{Q}(a_0 a_1 \dots a_m$ $x)$ evolves to $\mathfrak{Q}(a_0 a_1 \dots a_n$ $x)$ 
          according to the specification in ex-machine \ref{defn:Q_a0_a1_am_listing}, where one 
          substitutes $n$ for $m$.
\end{proof}

\begin{rem}\label{rem:Q_x_constructible}

\smallskip

When the binary string $a_0 a_1 \dots a_m$ is presented as input, 
the ex-machine instructions for $\mathfrak{Q}(a_0 a_1 \dots a_m$ $x)$, specified in 
ex-machine \ref{defn:Q_a0_a1_am_listing}, are constructible (i.e., can be printed) with a standard machine.  
\end{rem}

\medskip  

\noindent  In contrast with lemma \ref{lemma:Q_a0_a1_am}, 
$\mathfrak{Q}(a_0 a_1 \dots a_m$ $x)$'s instructions 
 are not executable with a standard machine when the input tape 
\verb|# #a|$^i$\verb|#| satisfies $i > m$ because meta and quantum random instructions are required.  Thus, 
remark \ref{rem:Q_x_constructible}  distinguishes the construction of  
$\mathfrak{Q}(a_0 a_1 \dots a_m$ $x)$'s instructions from the execution 
of $\mathfrak{Q}(a_0 a_1 \dots a_m$ $x)$'s  instructions.  


\begin{proof}[Proof]
When given a finite list $(a_0$  \verb| | $a_1$  \verb| | $\dots$ \verb| | $a_m)$, where each $a_i$ is 0 or 1, 
the code listing below constructs $\mathfrak{Q}(a_0 a_1 \dots a_m$ $x)$'s instructions.   
Starting with comment \verb| ;;  Qx_builder.lsp|, \verb| | the code listing is 
expressed in a dialect \cite{mueller} of LISP.   The LISP 
language \cite{mccarthy_symbolic,mccarthy_LISP,mccarthy_math} originated from  
the lambda calculus  \cite{church_unsolvable,church_lambda,kleene}.  
The appendix of \cite{turing36} outlines a proof that the lambda calculus is computationally equivalent 
to a standard machine.  The following 3 instructions print the ex-machine instructions for  
$\mathfrak{Q}(11010$ $x)$, listed in ex-machine  \ref{instructions:Q_x_11010}.  

\medskip 

\noindent {\footnotesize  \verb|(set  'a0_a1_dots_am  (list 1 1 0 1 0) )  |  

\noindent \verb|(set  'Qx_machine  (build_Qx_machine  a0_a1_dots_am) )  |

\noindent \verb|(print_xmachine  Qx_machine) | 
}
\end{proof}



\smallskip 

{  \footnotesize  
\begin{verbatim}
;;  Qx_builder.lsp  
(define (make_qr_instruction  q_in  a_in  q_out   move)
   (list (string q_in) (string a_in)  (string q_out) (string move))  )

(define (make_instruction  q_in  a_in  q_out  a_out  move)
   (list (string q_in) (string a_in) 
         (string q_out) (string a_out) (string move)) )

(define (make_meta_instruction  q_in  a_in  q_out  a_out  move_standard  r_in  b_in  r_out  b_out  move_meta)
     (list (string q_in) (string a_in) 
           (string q_out) (string a_out) (string move_standard)   
           (make_instruction r_in b_in r_out b_out move_meta) )  )

(define  (initial_Qx_machine)
  (list
    (make_instruction  "0"  "#"  "8"  "#"  1)
    (make_instruction  "8"  "#"  "x"  "#"  0)

    (make_instruction  "y"  "#"  "h"  "Y"  0)   ;;  This means string a^i  is in the  language.  
    (make_instruction  "n"  "#"  "h"  "N"  0)   ;;  This means string a^i  is NOT in the language. 
      
    (make_qr_instruction  "x"  "#"  "x"  0)
    (make_qr_instruction  "x"  "a"  "t"  0)
  
    (make_meta_instruction "x"  "0"  "v"  "#"  0   "|Q|-1"  "#"  "n"  "#"  1)
    (make_meta_instruction "x"  "1"  "w"  "#"  0   "|Q|-1"  "#"  "y"  "#"  1)
     
    (make_meta_instruction "t"  "0"  "w"  "a"  0  "|Q|-1"  "#"  "n"  "#"  1)
    (make_meta_instruction "t"  "1"  "w"  "a"  0  "|Q|-1"  "#"  "y"  "#"  1)

    (make_meta_instruction "v"  "#"  "n"  "#"  1  "|Q|-1" "a" "|Q|" "a" 1)
    (make_meta_instruction "w"  "#"  "y"    "#"  1  "|Q|-1" "a" "|Q|" "a" 1)
    (make_meta_instruction "w"  "a"  "|Q|"  "a"  1  "|Q|-1" "a" "|Q|" "a" 1)
      
    (make_instruction "|Q|-1"  "a"  "x"  "a"  0)
    (make_instruction "|Q|-1"  "#"  "x"  "#"  0)
))

(define (add_instruction  instruction  q_list)
   (append  q_list  (list  instruction) )  )

(define  (check_a0_a1_dots_am   a0_a1_dots_am)
   (if (list?  a0_a1_dots_am)
       (dolist  (a_i  a0_a1_dots_am)   
          (if (member a_i  (list 0 1) )
              true
              (begin 
                 (println "ERROR! (build_Qx_machine a0_a1_dots_am). a_i = " a_i)
                 (exit)  
              )
          )
          a0_a1_dots_am  )
       nil   
))
\end{verbatim}

\begin{verbatim}
;;;;;;;;;;;;;;;;   BUILD  MACHINE  Q(a0 a1 . . .  am x)
;;   a0_a1_dots_am  has to be a list of 0's and 1's  or  '() 
(define  (build_Qx_machine  a0_a1_dots_am)
  (let 
    (  (Qx_machine  (initial_Qx_machine))
       (|Q|  8)
       (b_|Q|  nil)
       (ins1  nil)
       (ins2  nil)
    )
    (set  'check  (check_a0_a1_dots_am   a0_a1_dots_am)  )

    ;;  if nil OR check is an empty list, remove instruction  (8, #, x, #, 0)
     (if (or  (not check)  (empty? check) )
        true
        (set 'Qx_machine (append (list (Qx_machine 0)) (rest (rest Qx_machine))))
     )

     (if (list?  a0_a1_dots_am)
         (dolist  (a_i  a0_a1_dots_am)   
            (if (= a_i 1)
                (set  'b_|Q|  "y")
                (set  'b_|Q|  "n") )

            (set  'ins1        (make_instruction |Q| "#"  b_|Q| "#"  1) )
            (set  'Qx_machine  (add_instruction  ins1  Qx_machine) )

            (set  'ins2        (make_instruction |Q| "a"  (+ |Q| 1) "a"  1))
            (set  'Qx_machine  (add_instruction  ins2  Qx_machine) )
            (set '|Q|  (+ |Q| 1) )
         )
      )

      Qx_machine
))
\end{verbatim}


\begin{verbatim}
(define  (print_elements  instruction)
  (let
     (  (idx_ub   (min  4  (-  (length instruction) 1))  )
        (i  0)
     )
     
     (for (i  0  idx_ub)  
        (print  (instruction i) )
        (if  (< i idx_ub)  (print ", "))
     )
))

(define (print_instruction  instruction)
   (print "(")
   (if  (<= (length instruction) 5)
        (print_elements  instruction)
        (begin 
             (print_elements instruction)
             (print ", (")
             (print_elements (instruction 5) )
             (print ") ")  )  )
   (println ")")
)

(define  (print_xmachine  x_machine)
   (println)
   (dolist (instruction  x_machine)
       (print_instruction  instruction))
   (println)
)
\end{verbatim}
}


\begin{defn}

Define $\mathfrak{U}$ as the union of $\mathfrak{Q}(x)$ and all ex-machines 
$\mathfrak{Q}(a_0  \dots a_m$ $x)$ for each  $m \in \mathbb{N}$ and for each $a_0 \dots a_m$ in  $\{0, 1\}^{m+1}$.   
In other words, $$\mathfrak{U} = \big{\{} \mathfrak{Q}(x) \big{\} }  \mbox{\hskip 0.4pc}    \bigcup    \mbox{\hskip 0.8pc}
{\overset{\infty}{\underset{m=0} \bigcup}}  \mbox{\hskip 0.6pc}
{\underset{a_0 \dots a_m \in \{0, 1\}^{m+1}} \bigcup}  \big{\{} \mathfrak{Q}(a_0 a_1 \dots a_m \mbox{\hskip 0.2pc} x) \big{\}}.$$

\end{defn}

\begin{thm}\label{theorem:Q_x_evolves_to_L_f}
Each language $L_f$ in $\mathfrak{L}$ can be computed by the evolving sequence of ex-machines \verb| | 
$\mathfrak{Q}(x)$, \verb|| $\mathfrak{Q}(f(0)$ $x)$,  $\mathfrak{Q}(f(0)f(1)$ $x)$, \verb| | 
$\dots$, \verb| | $\mathfrak{Q}(f(0)f(1)$ $\dots f(n)$ $x)$, $\dots$.  
\end{thm}

\begin{proof}[Proof]   
The theorem follows from ex-machine \ref{machine:meta_a_machine}, 
ex-machine \ref{defn:Q_a0_a1_am_listing} and lemma \ref{lemma:Q_a0_a1_am}. 
\end{proof}

\begin{cor}
Given function $f : \mathbb{N} \rightarrow \{0, 1\}$, for any arbitrarily large $n$, 
the evolving sequence of ex-machines 
$\mathfrak{Q}(f(0)f(1) \dots f(n)$ $x)$, $\mathfrak{Q}(f(0)f(1)$ 
$\dots$ $f(n) f(n+1)$ $x)$, $\dots$.
computes language $L_f$.  
\end{cor}

\begin{cor}\label{corollary:Q_x_finite_resources}
 Moreover, for each $n$, all ex-machines $\mathfrak{Q}(x)$, \verb| | 
 $\mathfrak{Q}(f(0) x)$,  \verb| | 
 $\mathfrak{Q}(f(0)f(1)$ $x)$, \verb| | 
 $\dots$, \verb|| 
 $\mathfrak{Q}(f(0)f(1) \dots$ $f(n)$ $x)$ combined 
have used only a finite amount of tape, finite number of states, finite number of instructions, 
finite number of executions of instructions and only a finite amount of quantum random information  
measured by the quantum random instructions.  
\end{cor}

\begin{proof}[Proof]
For each $n$, the finite use of computational resources follows immediately from remark \ref{rem:finite_conditions},  
definition \ref{defn:evolving}  
and the specification of ex-machine \ref{defn:Q_a0_a1_am_listing}.   
\end{proof}

A set $X$ is called countable if there exists a bijection between $X$ and $\mathbb{N}$.    
Since the set of all Turing machines is countable and each Turing machine only 
recognizes a single language 
{\it most  (in the sense of Cantor's heirarchy of infinities \cite{cantor_transfinite}) 
languages $L_f$  are not computable with a Turing machine}.  More precisely, 
the set of languages $L_f$  computable with a Turing machine is a countable set,         
while the set of all languages $\mathfrak{L}$ is an uncountable set.  

\smallskip 

For each non-negative integer $n$, define the language tree 
$\mathcal{L}(a_0 a_1 \dots a_n)$ 
$= \{ L_f : f \in \{0, 1\}^{\mathbb{N}}$ \verb|and | $f(i) = a_i$ \verb| for | 
$i$ \verb| satisfying | $0 \le i \le n \}$.
Define the corresponding subset of  $\{0, 1\}^{\mathbb{N}}$ as  
$\mathcal{S}(a_0 a_1 \dots a_n)$ 
$= \{ f \in \{0, 1\}^{\mathbb{N}}:$ $f(i) = a_i$ \verb| for | 
$i$ \verb| satisfying | $0 \le i \le n \}$.  
Let $\Psi$ denote this 1-to-1 correspondence, where   
$\mathfrak{L}$ ${\overset{\Psi}\leftrightarrow}$ $\{0, 1\}^{\mathbb{N}}$  and 
$\mathcal{L}(a_0 a_1 \dots a_n) \mbox{\hskip 0.5pc}  {\overset{\Psi}\leftrightarrow}  \mbox{\hskip 0.5pc}
\mathcal{S}(a_0 a_1 \dots a_n).$

Since the two quantum random axioms \ref{axiom_qr1} and \ref{axiom_qr2} 
are satisfied, each finite path $f(0)f(1) \dots f(n)$ is equally likely and 
there are $2^{n+1}$ of these paths.  Thus, each path of length $n+1$ has 
probability $2^{-(n+1)}$.  These uniform probabilities on finite strings of the same length 
can be extended to the Lebesgue measure $\mu$ on 
probability space $\{0, 1\}^{\mathbb{N}}$ \hskip 0.2pc \cite{feller_vol1,feller_vol2}.   
Hence, each subset $\mathcal{S}(a_0 a_1 \dots a_n)$ has measure $2^{-(n+1)}$.  
That is, $\mu \big{(} \mathcal{S}(a_0 a_1 \dots a_n) \big{)} = 2^{-(n+1)}$   and  
$\mu(\{0, 1\}^{\mathbb{N}}) = 1$.  Via the $\Psi$ correspondence between each language tree
$\mathcal{L}(a_0 a_1 \dots a_n)$ and subset $\mathcal{S}(a_0 a_1 \dots a_n)$, 
uniform probability measure  $\mu$ induces a uniform probability measure $\nu$ on $\mathfrak{L}$, 
where $\nu \big{(} \mathcal{L}(a_0 a_1 \dots a_n) \big{)} = 2^{-(n+1)}$ 
and $\nu( \mathfrak{L} ) = 1$.

\begin{thm}\label{thm:L_f_incomputable_measure_1}
For functions $f : \mathbb{N} \rightarrow \{0, 1\}$, the probability that language $L_f$ is Turing 
incomputable has measure 1 in  $(\nu, \mathfrak{L})$.  
\end{thm}

\begin{proof}   
The Turing machines are countable and therefore the number of functions 
$f: \mathbb{N} \rightarrow \{0, 1\}$ that are Turing computable is countable.  
Hence, via the $\Psi$ correspondence, the Turing computable languages 
$L_f$ have measure 0 in $\mathfrak{L}$.  
\end{proof}

\noindent Moreover, the Martin-L{\"o}f random sequences $f: \mathbb{N} \rightarrow \{0, 1\}$ 
have Lebesgue measure 1 in $\{0, 1\}^{\mathbb{N}}$ and are a proper subset of the Turing 
incomputable sequences.  See \cite{calude,martin_lof}.

\begin{cor}\label{cor:Q_x_not_a_Turing_machine} 
$\mathfrak{Q}(x)$ is not a Turing machine.  
Each ex-machine $\mathfrak{Q}(a_0 a_1$ $\dots a_m$ $x)$ in $\mathfrak{U}$ is not a Turing machine.
\end{cor}

\begin{proof}[Proof]
$\mathfrak{Q}(x)$ can evolve to compute Turing incomputable languages on a set of probability measure 1 
with respect to  $(\nu, \mathfrak{L})$.  Also,  $\mathfrak{Q}(a_0 a_1$ $ \dots a_m$ $x)$ 
can evolve to compute Turing incomputable languages on a set of 
measure $2^{-(m+1)}$  with respect to  $(\nu, \mathfrak{L})$.  
In contrast, each Turing machine only recognizes a single language, which has measure 0.  In fact, the 
measure of all Turing computable languages is 0 in $\mathfrak{L}$.  
\end{proof}

\begin{rem}

The statements in theorem \ref{thm:L_f_incomputable_measure_1} 
and corollary \ref{cor:Q_x_not_a_Turing_machine} can be sharpened 
when deeper results are obtained for the quantum random source 
\cite{calude_qr2012,calude_qr2014,calude_qr2015,kulikov}
used by the quantum random instructions.  

\end{rem}



\section{Some  $\mathfrak{Q}(x)$  Observations  Based on Cantor and G{\"o}del  }

At first glance, the results from the prior section may seem paradoxical.  Even though
there are only a countable number of initial ex-machines in $\mathfrak{U}$, 
the ex-machines evolving from $\mathfrak{Q}(x)$ can compute languages $L_f$ where each 
$f: \mathbb{N} \rightarrow \{0, 1\}$ corresponds to a particular instance selected from 
an uncountable number of infinite paths in the infinite binary tree    
(i.e, $\{0, 1 \}^{\mathbb{N}}$ is uncountable \cite{cantor_diagonal}).  
With initial state 0 and initial tape   
\verb|#| \hskip 0.5pc \verb|#|{\verb|a|}$^n$\verb|#|, 
for every $n$ and $m$ with $n > m$, each ex-machine 
$\mathfrak{Q}(a_0 a_1 \dots a_m$ $x)$  has an uncountably infinite number 
of possible execution behaviors.  On the other hand, a Turing machine with the same initial 
state 0 and initial tape 
\verb|#| \hskip 0.5pc \verb|#|{\verb|a|}$^n$\verb|#| always 
has exactly one execution behavior.  Hence, a Turing machine can only have 
a countable number of execution behaviors for all initial tapes 
\verb|#| \hskip 0.5pc \verb|#|{\verb|a|}$^n$\verb|#|, where $n > m$.

It may seem peculiar that the countable set $\mathfrak{U}$ of ex-machines can evolve to 
compute an uncountable number of languages $L_f$.  However, there is an analogous 
phenomenon in elementary  analysis 
that mathematicians routinely accept.   The rational numbers $\mathbb{Q}$ are countable.       
The set $\mathbb{Q} \cap [0, 1]$ is dense in the closed interval $[0, 1]$ of real numbers.   
Any real number $r \in [0, 1]$ can be expressed as 
$r = {\underset{i = 1}{\overset{\infty} \sum}} b_i 2^{-i}$, where each $b_i \in \{0, 1\}$. 
Set the $m$th rational number $q_m  = {\underset{i = 1}{\overset{m} \sum}} b_i 2^{-i}$.  
Then ${\underset{m \rightarrow \infty}\lim} q_m = r$.   Thus, each real number can be realized as a 
sequence of rational numbers, even though the real numbers are uncountable.   Furthermore, each 
rational number in that sequence is representable with a finite amount of information 
(bits).  Similar to the sequence of rationals $q_m$ converging to a real number, each language $L_f$ 
can be computed (i.e., realized) by the evolving sequence of ex-machines 
$\mathfrak{Q}(x)$, \hskip 0.1pc $\mathfrak{Q}(f(0) x)$,  \hskip 0.1pc  $\mathfrak{Q}(f(0)f(1)$ $x)$, \hskip 0.1pc  
$\dots$, \hskip 0.1pc  $\mathfrak{Q}(f(0)f(1) \dots f(n)$ $x)$, $\dots$, where for each $n$, all ex-machines 
$\mathfrak{Q}(x)$, \hskip 0.1pc  $\mathfrak{Q}(f(0) x)$,  \hskip 0.1pc  
$\mathfrak{Q}(f(0)f(1)$ $x)$, \hskip 0.1pc  $\dots$, 
\hskip 0.1pc  $\mathfrak{Q}(f(0)f(1) \dots f(n)$ $x)$
have used only a finite amount of tape, finite number of states, finite number of instructions, 
finite number of executions of instructions and a finite amount of quantum random information 
has been measured.


Finally, our attention turns to an insightful remark by G{\"o}del, entitled 
{\it A philosophical error in Turing's work } \cite{godel_72}.   G{\"o}del states:

\begin{quote}
Turing in his [1936 \cite{turing36}, section 9] gives an argument which is supposed to show that mental procedures cannot go beyond 
mechanical procedures.  However, this argument is inconclusive.  What Turing disregards completely is the fact that 
{\it mind, its use, is not static, but constantly developing}, i.e., that we understand abstract terms more and more precisely 
as we go on using them, and that more and more abstract terms enter the sphere of our understanding.  
There may exist systematic methods of actualizing this development, which could form part of the procedure. 
Therefore, although at each stage the number and precision of the abstract terms at our disposal may be {\it finite}, both 
(and, therefore, also Turing's number of {\it distinguishable states of mind}) may {\it converge toward infinity} in the
course of the application of the procedure.  Note that something like this indeed seems to happen in the process of forming stronger and 
stronger axioms of infinity in set theory.  This process, however, today is far from being sufficiently understood to form a well-defined 
procedure which could actually be carried out (and would yield a non-recursive number-theoretic function).  
\end{quote}

\medskip 

Although we make no claim whatsoever that the execution of $\mathfrak{Q}(x)$ functions anything 
like a mental procedure, G{\"o}del attributes some properties to mental procedures that are strikingly similar 
to the ex-machine.   First, the ex-machine $\mathfrak{Q}(a_0 a_1 \dots a_m$ $x)$ is not static and 
constantly develops each time it is queried with a string {\verb|a|}$^n$ such that $n > m$.  (String {\verb|a|}$^n$  
is longer than any prior string that an ancestor of $\mathfrak{Q}(a_0 a_1 \dots a_m$ $x)$  has executed upon.)

After $\mathfrak{Q}(a_0 a_1 \dots a_n$ $x)$'s computation halts, the resulting ex-machine always has a finite number of 
states and a finite number of instructions, so at each stage of the evolution, the ex-machine is finite.   
Lastly, consider G{\"o}del's comment:  
``(and, therefore, also Turing's number of {\it distinguishable states of mind}) may {\it converge toward infinity} in the 
course of the application of the procedure''.  G{\"o}del's insight seems to foresee the ex-machine's ability to 
add new states;  moreover, to compute a Turing incomputable language $L_f$, any ex-machine 
$\mathfrak{Q}(f(0) f(1) \dots f(n)$ $x)$ must have an evolutionary path that has an unbounded number of states.


\section{An Ex-Machine Halting Problem }\label{msf:halting_problem}

In \cite{turing36}, Alan Turing posed the halting problem for Turing machines.  
Does there exist a Turing machine $\mathcal{D}$ that can determine 
for any given Turing machine $\mathcal{M}$ and finitely bounded initial tape $T$ 
whether $\mathcal{M}$'s execution on tape $T$ eventually halts?  
In the same paper \cite{turing36}, Turing proved that no single Turing machine could 
solve his halting problem.

Next, we explain what Turing's seminal result relies upon in terms 
of abstract computational resources.  Turing's result means that there does not 
exist a single Turing machine $\mathcal{H}$ -- regardless of the size of $\mathcal{H}$'s 
finite state set $Q$ and finite alphabet set $A$  -- so that when this special machine 
$\mathcal{H}$ is presented with any Turing machine $\mathcal{M}$ with a finitely bounded initial tape $T$ 
and initial state $q_0$,  then $\mathcal{H}$ can execute a finite number of computational 
steps, halt and correctly determine whether $\mathcal{M}$ halts or does not halt 
with a tape $T$ and initial state $q_0$.  In terms of {\it definability}, the statement 
of Turing's halting problem ranges over all possible Turing machines and all possible 
finitely bounded initial tapes.  
This means:  for each tape $T$ and machine $\mathcal{M}$, there are finite initial 
conditions imposed on tape $T$ and machine $\mathcal{M}$. However, as tape $T$ and 
machine $\mathcal{M}$ range over all possibilities, the computational resources 
required for tape $T$ and machine $\mathcal{M}$ are unbounded.   
Thus, the computational resources required by $\mathcal{H}$ are unbounded 
as its input ranges over all finitely bounded initial tapes $T$ and machines $\mathcal{M}$.  


The previous paragraph provides some observations about Turing's halting problem because   
any philosophical objection to $\mathfrak{Q}(x)$'s unbounded computational resources during an 
evolution should also present a similar philosophical objection to Turing's assumptions 
in his statement and proof of his halting problem.  Notice that 
corollary \ref{corollary:Q_x_finite_resources} supports our claim.


Since $\mathfrak{Q}(x)$  and every other  ex-machine $\mathfrak{Q}(a_0 a_1 \dots a_m$ $x)$ 
in $\mathfrak{U}$ is not a Turing machine, there is a natural extension of 
Turing's halting problem.  Does there exist an ex-machine $\mathfrak{G}(x)$ such that 
for any given Turing machine $\mathcal{M}$ and finite initial tape $T$, then 
$\mathfrak{G}(x)$ can sometimes compute whether $\mathcal{M}$'s execution on tape 
$T$ will eventually halt?     Before we call this the ex-{\it machine halting problem}, 
the phrase {\it can sometimes compute whether} must be defined so that this 
problem is well-posed.  A reasonable definition requires some work.

From the universal Turing machine / enumeration theorem \cite{downey,rogers}, 
there exists a Turing computable enumeration
$\mathcal{E}: \mathbb{N} \rightarrow$ 
$\{$\verb|all Turing machines| $\mathcal{M} \}$ $\times$ 
$\{$\verb|Each of |$\mathcal{M}$\verb|'s states as an initial state|$\}$ 
of every Turing machine.  Similar to ex-machines, for each machine $\mathcal{M}$, 
the set $\{$\verb|Each of |$\mathcal{M}$\verb|'s states as|
\verb|an initial state|$\}$  can be realized as a 
finite subset $\{0, \dots, n-1\}$ of $\mathbb{N}$.  Since $\mathcal{E}(n)$ 
is an ordered pair, the phrase ``Turing machine $\mathcal{E}(n)$" refers to the first coordinate of 
$\mathcal{E}(n)$.  Similarly, the ``initial state $\mathcal{E}(n)$" refers to    
the second coordinate of $\mathcal{E}(n)$.

Recall that the Turing machine halting problem is equivalent 
to the blank-tape halting problem.  (See pages 150-151 of \cite{minsky}).   
For our discussion, the blank-tape halting problem translates to:  for each Turing machine 
$\mathcal{E}(n)$, does Turing machine $\mathcal{E}(n)$ halt  
when $\mathcal{E}(n)$ begins its execution with a blank initial tape  
and initial state $\mathcal{E}(n)$?

Lemma \ref{lemma:Q_a0_a1_am} implies that the same 
initial ex-machine can evolve to two different ex-machines; furthermore, 
these two ex-machines will never compute the same language no matter what 
descendants they evolve to.   For example, $\mathfrak{Q}(0$ $x)$ and 
$\mathfrak{Q}(1$ $x)$ can never compute the same language in $\mathfrak{L}$.
Hence, {\it sometimes} means that for each $n$, there exists an evolution of $\mathfrak{G}(x)$  
to  $\mathfrak{G}(a_0 x)$, and then to $\mathfrak{G}(a_0 a_1 x)$ and so on up to 
$\mathfrak{G}(a_0 a_1 \dots a_n$ $x)$ $\dots$, where for each $i$ with $0 \le i \le n$, 
then  $\mathfrak{G}(a_0 a_1 \dots a_n$ $x)$ correctly computes whether Turing machine 
$\mathcal{E}(n)$ -- executing on an initial blank tape with initial state $\mathcal{E}(n)$  -- 
halts or does not halt.

In the prior sentence, the word {\it computes} means that 
$\mathfrak{G}(a_0 a_1 \dots a_i$ $x)$  halts after a finite number of instructions 
have been executed and the halting output written by $\mathfrak{G}(a_0 a_1 \dots a_i$ $x)$  
on its tape indicates whether machine $\mathcal{E}(n)$ halts or does not halt.  
For example, if the input tape is \verb|# #|\verb|a|$^i$\verb|#|, then 
enumeration machine $\mathcal{M}_{\mathcal{E}}$ writes the representation of $\mathcal{E}(i)$ on the tape, 
and then   $\mathfrak{G}(a_0 a_1 \dots a_m$ $x)$ with $m \ge i$ halts with \verb|# Y#| written 
to the right of the representation for machine $\mathcal{E}(i)$.  
Alternatively $\mathfrak{G}(a_0 a_1 \dots a_m$ $x)$ with $m \ge i$ halts with \verb|# N#| written 
to the right of the representation for machine $\mathcal{E}(i)$. 
The word {\it correctly} means that ex-machine $\mathfrak{G}(a_0 a_1 \dots a_m$ $x)$  halts with 
\verb|# Y#| written on the tape if machine  $\mathcal{E}(i)$ halts and  
ex-machine $\mathfrak{G}(a_0 a_1 \dots a_m$ $x)$  halts with  \verb|# N#| written on the tape 
if machine  $\mathcal{E}(i)$ does not halt.

Next, our goal is to transform the ex-machine halting problem to a form 
so that the results from the previous section can be applied.  
 Choose the alphabet as $\mathcal{A} = \{$\verb|#|, \verb|0|, \verb|1|, \verb|a|, 
\verb|A|, \verb|B|, \verb|M|, \verb|N|, \verb|S|, \verb|X|, \verb|Y|$\}$.   
As before, for each Turing machine, it is helpful to identify the set of machine 
states $Q$ as a finite subset of $\mathbb{N}$.  
Let $\mathcal{M}_{\mathcal{E}}$ be the Turing machine that computes a Turing computable enumeration 
\footnote{Chapter 7 of \cite{minsky} provides explicit details of encoding quintuples with a particular 
universal Turing machine.  Alphabet $\mathcal{A}$ was selected so that it is compatible 
with this encoding.  A careful study of chapter 7 should provide a clear path of how 
$\mathcal{M}_{\mathcal{E}}$'s instructions  can be specified to implement $\mathcal{E}_a$.}
as $\mathcal{E}_a:  \mathbb{N}  \rightarrow \{ \mathcal{A} \}^* \times \mathbb{N}$, 
where the tape  \verb|# #|\verb|a|$^n$\verb|#| represents natural number $n$.
Each $\mathcal{E}_a(n)$ is an ordered pair where the first coordinate is the Turing machine and the 
second coordinate is an initial state chosen from one of $\mathcal{E}_a(n)$'s states.


 \begin{rem}\label{rem:E_a_two_choices}  \hskip 1pc   For each $n \in \mathbb{N}$, 
 with blank initial tape and initial state  $\mathcal{E}_a(n)$, then Turing machine 
 $\mathcal{E}_a(n)$ either halts or does not halt.  
\end{rem}

\begin{proof}[Proof]
The execution behavior of Turing machine computation is unambiguous.  
For each $n$, there are only two possibilities.
\end{proof}

For our particular instance of $\mathcal{E}_a$, define the {\it halting function}  
$h_{\mathcal{E}_a}: \mathbb{N} \rightarrow \{0, 1\}$ as follows.  
For each $n$, set $h_{\mathcal{E}_a}(n) = 1$,  
whenever Turing machine $\mathcal{E}_a(n)$ with blank initial tape 
and initial state  $\mathcal{E}_a(n)$ halts.  Otherwise, set 
$h_{\mathcal{E}_a}(n) = 0$,   
if Turing machine $\mathcal{E}_a(n)$ with blank initial tape 
and initial state  $\mathcal{E}_a(n)$ does not halt.
Remark \ref{rem:E_a_two_choices} implies that function $h_{\mathcal{E}_a}(n)$ is well-defined.  
Via the halting function $h_{\mathcal{E}_a}(n)$  and definition 
\ref{defn:language_L_f},   
define the {\it halting language} $L_{h_{\mathcal{E}_a}}$.  

\begin{thm}\label{thm:x_maching_halting_evolution}  
The  ex-machine $\mathfrak{Q}(x)$ has an evolutionary path that computes halting language  
$L_{h_{\mathcal{E}_a}}$.   This evolutionary path is 
\hskip 0.2pc 
$\mathfrak{Q}(h_{\mathcal{E}_a}(0)$  \hskip 0.2pc $x)$  \hskip 0.2pc    $\rightarrow$  \hskip 0.2pc 
$\mathfrak{Q}(h_{\mathcal{E}_a}(0)$ $h_{\mathcal{E}_a}(1)$ \hskip 0.2pc $x)$ \hskip 0.2pc  $\rightarrow$ \hskip 0.2pc 
$\dots$  \hskip 0.2pc 
$\mathfrak{Q}(h_{\mathcal{E}_a}(0)$ $h_{\mathcal{E}_a}(1)$ $\dots$ $h_{\mathcal{E}_a}(m)$  \hskip 0.2pc $x)$ 
\hskip 0.2pc $\dots$
\end{thm}

\begin{proof}
Theorem \ref{thm:x_maching_halting_evolution} follows from the previous discussion, 
including the definition of halting function $h_{\mathcal{E}_a}(n)$ and 
halting language $L_{h_{\mathcal{E}_a}}$ and 
theorem \ref{theorem:Q_x_evolves_to_L_f}.
\end{proof}

\subsection{ Some Observations Based on Theorem \ref{thm:x_maching_halting_evolution}  }

Although theorem \ref{thm:x_maching_halting_evolution}  provides an affirmative answer to the 
ex-machine halting problem, in practice, a particular execution of $\mathfrak{Q}(x)$ will not,  
 from a probabilistic perspective, evolve to compute $L_{h(\mathcal{E}_a)}$.  
For example, the probability  is  $2^{-128}$ that a particular execution of 
$\mathfrak{Q}(x)$ will evolve to  $\mathfrak{Q}(a_0 a_1 \dots a_{127}$ $x)$  
so that $\mathfrak{Q}(a_0 a_1 \dots a_{99}$ $x)$  correctly computes whether 
each string \verb|a|$^0$, \verb|a|, \verb|a|$^2$ $\dots$ \verb|a|$^{127}$ is a member of  
$L_{h(\mathcal{E}_a)}$ or not a member of  $L_{h(\mathcal{E}_a)}$.

Furthermore, theorem \ref{thm:x_maching_halting_evolution} provides no general method 
for infallibly testing (proving) that an evolution of $\mathfrak{Q}(x)$ to some new machine   
$\mathfrak{Q}(a_0 a_1 \dots a_m$ $x)$ satisfies $a_i = h_{\mathcal{E}_a}(i)$
 for each $0 \le i \le m$.  We also know that any such general testing method that works 
 for all natural numbers $m$  would require at least an ex-machine 
 (or some computational object more powerful than an ex-machine if that object exists)
because any general testing method cannot be implemented with a standard machine.  
Otherwise, if such a testing method could be executed by a standard machine, then 
this special standard machine could be used to solve  Turing's halting problem:  
this is logically impossible due to Turing's proof 
of the unsolvability of the halting problem with a standard machine.

Despite all of this, it is still {\it logically possible} for this evolution to happen, 
since  $\mathfrak{Q}(x)$ can in principle evolve to compute any language $L_f$ in $\mathfrak{L}$.  
In other words, every infinite, downward path in the infinite binary tree 
of figure 3 is possible and equally likely.  Clearly, $\mathfrak{Q}(x)$ is 
{\it not an ``intelligent" ex-machine}, by any reasonable definition of ``intelligent",  
since evolutionary path
$\mathfrak{Q}(h_{\mathcal{E}_a}(0)$  \hskip 0.12pc $x)$  \hskip 0.1pc    $\rightarrow$  \hskip 0.1pc 
$\mathfrak{Q}(h_{\mathcal{E}_a}(0)$ $h_{\mathcal{E}_a}(1)$ \hskip 0.12pc $x)$ \hskip 0.1pc  $\rightarrow$ \hskip 0.1pc 
$\dots$  \hskip 0.1pc 
$\mathfrak{Q}(h_{\mathcal{E}_a}(0)$ $h_{\mathcal{E}_a}(1)$ $\dots$ $h_{\mathcal{E}_a}(m)$  \hskip 0.15pc $x)$ 
\hskip 0.1pc $\dots$ 
relies solely on blind luck. 

In some ways, theorem \ref{thm:x_maching_halting_evolution}  has an analogous result in pure mathematics.          
The Brouwer fixed point theorem \cite{brouwer} guarantees that a continuous map from an 
$n$-simplex to an $n$-simplex has at least one fixed point and 
demonstrates the power of algebraic topology \cite{steenrod}. 
However,  the early proofs were indirect and provided no constructive methods for 
finding the fixed point(s).  The parallel here is that theorem \ref{thm:x_maching_halting_evolution}  
guarantees that an evolutionary path exists, but the proof provides no general method for infallibly testing 
that for an evolution up to stage $m$, then  $\mathfrak{Q}(a_0 a_1 \dots a_m$ $x)$ satisfies 
$a_i = h_{\mathcal{E}_a}(i)$ for each $0 \le i \le m$.   

Algorithmic methods for finding fixed points were developed about 60 years later \cite{scarf,yorke}.
The part of the analogy that has not yet played out could break down due to the extreme ramifications 
of reaching large enough sizes of $m$, whereby currently intractable problems in mathematics could be proven.  
However, this really depends upon the computing speeds and the ex-machine learning and 
mathematical methods developed over the next few centuries.   We believe that deeper, ex-machine learning 
and mathematical methods could have a larger impact than hardware advances   
because a clever proof can save a large number of mechanical steps over a mediocre proof.  
And the clever use of a prior theorem or symmetry in a new proof can save an infinite number 
of computational steps.

With the history of the Brouwer fixed point theorem in mind, the logical possibility, 
demonstrated by theorem \ref{thm:x_maching_halting_evolution}, suggests that there might 
be an opportunity to develop new problem solving methods and 
apply more advanced ex-machines to key instances of the halting problem.  
As far as more advanced ex-machines, a broad research direction is to explore 
the use of populations of ex-machines that evolve and also communicate formal languages 
with each other, analogous to the methods that human mathematicians use 
in their mathematical research.

The Goldbach conjecture states that  every even number greater than 2 is the sum of two primes.   
This conjecture seems to be an intractable problem in 
number theory as Goldbach posed it  \cite{goldbach} in the year 1742.  
Despite its apparent intractability, a fairly simple Turing machine can be specified;  
namely,  proving that the following Goldbach machine never halts 
provides a proof of the conjecture.

\begin{Instructions}\label{ins:goldbach}  \hskip 1pc  {\it A Goldbach Machine }   

{ \small 
\begin{verbatim}
set n = 2
set g = true
set prime_list = (2)
while (g == true)   
{
    set n = n + 2
    set g = false
    for each p in prime_list  
    {
       set x = n - p
       if (x is a member of prime_list)  set g = true
    }  

    if (n-1 is prime)  store n-1 in prime_list
}

print ("Even number " n " is not the sum of two primes.")
print ("The Goldbach conjecture is false!") 
halt 
\end{verbatim}
}
\end{Instructions}

We wrap up this section with some advice from mathematician George P\'olya.  
P\'olya emphasizes the use of heuristics for mathematical problem solving, which  
may shed further light on what we broadly have in mind for  
ex-machines that help solve intractable math problems.  Advanced 
ex-machines may help human mathematicians with the conception of a proof. 
A propos to our discussion, P\'olya \cite{polya} distinguishes between 
conceiving of a proof versus verifying a proof:  

\medskip 

\begin{quote}
The following pages are written somewhat concisely, but simply as possible, and 
are based on a long and serious study of methods of solution.  
This sort of study, called heuristic by some writers, is not in fashion 
nowadays but has a long past and, perhaps, some future.

Studying the methods of solving problems, we perceive another face of mathematics. 
Yes, mathematics has two faces;  it is the rigorous science of Euclid but it is also 
something else.  Mathematics presented in the Euclidean way appears as a systematic, 
deductive science; but mathematics in the making appears as an experimental, inductive 
science.  Both aspects are as old as the science of mathematics itself.  
\end{quote}

\section{Two Research Problems}\label{sect:two_research_problems}

P\'olya expresses a broad vision, but without some concrete research problems aimless wandering is likely.  
We propose two mathematical problems, where the goal is to express each one as a 
halting problem for a single Turing machine.  Furthermore, each problem has a different 
strategy for teaching us more about new ex-machine computation that advances beyond 
$\mathfrak{Q}(x)$'s blind luck. 

\smallskip 

{\bf Mathematical Problem 1.}  

Specify explicit initial instructions of an ex-machine that can evolve to 
compute a proof that $\sqrt{2}$ is irrational.  One reason for proposing this problem 
is that we already know the correct answer.  Another reason is that the human proof is 
short and clever.    One possible ex-machine approach follows the traditional 
method of constructing a proof by contradiction based on the theorems about odd and even natural numbers. 

A second approach is more involved.  However, new techniques, gained from this approach, 
may be applicable to other instances of halting problems.  Consider the pseudocode below that computes 
the $\sqrt{2}$, by executing a bisection algorithm on the curve $y = x^2 - 2$.  

\smallskip 

\begin{Instructions}\label{ins:sqrt_2} \hskip 1pc   {A  \hskip 0.1pc $\sqrt{2}$ \hskip 0.2pc Standard Machine}
{ \small 
\begin{verbatim}
  Set l = 1
  Set u = 2
  Set a = true
  while (a == true)   
  {
      set x = (l + u) / 2

      if  ( (x * x) > 2)   set l = x
      else   set u = x

      if (the tape representation of x so far has a periodic sequence)
          and (the periodicity will continue indefinitely) 
      {
         set a = false
      } 
  }
  halt
  \end{verbatim}
}
\end{Instructions}

The bisection approach searches for a periodic sequence of writing the tape symbols.  
The critical part is the instruction \verb|if (the tape representation|
\verb|of| \verb| x has a periodic sequence)| 
\verb|and (the periodicity| \verb| will continue| \verb|indefinitely)|.  
How does an ex-machine evolve rules to recognize a periodic sequence of symbols 
(written on the tape by the $\sqrt{2}$ standard machine)
and also guarantee that the periodic sequence on the tape will repeat indefinitely?   In other words, the 
{\it ex-machine  evolves rules that adequately represent knowledge about the 
     dynamics of the $\sqrt{2}$ \hskip 0.2pc machine's instructions}.

A consecutive repeating state cycle in a Turing machine occurs when a finite sequence of standard machine 
instructions $\{I_i\}$ is executed by the Turing machine two consecutive times: 
$I_1 \rightarrow I_2 \rightarrow \dots I_k \rightarrow I_1 \rightarrow I_2 \dots I_k$ and 
the machine configuration before the first instruction $I_1$ is executed equals the machine 
configuration after the instruction $I_k$ has completed its execution a second time.

In \cite{fiske_state_cycle}, the main theorem shows  that consecutive repeating state cycles characterize
 the periodic points of a Turing machine.  
A periodic point that does not reach a halting state indicates that the Turing machine execution is 
immortal (i.e., never halts).  Can this consecutive repeating state cycle theorem or an extension of this 
theorem be used to help an ex-machine find a proof? If the standard $\sqrt{2}$ machine writes symbols on the 
tape in a periodic sequence, this indicates that $\sqrt{2}$ is rational.   If an ex-machine can 
construct rules which prove that the standard $\sqrt{2}$ machine  \label{ins:sqrt_2} never halts, 
then these ex-machine rules provide a proof that the $\sqrt{2}$ is irrational.


\smallskip  

{\bf Mathematical Problem 2.}

Transform Collatz machine \ref{ins:collatz_machine}'s execution of each individual orbit 
$\mathcal{O}(f, n)$ into a single ex-machine computation 
that collectively makes a determination about all individual orbits.  
That is, find an ex-machine  
computation that evolves to a decision whether $1$ is in $\mathcal{O}(f, n)$ for 
all $n \in \mathbb{N}$. 
Is it possible to accomplish this with an ex-machine computation?  
If it is impossible, why?

Consider the augmentation of Collatz machine  \ref{ins:collatz_machine} to 
an enumerated Collatz machine $\mathcal{E}$.  The standard machine   
$\mathcal{E}$  iterates over the odd numbers $3, 5, 7, \dots$. $\mathcal{E}$   
first writes \verb|# #111#| on the input tape and hands this computation over 
to Collatz machine \ref{ins:collatz_machine}.  After  Collatz machine 
\ref{ins:collatz_machine}  halts at 1, then  $\mathcal{E}$ updates the input 
tape to \verb|# #11111#|, representing 5, and hands this to the 
Collatz machine again. After the Collatz machine halts at 1, then 
 $\mathcal{E}$ updates the input tape to \verb|# #1111111#|, and so on.  
If the Collatz conjecture is true, this execution of $\mathcal{E}$ never halts 
and $\mathcal{E}$ iterates over every odd number.

At least part of the challenge with  machine  $\mathcal{E}$  seems to be that 
there could exist some $n$ such that $n$'s Collatz orbit reaches a periodic 
attractor that does not contain 1.   Another possibility is that there exists some $u$  
whose Collatz orbit aperiodicly oscillates and never reaches 1.  
In this case, $u$'s orbit does not have an upper bound.  That is, 
\verb|sup| $\mathcal{O}(f, u) = \infty$.  In both cases, the orbit of $n$ and the orbit of 
$u$ do not halt at 1.  If the conjecture is true, how does one distinguish these two 
different types of immortal orbits from the enumerated Collatz machine that halts at 1 
for each odd input, but is also immortal?

Is it possible to transform (either by human ingenuity or by ex-machine evolution or a combination) 
this enumerated Collatz machine $\mathcal{E}$ into a non-vacuous, explicit Turing machine so 
that an immortal orbit proves or disproves that the Collatz conjecture is true?  
If this transformation exists, does there exist an ex-machine that can construct this transformation?
Perhaps, the answers to these types of questions will provide some insight on a famous remark by 
mathematician Erd{\"o}s about the Collatz conjecture.  
{\it Mathematics is not ready for such problems.}



\newpage





\newpage

\section{Appendix -- Execution of the Collatz Machine on $n=5$ }\label{appendix:collatz}

{ \scriptsize

\noindent   Each row shows the current tape and machine state after the instruction in that 
row has been executed.  Before the machine starts executing, the initial tape is 
 \verb|# #11111#|  and the initial state is \verb|q|.  
The space indicates the location of the tape head. 

\smallskip 

\begin{verbatim}
STATE   TAPE                    TAPE HEAD        INSTRUCTION             COMMENTS
  a     ## 11111#########       1                (q, #, a, #, 1)   
  b     ##1 1111#########       2                (a, 1, b, 1, 1)   
  c     ##11 111#########       3                (b, 1, c, 1, 1)   
  d     ##111 11#########       4                (c, 1, d, 1, 1)   
  c     ##1111 1#########       5                (d, 1, c, 1, 1)   
  d     ##11111 #########       6                (c, 1, d, 1, 1)   
  k     ##1111 1#########       5                (d, #, k, #, -1)        Compute 3*5+1
  l     ##11110 #########       6                (k, 1, l, 0, 1)         
  m     ##111100 ########       7                (l, #, m, 0, 1)   
  k     ##11110 00#######       6                (m, #, k, 0, -1)   
  k     ##1111 000#######       5                (k, 0, k, 0, -1)   
  k     ##111 1000#######       4                (k, 0, k, 0, -1)   
  l     ##1110 000#######       5                (k, 1, l, 0, 1)   
  l     ##11100 00#######       6                (l, 0, l, 0, 1)   
  l     ##111000 0#######       7                (l, 0, l, 0, 1)   
  l     ##1110000 #######       8                (l, 0, l, 0, 1)   
  m     ##11100000 ######       9                (l, #, m, 0, 1)   
  k     ##1110000 00#####       8                (m, #, k, 0, -1)   
  k     ##111000 000#####       7                (k, 0, k, 0, -1)   
  k     ##11100 0000#####       6                (k, 0, k, 0, -1)   
  k     ##1110 00000#####       5                (k, 0, k, 0, -1)   
  k     ##111 000000#####       4                (k, 0, k, 0, -1)   
  k     ##11 1000000#####       3                (k, 0, k, 0, -1)   
  l     ##110 000000#####       4                (k, 1, l, 0, 1)   
  l     ##1100 00000#####       5                (l, 0, l, 0, 1)   
  l     ##11000 0000#####       6                (l, 0, l, 0, 1)   
  l     ##110000 000#####       7                (l, 0, l, 0, 1)   
  l     ##1100000 00#####       8                (l, 0, l, 0, 1)   
  l     ##11000000 0#####       9                (l, 0, l, 0, 1)   
  l     ##110000000 #####      10                (l, 0, l, 0, 1)   
  m     ##1100000000 ####      11                (l, #, m, 0, 1)   
  k     ##110000000 00###      10                (m, #, k, 0, -1)   
  k     ##11000000 000###       9                (k, 0, k, 0, -1)   
  k     ##1100000 0000###       8                (k, 0, k, 0, -1)   
  k     ##110000 00000###       7                (k, 0, k, 0, -1)   
  k     ##11000 000000###       6                (k, 0, k, 0, -1)   
  k     ##1100 0000000###       5                (k, 0, k, 0, -1)   
  k     ##110 00000000###       4                (k, 0, k, 0, -1)   
  k     ##11 000000000###       3                (k, 0, k, 0, -1)   
  k     ##1 1000000000###       2                (k, 0, k, 0, -1)   
  l     ##10 000000000###       3                (k, 1, l, 0, 1)   
  l     ##100 00000000###       4                (l, 0, l, 0, 1)   
  l     ##1000 0000000###       5                (l, 0, l, 0, 1)   
  l     ##10000 000000###       6                (l, 0, l, 0, 1)   
  l     ##100000 00000###       7                (l, 0, l, 0, 1)   
  l     ##1000000 0000###       8                (l, 0, l, 0, 1)   
  l     ##10000000 000###       9                (l, 0, l, 0, 1)   
  l     ##100000000 00###      10                (l, 0, l, 0, 1)   
  l     ##1000000000 0###      11                (l, 0, l, 0, 1)   
  l     ##10000000000 ###      12                (l, 0, l, 0, 1)   
  m     ##100000000000 ##      13                (l, #, m, 0, 1)   
  k     ##10000000000 00#      12                (m, #, k, 0, -1)   
  k     ##1000000000 000#      11                (k, 0, k, 0, -1)   
  k     ##100000000 0000#      10                (k, 0, k, 0, -1)   
  k     ##10000000 00000#       9                (k, 0, k, 0, -1)   
  k     ##1000000 000000#       8                (k, 0, k, 0, -1)   
 
\end{verbatim}
}

\newpage

{ \scriptsize
\begin{verbatim}
STATE   TAPE                      TAPE HEAD      INSTRUCTION              COMMENTS
  k     ##100000 0000000####      7              (k, 0, k, 0, -1)   
  k     ##10000 00000000####      6              (k, 0, k, 0, -1)   
  k     ##1000 000000000####      5              (k, 0, k, 0, -1)   
  k     ##100 0000000000####      4              (k, 0, k, 0, -1)   
  k     ##10 00000000000####      3              (k, 0, k, 0, -1)   
  k     ##1 000000000000####      2              (k, 0, k, 0, -1)   
  k     ## 1000000000000####      1              (k, 0, k, 0, -1)   
  l     ##0 000000000000####      2              (k, 1, l, 0, 1)   
  l     ##00 00000000000####      3              (l, 0, l, 0, 1)   
  l     ##000 0000000000####      4              (l, 0, l, 0, 1)   
  l     ##0000 000000000####      5              (l, 0, l, 0, 1)   
  l     ##00000 00000000####      6              (l, 0, l, 0, 1)   
  l     ##000000 0000000####      7              (l, 0, l, 0, 1)   
  l     ##0000000 000000####      8              (l, 0, l, 0, 1)   
  l     ##00000000 00000####      9              (l, 0, l, 0, 1)          
  l     ##000000000 0000####     10              (l, 0, l, 0, 1)   
  l     ##0000000000 000####     11              (l, 0, l, 0, 1)   
  l     ##00000000000 00####     12              (l, 0, l, 0, 1)   
  l     ##000000000000 0####     13              (l, 0, l, 0, 1)   
  l     ##0000000000000 ####     14              (l, 0, l, 0, 1)   
  m     ##00000000000000 ###     15              (l, #, m, 0, 1)   
  k     ##0000000000000 00##     14              (m, #, k, 0, -1)   
  k     ##000000000000 000##     13              (k, 0, k, 0, -1)   
  k     ##00000000000 0000##     12              (k, 0, k, 0, -1)   
  k     ##0000000000 00000##     11              (k, 0, k, 0, -1)   
  k     ##000000000 000000##     10              (k, 0, k, 0, -1)   
  k     ##00000000 0000000##      9              (k, 0, k, 0, -1)   
  k     ##0000000 00000000##      8              (k, 0, k, 0, -1)   
  k     ##000000 000000000##      7              (k, 0, k, 0, -1)   
  k     ##00000 0000000000##      6              (k, 0, k, 0, -1)   
  k     ##0000 00000000000##      5              (k, 0, k, 0, -1)   
  k     ##000 000000000000##      4              (k, 0, k, 0, -1)   
  k     ##00 0000000000000##      3              (k, 0, k, 0, -1)   
  k     ##0 00000000000000##      2              (k, 0, k, 0, -1)   
  k     ## 000000000000000##      1              (k, 0, k, 0, -1)   
  k     # #000000000000000##      0              (k, 0, k, 0, -1)   
  n     ## 000000000000000##      1              (k, #, n, #, 1)   
  n     ##1 00000000000000##      2              (n, 0, n, 1, 1)   
  n     ##11 0000000000000##      3              (n, 0, n, 1, 1)   
  n     ##111 000000000000##      4              (n, 0, n, 1, 1)   
  n     ##1111 00000000000##      5              (n, 0, n, 1, 1)   
  n     ##11111 0000000000##      6              (n, 0, n, 1, 1)   
  n     ##111111 000000000##      7              (n, 0, n, 1, 1)   
  n     ##1111111 00000000##      8              (n, 0, n, 1, 1)   
  n     ##11111111 0000000##      9              (n, 0, n, 1, 1)   
  n     ##111111111 000000##     10              (n, 0, n, 1, 1)   
  n     ##1111111111 00000##     11              (n, 0, n, 1, 1)   
  n     ##11111111111 0000##     12              (n, 0, n, 1, 1)   
  n     ##111111111111 000##     13              (n, 0, n, 1, 1)   
  n     ##1111111111111 00##     14              (n, 0, n, 1, 1)   
  n     ##11111111111111 0##     15              (n, 0, n, 1, 1)   
  n     ##111111111111111 ##     16              (n, 0, n, 1, 1)          Compute 16/2 
  f     ##11111111111111 10#     15              (n, #, f, 0, -1)   
  f     ##1111111111111 110#     14              (f, 1, f, 1, -1)   
  f     ##111111111111 1110#     13              (f, 1, f, 1, -1)   
  f     ##11111111111 11110#     12              (f, 1, f, 1, -1)   
  f     ##1111111111 111110#     11              (f, 1, f, 1, -1)   
  f     ##111111111 1111110#     10              (f, 1, f, 1, -1)   
  f     ##11111111 11111110#      9              (f, 1, f, 1, -1)   
  f     ##1111111 111111110#      8              (f, 1, f, 1, -1)   
  f     ##111111 1111111110#      7              (f, 1, f, 1, -1)   
  f     ##11111 11111111110#      6              (f, 1, f, 1, -1)   
\end{verbatim}
}

\newpage

{ \scriptsize
\begin{verbatim}
STATE   TAPE                   TAPE HEAD        INSTRUCTION  
  f     ##1111 111111111110#      5             (f, 1, f, 1, -1)   
  f     ##111 1111111111110#      4             (f, 1, f, 1, -1)   
  f     ##11 11111111111110#      3             (f, 1, f, 1, -1)   
  f     ##1 111111111111110#      2             (f, 1, f, 1, -1)   
  f     ## 1111111111111110#      1             (f, 1, f, 1, -1)   
  f     # #1111111111111110#      0             (f, 1, f, 1, -1)   
  g     ## 1111111111111110#      1             (f, #, g, #, 1)   
  i     ### 111111111111110#      2             (g, 1, i, #, 1)   
  i     ###1 11111111111110#      3             (i, 1, i, 1, 1)   
  i     ###11 1111111111110#      4             (i, 1, i, 1, 1)   
  i     ###111 111111111110#      5             (i, 1, i, 1, 1)   
  i     ###1111 11111111110#      6             (i, 1, i, 1, 1)   
  i     ###11111 1111111110#      7             (i, 1, i, 1, 1)   
  i     ###111111 111111110#      8             (i, 1, i, 1, 1)   
  i     ###1111111 11111110#      9             (i, 1, i, 1, 1)   
  i     ###11111111 1111110#     10             (i, 1, i, 1, 1)   
  i     ###111111111 111110#     11             (i, 1, i, 1, 1)   
  i     ###1111111111 11110#     12             (i, 1, i, 1, 1)   
  i     ###11111111111 1110#     13             (i, 1, i, 1, 1)   
  i     ###111111111111 110#     14             (i, 1, i, 1, 1)   
  i     ###1111111111111 10#     15             (i, 1, i, 1, 1)   
  i     ###11111111111111 0#     16             (i, 1, i, 1, 1)   
  e     ###1111111111111 10#     15             (i, 0, e, 0, -1)   
  f     ###111111111111 100#     14             (e, 1, f, 0, -1)   
  f     ###11111111111 1100#     13             (f, 1, f, 1, -1) 
  f     ###1111111111 11100#     12             (f, 1, f, 1, -1)   
  f     ###111111111 111100#     11             (f, 1, f, 1, -1)   
  f     ###11111111 1111100#     10             (f, 1, f, 1, -1)   
  f     ###1111111 11111100#      9             (f, 1, f, 1, -1)    
  f     ###111111 111111100#      8             (f, 1, f, 1, -1)   
  f     ###11111 1111111100#      7             (f, 1, f, 1, -1)   
  f     ###1111 11111111100#      6             (f, 1, f, 1, -1)   
  f     ###111 111111111100#      5             (f, 1, f, 1, -1)   
  f     ###11 1111111111100#      4             (f, 1, f, 1, -1)   
  f     ###1 11111111111100#      3             (f, 1, f, 1, -1)   
  f     ### 111111111111100#      2             (f, 1, f, 1, -1)   
  f     ## #111111111111100#      1             (f, 1, f, 1, -1)   
  g     ### 111111111111100#      2             (f, #, g, #, 1)   
  i     #### 11111111111100#      3             (g, 1, i, #, 1)   
  i     ####1 1111111111100#      4             (i, 1, i, 1, 1)   
  i     ####11 111111111100#      5             (i, 1, i, 1, 1)   
  i     ####111 11111111100#      6             (i, 1, i, 1, 1)   
  i     ####1111 1111111100#      7             (i, 1, i, 1, 1)   
  i     ####11111 111111100#      8             (i, 1, i, 1, 1)   
  i     ####111111 11111100#      9             (i, 1, i, 1, 1)   
  i     ####1111111 1111100#     10             (i, 1, i, 1, 1)   
  i     ####11111111 111100#     11             (i, 1, i, 1, 1)   
  i     ####111111111 11100#     12             (i, 1, i, 1, 1)   
  i     ####1111111111 1100#     13             (i, 1, i, 1, 1)   
  i     ####11111111111 100#     14             (i, 1, i, 1, 1)   
  i     ####111111111111 00#     15             (i, 1, i, 1, 1)   
  e     ####11111111111 100#     14             (i, 0, e, 0, -1)   
  f     ####1111111111 1000#     13             (e, 1, f, 0, -1)   
  f     ####111111111 11000#     12             (f, 1, f, 1, -1)   
  f     ####11111111 111000#     11             (f, 1, f, 1, -1)   
  f     ####1111111 1111000#     10             (f, 1, f, 1, -1)   
  f     ####111111 11111000#      9             (f, 1, f, 1, -1)   
  f     ####11111 111111000#      8             (f, 1, f, 1, -1)   
  f     ####1111 1111111000#      7             (f, 1, f, 1, -1)   
  f     ####111 11111111000#      6             (f, 1, f, 1, -1)   
  f     ####11 111111111000#      5             (f, 1, f, 1, -1)   
  f     ####1 1111111111000#      4             (f, 1, f, 1, -1)   
\end{verbatim}
}

\newpage

{ \scriptsize
\begin{verbatim}
STATE   TAPE                   TAPE HEAD         INSTRUCTION        
  f     #### 11111111111000#       3             (f, 1, f, 1, -1)   
  f     ### #11111111111000#       2             (f, 1, f, 1, -1)   
  g     #### 11111111111000#       3             (f, #, g, #, 1)   
  i     ##### 1111111111000#       4             (g, 1, i, #, 1)   
  i     #####1 111111111000#       5             (i, 1, i, 1, 1)   
  i     #####11 11111111000#       6             (i, 1, i, 1, 1)   
  i     #####111 1111111000#       7             (i, 1, i, 1, 1)   
  i     #####1111 111111000#       8             (i, 1, i, 1, 1)   
  i     #####11111 11111000#       9             (i, 1, i, 1, 1)   
  i     #####111111 1111000#      10             (i, 1, i, 1, 1)   
  i     #####1111111 111000#      11             (i, 1, i, 1, 1)   
  i     #####11111111 11000#      12             (i, 1, i, 1, 1)   
  i     #####111111111 1000#      13             (i, 1, i, 1, 1)   
  i     #####1111111111 000#      14             (i, 1, i, 1, 1)   
  e     #####111111111 1000#      13             (i, 0, e, 0, -1)   
  f     #####11111111 10000#      12             (e, 1, f, 0, -1)   
  f     #####1111111 110000#      11             (f, 1, f, 1, -1)   
  f     #####111111 1110000#      10             (f, 1, f, 1, -1)   
  f     #####11111 11110000#       9             (f, 1, f, 1, -1)   
  f     #####1111 111110000#       8             (f, 1, f, 1, -1)   
  f     #####111 1111110000#       7             (f, 1, f, 1, -1)   
  f     #####11 11111110000#       6             (f, 1, f, 1, -1)   
  f     #####1 111111110000#       5             (f, 1, f, 1, -1)   
  f     ##### 1111111110000#       4             (f, 1, f, 1, -1)  
  f     #### #1111111110000#       3             (f, 1, f, 1, -1)   
  g     ##### 1111111110000#       4             (f, #, g, #, 1)   
  i     ###### 111111110000#       5             (g, 1, i, #, 1)   
  i     ######1 11111110000#       6             (i, 1, i, 1, 1)   
  i     ######11 1111110000#       7             (i, 1, i, 1, 1)   
  i     ######111 111110000#       8             (i, 1, i, 1, 1)   
  i     ######1111 11110000#       9             (i, 1, i, 1, 1)   
  i     ######11111 1110000#      10             (i, 1, i, 1, 1)   
  i     ######111111 110000#      11             (i, 1, i, 1, 1)   
  i     ######1111111 10000#      12             (i, 1, i, 1, 1)   
  i     ######11111111 0000#      13             (i, 1, i, 1, 1)   
  e     ######1111111 10000#      12             (i, 0, e, 0, -1)   
  f     ######111111 100000#      11             (e, 1, f, 0, -1)   
  f     ######11111 1100000#      10             (f, 1, f, 1, -1)   
  f     ######1111 11100000#       9             (f, 1, f, 1, -1)   
  f     ######111 111100000#       8             (f, 1, f, 1, -1)   
  f     ######11 1111100000#       7             (f, 1, f, 1, -1)   
  f     ######1 11111100000#       6             (f, 1, f, 1, -1)   
  f     ###### 111111100000#       5             (f, 1, f, 1, -1)   
  f     ##### #111111100000#       4             (f, 1, f, 1, -1)   
  g     ###### 111111100000#       5             (f, #, g, #, 1)   
  i     ####### 11111100000#       6             (g, 1, i, #, 1)   
  i     #######1 1111100000#       7             (i, 1, i, 1, 1)   
  i     #######11 111100000#       8             (i, 1, i, 1, 1)   
  i     #######111 11100000#       9             (i, 1, i, 1, 1)   
  i     #######1111 1100000#      10             (i, 1, i, 1, 1)   
  i     #######11111 100000#      11             (i, 1, i, 1, 1)   
  i     #######111111 00000#      12             (i, 1, i, 1, 1)   
  e     #######11111 100000#      11             (i, 0, e, 0, -1)   
  f     #######1111 1000000#      10             (e, 1, f, 0, -1)   
  f     #######111 11000000#       9             (f, 1, f, 1, -1)   
  f     #######11 111000000#       8             (f, 1, f, 1, -1)   
  f     #######1 1111000000#       7             (f, 1, f, 1, -1)   
  f     ####### 11111000000#       6             (f, 1, f, 1, -1)   
  f     ###### #11111000000#       5             (f, 1, f, 1, -1)   
  g     ####### 11111000000#       6             (f, #, g, #, 1)   
  i     ######## 1111000000#       7             (g, 1, i, #, 1)   
  i     ########1 111000000#       8             (i, 1, i, 1, 1)   
\end{verbatim}
}

\newpage

{   \scriptsize
\begin{verbatim}
STATE   TAPE                   TAPE HEAD       INSTRUCTION                 COMMENTS
  i     ########11 11000000#       9           (i, 1, i, 1, 1)   
  i     ########111 1000000#      10           (i, 1, i, 1, 1)   
  i     ########1111 000000#      11           (i, 1, i, 1, 1)   
  e     ########111 1000000#      10           (i, 0, e, 0, -1)   
  f     ########11 10000000#       9           (e, 1, f, 0, -1)   
  f     ########1 110000000#       8           (f, 1, f, 1, -1)   
  f     ######## 1110000000#       7           (f, 1, f, 1, -1)   
  f     ####### #1110000000#       6           (f, 1, f, 1, -1)   
  g     ######## 1110000000#       7           (f, #, g, #, 1)   
  i     ######### 110000000#       8           (g, 1, i, #, 1)   
  i     #########1 10000000#       9           (i, 1, i, 1, 1)   
  i     #########11 0000000#      10           (i, 1, i, 1, 1)   
  e     #########1 10000000#       9           (i, 0, e, 0, -1)   
  f     ######### 100000000#       8           (e, 1, f, 0, -1)   
  f     ######## #100000000#       7           (f, 1, f, 1, -1)   
  g     ######### 100000000#       8           (f, #, g, #, 1)   
  i     ########## 00000000#       9           (g, 1, i, #, 1)   
  e     ######### #00000000#       8           (i, 0, e, 0, -1)   
  g     ########## 00000000#       9           (e, #, g, #, 1)   
  g     ##########1 0000000#      10           (g, 0, g, 1, 1)   
  g     ##########11 000000#      11           (g, 0, g, 1, 1)   
  g     ##########111 00000#      12           (g, 0, g, 1, 1)
  g     ##########1111 0000#      13           (g, 0, g, 1, 1)   
  g     ##########11111 000#      14           (g, 0, g, 1, 1)   
  g     ##########111111 00#      15           (g, 0, g, 1, 1)   
  g     ##########1111111 0#      16           (g, 0, g, 1, 1)   
  g     ##########11111111 #      17           (g, 0, g, 1, 1)   
  j     ##########1111111 1#      16           (g, #, j, #, -1)   
  j     ##########111111 11#      15           (j, 1, j, 1, -1)   
  j     ##########11111 111#      14           (j, 1, j, 1, -1)   
  j     ##########1111 1111#      13           (j, 1, j, 1, -1)   
  j     ##########111 11111#      12           (j, 1, j, 1, -1)   
  j     ##########11 111111#      11           (j, 1, j, 1, -1)   
  j     ##########1 1111111#      10           (j, 1, j, 1, -1)   
  j     ########## 11111111#       9           (j, 1, j, 1, -1)   
  j     ######### #11111111#       8           (j, 1, j, 1, -1)         Completed 16/2.
  a     ########## 11111111#       9           (j, #, a, #, 1)   
  b     ##########1 1111111#       10          (a, 1, b, 1, 1)            
  c     ##########11 111111#       11          (b, 1, c, 1, 1)   
  d     ##########111 11111#       12          (c, 1, d, 1, 1)   
  c     ##########1111 1111#       13          (d, 1, c, 1, 1)   
  d     ##########11111 111#       14          (c, 1, d, 1, 1)   
  c     ##########111111 11#       15          (d, 1, c, 1, 1)   
  d     ##########1111111 1#       16          (c, 1, d, 1, 1)   
  c     ##########11111111 #       17          (d, 1, c, 1, 1)   
  e     ##########1111111 1#       16          (c, #, e, #, -1)         Compute 8 / 2
  f     ##########111111 10#       15          (e, 1, f, 0, -1)   
  f     ##########11111 110#       14          (f, 1, f, 1, -1)   
  f     ##########1111 1110#       13          (f, 1, f, 1, -1)   
  f     ##########111 11110#       12          (f, 1, f, 1, -1)   
  f     ##########11 111110#       11          (f, 1, f, 1, -1)  
  f     ##########1 1111110#       10          (f, 1, f, 1, -1)   
  f     ########## 11111110#        9          (f, 1, f, 1, -1)   
  f     ######### #11111110#        8          (f, 1, f, 1, -1)   
  g     ########## 11111110#        9          (f, #, g, #, 1)   
  i     ########### 1111110#       10          (g, 1, i, #, 1)   
  i     ###########1 111110#       11          (i, 1, i, 1, 1)   
  i     ###########11 11110#       12          (i, 1, i, 1, 1)   
  i     ###########111 1110#       13          (i, 1, i, 1, 1)   
  i     ###########1111 110#       14          (i, 1, i, 1, 1)   
  i     ###########11111 10#       15          (i, 1, i, 1, 1)   
  i     ###########111111 0#       16          (i, 1, i, 1, 1)    
\end{verbatim}
}

\newpage

{    \scriptsize
     \begin{verbatim}
STATE   TAPE                    TAPE HEAD        INSTRUCTION               COMMENTS
  e     ###########11111 10#       15            (i, 0, e, 0, -1)   
  f     ###########1111 100#       14            (e, 1, f, 0, -1)   
  f     ###########111 1100#       13            (f, 1, f, 1, -1)   
  f     ###########11 11100#       12            (f, 1, f, 1, -1)   
  f     ###########1 111100#       11            (f, 1, f, 1, -1)   
  f     ########### 1111100#       10            (f, 1, f, 1, -1)   
  f     ########## #1111100#        9            (f, 1, f, 1, -1)   
  g     ########### 1111100#       10            (f, #, g, #, 1)   
  i     ############ 111100#       11            (g, 1, i, #, 1)   
  i     ############1 11100#       12            (i, 1, i, 1, 1)   
  i     ############11 1100#       13            (i, 1, i, 1, 1)   
  i     ############111 100#       14            (i, 1, i, 1, 1)   
  i     ############1111 00#       15            (i, 1, i, 1, 1)   
  e     ############111 100#       14            (i, 0, e, 0, -1)   
  f     ############11 1000#       13            (e, 1, f, 0, -1)   
  f     ############1 11000#       12            (f, 1, f, 1, -1)   
  f     ############ 111000#       11            (f, 1, f, 1, -1)   
  f     ########### #111000#       10            (f, 1, f, 1, -1)   
  g     ############ 111000#       11            (f, #, g, #, 1)   
  i     ############# 11000#       12            (g, 1, i, #, 1)   
  i     #############1 1000#       13            (i, 1, i, 1, 1)   
  i     #############11 000#       14            (i, 1, i, 1, 1)   
  e     #############1 1000#       13            (i, 0, e, 0, -1)   
  f     ############# 10000#       12            (e, 1, f, 0, -1)   
  f     ############ #10000#       11            (f, 1, f, 1, -1)   
  g     ############# 10000#       12            (f, #, g, #, 1)   
  i     ############## 0000#       13            (g, 1, i, #, 1)   
  e     ############# #0000#       12            (i, 0, e, 0, -1)   
  g     ############## 0000#       13            (e, #, g, #, 1)   
  g     ##############1 000#       14            (g, 0, g, 1, 1)   
  g     ##############11 00#       15            (g, 0, g, 1, 1)   
  g     ##############111 0#       16            (g, 0, g, 1, 1)   
  g     ##############1111 #       17            (g, 0, g, 1, 1)   
  j     ##############111 1#       16            (g, #, j, #, -1)   
  j     ##############11 11#       15            (j, 1, j, 1, -1)   
  j     ##############1 111#       14            (j, 1, j, 1, -1)   
  j     ############## 1111#       13            (j, 1, j, 1, -1)   
  j     ############# #1111#       12            (j, 1, j, 1, -1)         Completed 8/2
  a     ############## 1111#       13            (j, #, a, #, 1)   
  b     ##############1 111#       14            (a, 1, b, 1, 1)   
  c     ##############11 11#       15            (b, 1, c, 1, 1)   
  d     ##############111 1#       16            (c, 1, d, 1, 1)   
  c     ##############1111 #       17            (d, 1, c, 1, 1)   
  e     ##############111 1#       16            (c, #, e, #, -1)         Compute 4/2
  f     ##############11 10#       15            (e, 1, f, 0, -1)   
  f     ##############1 110#       14            (f, 1, f, 1, -1)   
  f     ############## 1110#       13            (f, 1, f, 1, -1)   
  f     ############# #1110#       12            (f, 1, f, 1, -1)   
  g     ############## 1110#       13            (f, #, g, #, 1)   
  i     ############### 110#       14            (g, 1, i, #, 1)   
  i     ###############1 10#       15            (i, 1, i, 1, 1)   
  i     ###############11 0#       16            (i, 1, i, 1, 1)   
  e     ###############1 10#       15            (i, 0, e, 0, -1) 
  f     ############### 100#       14            (e, 1, f, 0, -1)   
  f     ############## #100#       13            (f, 1, f, 1, -1)   
  g     ############### 100#       14            (f, #, g, #, 1)   
  i     ################ 00#       15            (g, 1, i, #, 1)   
  e     ############### #00#       14            (i, 0, e, 0, -1)   
  g     ################ 00#       15            (e, #, g, #, 1)   
  g     ################1 0#       16            (g, 0, g, 1, 1)   
  g     ################11 #       17            (g, 0, g, 1, 1)   
  \end{verbatim}
}

\newpage

{    \scriptsize
     \begin{verbatim}
STATE   TAPE                    TAPE HEAD    INSTRUCTION          COMMENTS   
  j     ################1 1#       16        (g, #, j, #, -1)    
  j     ################ 11#       15        (j, 1, j, 1, -1)   
  j     ############### #11#       14        (j, 1, j, 1, -1)     Completed 4/2
  a     ################ 11#       15        (j, #, a, #, 1)   
  b     ################1 1#       16        (a, 1, b, 1, 1)   
  c     ################11 #       17        (b, 1, c, 1, 1)   
  e     ################1 1#       16        (c, #, e, #, -1)     Compute 2/2
  f     ################ 10#       15        (e, 1, f, 0, -1)   
  f     ############### #10#       14        (f, 1, f, 1, -1)   
  g     ################ 10#       15        (f, #, g, #, 1)   
  i     ################# 0#       16        (g, 1, i, #, 1)   
  e     ################ #0#       15        (i, 0, e, 0, -1)   
  g     ################# 0#       16        (e, #, g, #, 1)   
  g     #################1 #       17        (g, 0, g, 1, 1)   
  j     ################# 1#       16        (g, #, j, #, -1)   
  j     ################ #1#       15        (j, 1, j, 1, -1)   
  a     ################# 1#       16        (j, #, a, #, 1)   
  b     #################1 #       17        (a, 1, b, 1, 1)   
  h     ################# 1#       16        (b, #, h, #, -1)     n = 5 orbit reaches 1
  \end{verbatim}
}

\end{document}